%% file: generalvim.tex
\documentclass[11pt]{article}
\usepackage[margin=0.9in]{geometry}
\usepackage{amsthm,amsmath,amssymb,amscd,enumerate, graphicx,subcaption,pstricks,bm,booktabs}
\usepackage{setspace,array}
\usepackage{thmtools}
\usepackage{thm-restate}
\usepackage[title]{appendix}
\usepackage{color}
\usepackage{natbib}
\usepackage{multirow}
\usepackage[mathscr]{euscript}
\usepackage[shortlabels]{enumitem}
\usepackage[colorlinks=true, urlcolor=blue, linkcolor=blue, citecolor=blue]{hyperref}
\usepackage{algpseudocode}
\usepackage{algorithm}
\algrenewcommand\algorithmicindent{0.4em}
\DeclareMathAlphabet{\mathbbs}{U}{bbm}{m}{sl}
\newcolumntype{C}[1]{>{\centering\let\newline\\\arraybackslash\hspace{0pt}}m{#1}}
\declaretheorem[name=Theorem]{thm}

\allowdisplaybreaks

\DeclareMathOperator*{\argmax}{argmax}

\usepackage{natbib}
\usepackage[T1]{fontenc}
\usepackage[utf8]{inputenc}
\usepackage{authblk}

\newcommand{\revision}[1]{{\leavevmode#1}}

\title{A \revision{general framework} for inference\\on algorithm-agnostic variable importance}
\author[1]{Brian D. Williamson}
\author[1,2]{Peter B. Gilbert}
\author[2]{Noah R. Simon}
\author[2,1]{Marco Carone}
\affil[1]{Vaccine and Infectious Disease Division, Fred Hutchinson Cancer Research Center}
\affil[2]{Department of Biostatistics, University of Washington}

\begin{document}

\maketitle

\begin{abstract}

In many applications, it is of interest to assess the relative contribution of features (or subsets of features) toward the goal of predicting a response --- in other words, to gauge the variable importance of features. Most recent work on variable importance assessment has focused on describing the importance of features within the confines of a given prediction algorithm. However, such assessment does not necessarily characterize the prediction potential of features, and may provide a misleading reflection of the intrinsic value of these features. To address this limitation, we propose a general framework for nonparametric  inference on interpretable algorithm-agnostic variable importance. We define variable importance as a population-level contrast between the oracle predictiveness of all available features versus all features except those under consideration. We propose a nonparametric efficient estimation procedure that allows the construction of valid confidence intervals, even when machine learning techniques are used. We also outline a valid strategy for testing the null importance hypothesis. Through simulations, we show that our proposal has good operating characteristics, and we illustrate its use with data from a study of an antibody against HIV-1 infection.\vspace{.1in}

\begin{center}{\small \textbf{Keywords:} variable importance; statistical inference; machine learning; targeted learning.}\end{center}
\end{abstract}

\doublespacing

\section{Introduction}\label{sec:intro}

In many scientific problems, it is of interest to assess the contribution of features toward the objective of predicting a response, a notion that has been referred to as variable importance. Various approaches for quantifying variable importance have been proposed in the literature. In recent applications, variable importance has often been taken to reflect the extent to which a given algorithm makes use of particular features in rendering predictions \citep{breiman2001,lundberg2017,fisher2018,murdoch2019}. In this case, the goal is thus to characterize a fixed algorithm. While this notion of variable importance can help provide greater transparency to otherwise opaque black-box prediction tools \citep{guidotti2018,murdoch2019}, it does not quantify the algorithm-agnostic  relevance of features for the sake of prediction. Thus, a feature that holds great value for prediction may be deemed unimportant simply because it plays a minimal role in the given algorithm. This motivates the consideration of approaches in which the focus is instead on measuring the population-level predictiveness potential of features, which we can refer to as \emph{intrinsic} variable importance. By definition, any measure of intrinsic variable importance should not involve the external specification of a particular prediction algorithm.

Traditionally, intrinsic variable importance has been considered in the context of simple population models (e.g., linear models) \citep[see, e.g.,][]{gromping2006,nathans2012}. For such models, both the prediction algorithm and the associated variable importance measure (VIM) are easy to compute from model outputs and straightforward to interpret. Common VIMs based on simple models include, for example, the difference in $R^2$ and deviance values based on (generalized) linear models \citep{nelder1972,gromping2006}. However, overly simplistic models can lead to misleading estimates of intrinsic variable importance with little population relevance. In an effort to improve prediction performance, complex prediction algorithms, including machine learning tools, have been used as a substitute for algorithms resulting from simple population models. Many variable importance measures have been proposed for specific algorithms (see, e.g., reviews of the literature in \citealp{wei2015}, \citealp{fisher2018}, and \citealp{murdoch2019}), with a particularly rich literature on variable importance for random forests \citep[see, e.g.,][]{breiman2001,strobl2007,ishwaran2007,gromping2009} and neural networks \citep[see, e.g.,][]{garson1991,bach2015,shrikumar2017,sundararajan2017}. Several recent proposals aim to describe a broad class of fixed algorithms \citep{ledell2015,ribeiro2016,benkeser2018,lundberg2017,aas2019}. However, while some measures have been recently described for algorithm-independent variable importance \citep[see, e.g.,][]{vanderlaan2006,lei2017,williamson2020a}, there has been limited work on developing broad frameworks for algorithm-independent variable importance with corresponding theory for inference using machine learning tools.

In this article, we seek to circumvent the limitations of model-based approaches to assessing intrinsic variable importance. We provide a unified nonparametric approach to formulate variable importance as a model-agnostic population parameter, that is, a summary of the true but unknown data-generating mechanism. The VIMs we consider are defined as a contrast between the predictiveness of the best possible prediction function based on all available features versus all features except those under consideration. We allow predictiveness to be defined arbitrarily as relevant and appropriate for the task at hand, as we illustrate in several examples. In this framework, once a measure of predictiveness has been selected, estimation of VIM values from data can be carried out similarly as for any other statistical parameter of interest. This task involves estimation of oracle prediction functions based on all the features or various subsets of features, and the use of machine learning algorithms is advantageous for maximizing prediction performance for this purpose. Because we consider variable importance as a summary of the data-generating mechanism rather than a property of any particular prediction algorithm, its definition and implementation does not hinge on the use of any particular prediction algorithm. This perspective contrasts with the model-based approach, where the probabilistic population-level mechanism that generates data and the algorithm that makes predictions based on data are usually entangled.

% In \citet{williamson2020a}, we focused on an application of the proposed framework to infer about a model-agnostic $R^2$-based variable importance, for which we described a nonparametric efficient estimator. We also presented the construction of valid confidence intervals and hypothesis tests for features with some importance but found it challenging to assess features with zero-importance.
In \citet{williamson2020a}, the authors focused on an application of the proposed framework to infer about a model-agnostic $R^2$-based variable importance, for which the authors described a nonparametric efficient estimator. The authors also presented the construction of valid confidence intervals and hypothesis tests for features with some importance but found it challenging to assess features with zero-importance.
Here, we propose a \revision{general framework} to study general predictiveness measures and propose a valid strategy for hypothesis testing. Our framework allows us to tackle cases involving complex predictiveness measures (e.g., defined in terms of counterfactual outcomes or involving missing data). It can be used to describe the importance of groups of variables as easily as individual variables. Our framework formally incorporates the use of machine learning tools to construct efficient estimators and perform valid statistical inference. We emphasize that the latter is especially important if high-impact decisions will be made on the basis of the resulting VIM estimates.

This article is organized as follows. In Section~\ref{sec:vimp}, we define variable importance as a contrast in population-level oracle predictiveness and provide simple examples. In Section~\ref{sec:est}, we construct an asymptotically efficient VIM estimator for a large class of measures using flexibly estimated prediction algorithms (e.g., predictive models constructed via machine learning methods) and provide a valid test of the zero-importance null hypothesis. These results allow us to analyze nonparametric extensions of common measures, including the area under the receiver operating characteristic curve (AUC) and classification accuracy. In Section~\ref{sec:complex}, we explore an extension to deal with more complex predictiveness measures. In Section~\ref{sec:sims}, we illustrate the use of the proposed approach in numerical experiments and detail its operating characteristics. Finally, we study the importance of various HIV-1 viral protein sequence features in predicting resistance to neutralization by an antibody in Section~\ref{sec:data}, and provide concluding remarks in Section~\ref{sec:conclusions}. All technical details as well as results from additional simulation studies and data analyses can be found in the Supplementary Material.

\section{Variable importance}\label{sec:vimp}

\subsection{Data structure and notation}\label{subsec:notation}

Suppose that observations $Z_1, \ldots, Z_n$ are drawn independently from a data-generating distribution $P_0$ known only to belong to a rich (nonparametric) class $\mathcal{M}$ of distributions. For concreteness, suppose that $Z_i=(X_i, Y_i)$, where $X_i = (X_{i1}, \ldots, X_{ip})\in\mathcal{X}\subseteq \mathbb{R}^p$ is a covariate vector and $Y_i\in\mathcal{Y}\subseteq\mathbb{R}$ is the outcome. Here, $\mathcal{X}$ and $\mathcal{Y}$ denote the sample spaces of $X$ and $Y$, respectively. Below, we will use the shorthand notation $E_0$ to refer to expectation under $P_0$.

We denote by $s\subseteq \{1, \dots, p\}$ the index set of the covariate subgroup of interest, and for any $p$-dimensional vector $w$, we refer to the elements of $w$ with index in $\ell$ and not in $\ell$ as $w_\ell$ and $w_{-\ell}$, respectively. We also denote by $\mathcal{X}_s$ and $\mathcal{X}_{-s}$ the sample space of $X_s$ and $X_{-s}$, respectively. Finally, we consider a rich class $\mathcal{F}$ of functions from $\mathcal{X}$ to $\mathcal{Y}$ endowed with a norm $\|\cdot\|_\mathcal{F}$, and define the subset $\mathcal{F}_{s} := \{f \in \mathcal{F} : f(u) = f(v)\text{ for all } u,v\in\mathcal{X}\text{ satisfying }u_{-s}=v_{-s}\}$ of functions in $\mathcal{F}$ whose evaluation ignores elements of the input $x$ with index in $s$. In all examples we consider, we will take $\mathcal{F}$ to be essentially unrestricted up to regularity conditions. Common choices include the class of all $P_0$-square-integrable functions from $\mathcal{X}$ to $\mathcal{Y}$ endowed with $L_2(P_0)$-norm $f\mapsto\|f\|_{2,P_0}:=[\int \{f(x)\}^2dP_0(x)]^{1/2}$, and of all bounded functions from $\mathcal{X}$ to $\mathcal{Y}$ endowed with the supremum norm $f\mapsto\|f\|_{\infty,\mathcal{X}}:=\sup_{x\in\mathcal{X}}|f(x)|$.

\subsection{Oracle predictiveness and variable importance}\label{subsec:compare_risks}

We now detail how we define variable importance as a population parameter. Suppose that $V(f, P)$ is a measure of the predictiveness of a given candidate prediction function $f \in \mathcal{F}$ when $P$ is the true data-generating distribution, with large values of $V(f, P)$ implying high predictiveness. Examples of predictiveness measures --- including those based on $R^2$, deviance, the area under the ROC curve, and classification accuracy --- are discussed in detail in Section~\ref{subsec:examples}. If the true data-generating mechanism $P_0$ were known, a natural candidate prediction function would be any $P_0$-population maximizer $f_0$ of predictiveness over the class $\mathcal{F}$:
\begin{align}\label{eq:maximizer}
f_{0} \in \argmax_{f \in \mathcal{F}} V(f, P_0)\ .
\end{align}This population maximizer can be viewed as the oracle prediction function within $\mathcal{F}$ under $P_0$ relative to $V$. In particular, the definition of $f_0$ depends on the chosen predictiveness measure and on the data-generating mechanism. It can also depend on the choice of  function class, although in contexts we consider this is not the case as long as $\mathcal{F}$ is sufficiently rich. It is often true that $f_0$ is the underlying target of machine learning-based prediction algorithms or a transformation thereof, which facilitates the integration of machine learning tools in the estimation of $f_0$. The \emph{oracle predictiveness} $V(f_0,P_0)$ provides a measure of total prediction potential under $P_0$. Similarly, defining the oracle prediction function $f_{0,s}$ that maximizes $V(f,P_0)$ over all $f\in\mathcal{F}_s$, the \emph{residual oracle predictiveness} $V(f_{0,s},P_0)$ quantifies the remaining prediction potential after exclusion of covariate features with index in $s$.

We define the \textit{population-level importance} of the variable (or subgroup of variables) $X_s$ relative to the full covariate vector $X$ as the amount of oracle predictiveness lost by excluding $X_s$ from $X$. In other words, we consider the VIM value defined as
\begin{align}\label{eq:vimp_functional}
\psi_{0,s}:=V(f_0,P_0)-V(f_{0,s},P_0)\ .
\end{align}
By construction, we note that $\psi_{0,s} \geq 0$. Whether or not the loss in oracle predictiveness is sufficiently large to confer meaningful importance to a given subgroup of covariates depends on context. Once more, we emphasize that the definition of $\psi_{0,s}$ involves the oracle prediction function within $\mathcal{F}$, and if $\mathcal{F}$ is large enough, this definition is agnostic to this choice.

\subsection{Examples of predictiveness measures}\label{subsec:examples}

We now illustrate our definition of variable importance by listing common VIMs that are in this framework. As we will see, the conditional mean  $\mu_0: x \mapsto E_0(Y \mid X = x)$ plays a prominent role in the examples below. This is convenient since $\mu_0$ is the implicit target of estimation for many standard machine learning algorithms for predictive modeling. \vspace{.15in}

\noindent\emph{Example 1: $R^2$}\\
The $R^2$ predictiveness measure is defined as $V(f, P_0) := 1 - E_0\left\{Y - f(X)\right\}^2/\sigma^2_0$, where we set $\sigma^2_0:=E_0\left\{Y-E_0\left(Y\right)\right\}^2=E_0\left[Y-E_0\left\{\mu_0(X)\right\}\right]^2$, the variance of $Y$ under $P_0$. This measure quantifies the proportion of variability in $Y$ explained by $f(X)$ under $P_0$. Since $\mu_0$ is the unrestricted minimizer of the mean squared error mapping $f\mapsto E_0\left\{Y-f(X)\right\}^2$, the optimizer of $V(f,P_0)$ is given by $f_0=\mu_0$ as long as $\mu_0\in\mathcal{F}$.\vspace{.15in}

\noindent\emph{Example 2: deviance}\\
When $Y$ is binary, the deviance predictiveness measure is defined as
  \begin{align*}
        V(f, P_0) = & \ 1 - \frac{E_0\left[Y\log f(X) + (1 - Y)\log \{1 - f(X)\} \right]}{\pi_0\log \pi_0 + (1-\pi_0)\log(1-\pi_0)}\ ,
    \end{align*}
    where $\pi_0:=P_0\left(Y=1\right)$ is the marginal success probability of $Y$ under $P_0$. This measure quantifies in a Kullback-Leibler sense the information gain from using $X$ to predict $Y$ relative to the null model that does not use $X$ at all. Again, because the conditional mean $\mu_0$ is the unconstrained population maximizer of the average log-likelihood, we find the optimizer of $f\mapsto V(f,P_0)$ to be $f_0=\mu_0$ for any rich enough $\mathcal{F}$. This result similarly holds for a multinomial extension of deviance.\vspace{.15in}

\noindent\emph{Example 3: classification accuracy}\\
An alternative predictiveness measure in the context of binary outcomes is classification accuracy, defined as $V(f, P_0) = P_0\left\{Y=f(X)\right\}$. This measure quantifies how often the prediction $f(X)$ coincides with $Y$, and is commonly used in classification problems. As shown in the Supplementary Material, the Bayes classifier $b_0:x \mapsto I\{\mu_0(x) > 1/2\}$ is the unconstrained maximizer of $f\mapsto V(f,P_0)$, and so, $f_0=b_0$ as long as $b_0\in\mathcal{F}$.\vspace{.15in}

\noindent\emph{Example 4: area under the ROC curve}\\
The area under the receiver operating characteristic curve (AUC) is another popular predictiveness measure for use when $Y$ is binary. The AUC corresponding to $f$ is given by $V(f, P_0) = P_0\{f(X_1) < f(X_2) \mid Y_1 = 0, Y_2 = 1\}$, where $(X_1, Y_1)$ and $(X_2, Y_2)$ represent independent draws from $P_0$. As shown in the Supplementary Material, the unrestricted maximizer of $f\mapsto V(f,P_0)$ is the population mean $\mu_0$, so  that once more $f_0=\mu_0$  provided $\mu_0\in\mathcal{F}$.\vspace{.15in}

In all examples above, the unrestricted oracle prediction function $f_0$ equals or is a simple transformation of the conditional mean function $\mu_0$.  The  unrestricted oracle prediction function $f_{0,s}$ based on all covariates except those with index in $s$ is obtained similarly but with $\mu_0$ replaced by $\mu_{0,s}:x\mapsto E_0\left(Y\mid X_{-s}=x_{-s}\right)$.

\section{Estimation and inference}\label{sec:est}

\subsection{Plug-in estimation}\label{sec:plugin}

In our framework, the variable importance of $X_s$ relative to $X$ under $P_0$, denoted $\psi_{0,s}$, is a population parameter. Thus, assessing variable importance reduces to the task of inferring about $\psi_{0,s}$ from the available data. More formally, our goal is to construct a nonparametric (asymptotically) efficient estimator of  $\psi_{0,s}$ using independent observations $Z_1,\ldots,Z_n$ from $P_0$. Definition \eqref{eq:vimp_functional} suggests considering the plug-in estimator
\begin{align}\label{eq:plugin_est}
{\psi}_{n,s} :=  V({f}_{n}, P_n) - V({f}_{n, s}, P_n)\ ,
\end{align}
where $P_n$ is the empirical distribution based on $Z_1, \dots, Z_n$,  and  ${f}_n$ and ${f}_{n, s}$ are estimators of the population optimizers $f_0$ and $f_{0,s}$, respectively. Often, ${f}_n$ and $f_{n,s}$ are obtained by building a predictive model for outcome $Y$ using all features in $X$ or only those features in $X_{-s}$, respectively --- this might be done, for example, using tree-based methods, deep learning, or other machine learning algorithms,  including tuning via cross-validation. Using flexible learning techniques to construct ${f}_n$ and ${f}_{n,s}$ minimizes the risk of systematic bias due to model misspecification.

As an illustration of the form of the resulting plug-in estimates, we note that, in the case of classification accuracy (\emph{Example 3}), the VIM estimate is given by $\psi_{n,s}=\tfrac{1}{n}\sum_{i=1}^nI\{Y_i = f_n(X_i)\} - \tfrac{1}{n}\sum_{i=1}^nI\{Y_i = f_{n,s}(X_i)\}$, where $f_n$ and $f_{n,s}$ are estimates of the oracle prediction functions $f_0$ and $f_{0,s}$, respectively. Sensible estimates of $f_0$ and $f_{0,s}$ are given by \[f_n:x\mapsto I\left\{\mu_n(x)>0.5\right\}\mbox{\ \ and\ \ }f_{n,s}:x\mapsto I\left\{\mu_{n,s}(x)>0.5\right\},\] where $\mu_{n}$ and $\mu_{n,s}$ are estimates of the conditional mean functions $\mu_0$ and $\mu_{0,s}$, respectively. We provide the explicit form of $\psi_{n,s}$ for all examples in the Supplementary Material.

The simplicity of the plug-in construction makes it particularly appealing.  However, the literature on semiparametric inference and targeted learning suggests that such naively constructed plug-in estimators may fail to even be consistent at rate $n^{-1/2}$, let alone efficient, if they involve nuisance functions --- in this case, ${f}_0$ and ${f}_{0,s}$ --- that are flexibly estimated. This phenomenon is due to the fact that excess bias is often inherited by the plug-in estimator from the nuisance estimators. Generally, this fact would motivate the use of debiasing procedures, such as the one-step correction or targeted maximum likelihood estimation \citep[see, e.g.,][]{pfanzagl1982,vanderlaan2011}.
% However, in \citet{williamson2020a}, we noted the intriguing fact that the  plug-in estimator of the $R^2$ VIM did not require debiasing, being itself already efficient.
However, \citet{williamson2020a} noted the intriguing fact that the  plug-in estimator of the $R^2$ VIM did not require debiasing, being itself already efficient.
Below, we show that the same  holds true for a large class of VIMs. These plug-in estimators therefore benefit from a combination of simplicity and statistical optimality.

\subsection{Large-sample properties}\label{sec:inf}

We now study conditions under which ${\psi}_{n,s}$ is an asymptotically linear and nonparametric efficient estimator of the VIM value $\psi_{0,s}$, and we describe how to conduct valid inference on $\psi_{0,s}$. Below, we explicitly focus on inference for the oracle predictiveness value $v_0:=V(f_0,P_0)$ based on the plug-in estimator $v_n := V({f}_n, P_n)$, since results can readily be extended to the residual oracle predictiveness value $v_{0,s}:=V(f_{0,s},P_0)$ and thus to the VIM value $\psi_{0,s}$. The behavior of $v_n$ can be studied by first decomposing
\begin{align}\label{eq:plug_in_expansion}
    v_n - v_0\ =\ \{V(f_0, P_n) - V(f_0, P_0)\} + \{V({f}_n, P_0) - V(f_0, P_0)\}+r_n\ ,
\end{align} where $r_n:=[\{V({f}_n, P_n) - V({f}_n, P_0)\} - \{V(f_0, P_n) - V(f_0, P_0)\}]$.
Each term on the right-hand side of \eqref{eq:plug_in_expansion} can be studied separately to determine the large-sample properties of $v_n$. The first term is the contribution from having had to estimate the second argument value $P_0$. The third term is a difference-of-differences remainder term that can be expected to tend to zero in probability at a rate faster than $n^{-1/2}$ under some conditions. We must pay particular attention to the second term, which represents the contribution from having had to estimate the first argument value $f_0$. A priori, we may expect this term to dominate since the rate at which ${f}_n-f_0$ tends to zero (in suitable norms) is generally slower than $n^{-1/2}$ when flexible learning techniques are used. However, because $f_0$ is a maximizer of $f\mapsto V(f,P_0)$ over $\mathcal{F}$, we may reasonably expect that \[\left.\frac{d}{d\epsilon}V(f_{0,\epsilon},P_0)\right|_{\epsilon=0}=0\] for any smooth path $\{f_{0,\epsilon}:-\infty<\epsilon<+\infty\}\subset\mathcal{F}$ through $f_0$ at $\epsilon=0$, and thus that there is no first-order contribution of $V({f}_n,P_0)-V(f_0,P_0)$ to the behavior of $v_n-v_0$. Under regularity conditions, this indeed turns out to be the case, and thus, if ${f}_n-f_0$ does not tend to zero too slowly, the second term will be asymptotically negligible.

Our first result will make use of several conditions requiring additional notation. Below, we define the linear space $\mathcal{R}:=\{c(P_1-P_2):c\in[0,\infty),P_1,P_2\in\mathcal{M}\}$ of finite signed measures generated by $\mathcal{M}$. For any $R\in\mathcal{R}$, say $R=c(P_1-P_2)$, we refer to the supremum norm $\|R\|_\infty:=c\cdot\sup_{z}|F_1(z)-F_2(z)|$, where $F_1$ and $F_2$ are the distribution functions corresponding to $P_1$ and $P_2$, respectively. Furthermore, we denote by $\dot{V}(f,P_0;h)$ the G\^{a}teaux derivative of $P\mapsto V(f,P)$ at $P_0$ in the direction $h\in\mathcal{R}$, and define the random function $g_n:z\mapsto\dot{V}(f_n,P_0;\delta_z-P_0)-\dot{V}(f_0,P_0;\delta_z-P_0)$, where $\delta_z$ is the degenerate distribution on $\{z\}$. \revision{For any $P \in \mathcal{M}$, we also denote by $f_P$ any $P$-population maximizer of $f\mapsto V(f, P)$ over $\mathcal{F}$.} Finally, we define the following sets of conditions, classified as being either deterministic (A) or stochastic (B) in nature:

\begin{enumerate}
     \item[(A1)] (\textit{optimality}) there exists some constant $C>0$ such that, for each sequence $f_1,f_2, \dots \in \mathcal{F}$ such that $\lVert f_j - f_0 \rVert_\mathcal{F} \to 0$, $\lvert V(f_j,P_0) - V(f_0,P_0) \rvert  \leq C\lVert f_j - f_0 \rVert_{\mathcal{F}}^2$ for each $j$ large enough;
     \item[(A2)] (\textit{differentiability}) there exists some constant $\delta > 0$ such that for each sequence $\epsilon_1,\epsilon_2,\ldots\in\mathbb{R}$ and $h, h_1, h_2, \ldots \in \mathcal{R}$ satisfying that $\epsilon_j\rightarrow 0$ and $\lVert h_j - h \rVert_{\infty} \to 0$, it holds that
     \begin{align*}
         \sup_{f\in\mathcal{F}:\lVert f - f_0 \|_{\mathcal{F}} < \delta}\left\lvert \frac{V(f, P_0 + \epsilon_j h_j) - V(f, P_0)}{\epsilon_j} - \dot{V}(f,P_0; h_j) \right\rvert \longrightarrow 0\ ;
     \end{align*}
     \item[(A3)] (\textit{continuity of optimization}) $\|f_{P_0+\epsilon h}-f_0\|_{\mathcal{F}}=O(\epsilon)$ for each $h\in\mathcal{R}$;
     \item[(A4)] (\textit{continuity of derivative}) $f\mapsto \dot{V}(f,P_0;h)$ is continuous at $f_0$ relative to $\|\cdot\|_{\mathcal{F}}$ for each $h\in\mathcal{R}$;
     \item[(B1)] (\textit{minimum rate of convergence}) $\|f_n-f_0\|_\mathcal{F}=o_P(n^{-1/4})$;
     \item[(B2)] (\textit{weak consistency}) \revision{$\int \{g_n(z)\}^2dP_0(z)=o_P(1)$};
     \item[(B3)] (\textit{limited complexity}) there exists some $P_0$-Donsker class $\mathcal{G}_0$ such that $P_0\left(g_n \in \mathcal{G}_0\right) \to 1$.
\end{enumerate}

\begin{thm}\label{thm:general_vim}
    If conditions (A1)--(A2) and (B1)--(B3) hold, then $v_n$ is an asymptotically linear estimator of $v_0$ with influence function equal to $\phi_{0}: z \mapsto \dot{V}(f_0,P_0; \delta_z - P_0)$, that is, \[v_n-v_0\, =\, \frac{1}{n}\sum_{i=1}^{n}\dot{V}(f_0,P_0;\delta_{Z_i}-P_0)+o_P(n^{-1/2})\]under sampling from $P_0$. If conditions (A3)--(A4) also hold, then $\phi_0$ coincides with the nonparametric  efficient influence function (EIF) of $P\mapsto V(f_P,P)$ at $P_0$, and so, $v_n$ is nonparametric efficient.
\end{thm} This result implies, in particular, that the plug-in estimator $v_n$ of $v_0$ is often consistent as well as asymptotically normal and efficient. A similar theorem applies to the study of the estimator $v_{n,s}:=V(f_{n,s},P_n)$ of residual oracle predictiveness $v_{0,s}$ upon replacing instances of $f_n$, $f_0$ and $\mathcal{F}$ by $f_{n,s}$, $f_{0,s}$ and $\mathcal{F}_s$ in the conditions above, and denoting the resulting influence function by $\phi_{0,s}$. Thus, under the collection of all such conditions, the estimator $\psi_{n,s}$ of the VIM value $\psi_{0,s}$ is asymptotically linear with influence function $\varphi_{0,s}:z\mapsto \dot{V}(f_0,P_0;\delta_z-P_0)-\dot{V}(f_{0,s},P_0;\delta_z-P_0)$ and nonparametric efficient. If $\psi_{0,s} > 0$ and $0<\tau^2_{0,s}:=E_0\{\varphi^2_{0,s}(Z)\}<\infty$, this suggests that the asymptotic variance of $n^{1/2}\left(\psi_{n,s}-\psi_{0,s}\right)$ can be estimated by
\[\tau^2_{n,s}:=\frac{1}{n}\sum_{i=1}^{n}\left[\dot{V}(f_n,P_n;\delta_{Z_i}-P_n)-\dot{V}(f_{n,s},P_n;\delta_{Z_i}-P_n)\right]^2,\]
and that $(\psi_{n,s}-z_{1-\alpha/2}\tau_{n,s}n^{-1/2},\psi_{n,s}+z_{1-\alpha/2}\tau_{n,s}n^{-1/2})$ is an interval for $\psi_{0,s}$ with asymptotic coverage $1-\alpha$, where  $z_{1-\alpha/2}$ denotes the $(1-\alpha/2)^{th}$ quantile of the standard normal distribution. \revision{This procedure is summarized in Algorithm \ref{alg:est}.} We discuss settings in which $\psi_{0,s} = 0$ (and therefore $\tau^2_{0,s}=0$) in Section~\ref{sec:test}.

\vspace{.1in}
\begin{algorithm}
\caption{\revision{Inference on VIM value $\psi_{0,s}$ (valid in non-null settings)}}
\label{alg:est}
\begin{spacing}{1.1}
\revision{
\begin{algorithmic}[1]
\vspace{.1in}
\State construct estimators $f_{n}$ of $f_0$ and $f_{n,s}$ of $f_{0,s}$;
\State construct empirical distribution estimator $P_{n}$ of $P_0$;
\State compute estimator ${\psi}_{n,s}:=V(f_{n}, P_{n})-V(f_{n,s}, P_{n})$ of $\psi_{0,s}$;
\State compute estimator \vspace{-.08in}\[\tau_{n,s}^2:=\frac{1}{n}\sum_{i=1}^n\{\dot{V}(f_{n}, P_{n}; \delta_{Z_i} - P_{n})-\dot{V}(f_{n,s}, P_{n}; \delta_{Z_i} - P_{n})\}^2\vspace{-.08in}\] of the asymptotic variance $\tau^2_{0,s}$ of $n^{1/2}\,(\psi_{n,s}-\psi_{0,s})$.\vspace{.1in}
\end{algorithmic}
}
\end{spacing}
\end{algorithm}

Condition (A1) ensures that there is no first-order contribution that results from estimation of $f_0$. As indicated above, this condition can generally be established as a consequence  of the optimality of $f_0$.  However, in each particular problem, appropriate regularity conditions on $P_0$ and $\mathcal{F}$ must be determined for this condition to hold. We have provided details for Examples 1--4 in the Supplementary Material, though we summarize our findings here. In Example 1, we have that $|V(f,P_0)-V(f_0,P_0)|=E_0\{f(X)-f_0(X)\}^2/\sigma^2_0$ as long as $\mu_0\in\mathcal{F}$, and so, condition (A1) holds with $C=1/\sigma^{2}_0$ and $\|\cdot\|_\mathcal{F}$ taken to be either the $L_2(P_0)$ or supremum norm. In Example 2, provided that all elements of $\mathcal{F}$ are bounded between $\gamma$ and $1-\gamma$ for some $\gamma\in(0,1)$, and that $\mu_0\in\mathcal{F}$, then condition (A1) holds with $C=\{\gamma \log\left(1-\gamma\right) \}^{-1}$ and $\|\cdot\|_\mathcal{F}$ taken to be either the $L_2(P_0)$ or supremum norm. In Example 3, condition (A1) holds for $C=4\kappa$ and $\|\cdot\|_{\mathcal{F}}$ the supremum norm provided the classification margin condition $P_0\left\{|\mu_0(X)-0.5|\leq t\right\}\leq \kappa t$ holds for some $0<\kappa<\infty$ and all $t$ small. Similarly, in Example 4, condition (A1) holds for $C=2\kappa/\{\pi_0(1-\pi_0)\}$ with $\pi_0:=P_0\left(Y=1\right)$ and $\|\cdot\|_{\mathcal{F}}$ the supremum norm provided the margin condition $P_0\left\{|\mu_0(X_1)-\mu_0(X_2)|\leq t\right\}\leq \kappa t$ holds for some $0<\kappa<\infty$ and all $t$ small, where $X_1$ and $X_2$ are independent draws from $P_0$.

Condition (A2) is a form of locally uniform Hadamard differentiability of $P \mapsto V(f_0, P)$ at $P_0$ in a neighborhood of $f_0$. It can be readily verified in Examples 1--4; in fact, in Examples 1--3, this condition holds for any $\delta>0$. Condition (A3) requires that the optimizer $f_P$ vary smoothly in $P$ around $P_0$, and is often straightforward to verify when $f_P$ has a closed analytic form. Condition (A4) instead requires that the Hadamard derivative of $P\mapsto V(f,P)$ at $P_0$ vary smoothly in $f$ around $f_0$. Condition (B1) requires that $f_0$ be estimated at a sufficiently fast rate in order for second-order terms to be asymptotically negligible, while condition (B2) states that a particular parameter-specific functional of $f_n$ must tend to the corresponding evaluation of $f_0$, and is thus implied by consistency of $f_n$ with respect to some norm under which this functional is continuous. Condition (B3) restricts the complexity of the algorithm used to generate $f_n$. We note that conditions (B1)--(B3) depend not only on the predictiveness measure chosen and on the true data-generating mechanism but also on properties of the estimator of the oracle prediction function.

\subsection{Implementation based on cross-fitting}

Condition (B3) puts constraints on the complexity of the algorithm used to generate $f_n$. This condition is prone to violations when flexible machine learning tools are employed, as discussed in \citet{zheng2011} and \citet{chernozhukov2018}, for example. However, it can be eliminated by dividing the entire dataset into two parts (say, training and test sets), estimating $f_0$ using the training data, and then evaluating the predictiveness measure on the test data. This readily extends to $K$-fold cross-fitting. To construct a cross-fitted estimator in the current context, we begin by randomly partitioning the dataset into $K$ subsets of roughly equal size. Setting aside one such subset, we construct an estimator $f_{k,n}$ of $f_0$ based on the bulk of the data, and then store $v_{k,n}:=V(f_{k,n},P_{k,n})$, where $P_{k,n}$ is the empirical distribution estimator based on the data set aside. We note that $f_{k,n}$ and $P_{k,n}$ are therefore estimated using non-overlapping subsets of the data. After repeating this operation for each of the $K$ subsets, we finally construct the cross-fitted estimator $v^*_{n}:=\frac{1}{K}\sum_{k=1}^{n}v_{k,n}$ of $v_0$.

To describe the large-sample behavior of $v_n^*$, we require an adaptation of the previously defined conditions (B1) and (B2) to the context of cross-fitted estimators. Below, the random function $g_{k,n}$ is defined identically as $g_n$ but with $f_n$ replaced by $f_{k,n}$.
\begin{enumerate}
     \item[(B1')] (\textit{minimum rate of convergence}) $\|f_{k,n}-f_0\|_\mathcal{F}=o_P(n^{-1/4})$ for each $k\in\{1,\ldots,K\}$;
     \item[(B2')] (\textit{weak consistency}) \revision{$\int \{g_{k,n}(z)\}^2dP_0(z)=o_P(1)$} for each $k\in\{1,\ldots,K\}$.
\end{enumerate}
The resulting cross-fit estimator $v_n^*$ enjoys desirable large-sample properties under weaker conditions than those imposed on $v_n$, as the theorem below states. In particular, condition (B3), which in practice limits the complexity of machine learning tools used to estimate $f_0$, is no longer required.

\begin{thm}\label{thm:cv}
If conditions (A1)--(A2) and (B1')--(B2') hold, then $v_n^*$ is an asymptotically linear estimator of $v_0$ with influence function  equal to $\phi_0:x\mapsto \dot{V}(f_0,P_0;\delta_z-P_0)$, that is, \[v_n^*-v_0\, =\, \frac{1}{n}\sum_{i=1}^{n}\dot{V}(f_0,P_0;\delta_{Z_i}-P_0)+o_P(n^{-1/2})\]under sampling from $P_0$.  If conditions (A3)--(A4) also hold, then $v^*_n$ is nonparametric efficient.
\end{thm}

The cross-fitted construction can be used to obtain an improved estimator $v^*_{n,s}$ of $v_{0,s}$ as well, thereby resulting in a cross-fitted estimator $\psi^*_{n,s}:=v^*_{n}-v^*_{n,s}$ of the VIM value $\psi_{0,s}$. \revision{Cross-fitting can also be used to obtain an improved estimator $\tau_{n,s,*}^2$ of the asymptotic variance $\tau^2_{0,s}$. We summarize this construction in Algorithm~\ref{alg:cv_est}}, and provide the explicit form of $\psi^*_{n,s}$ for Examples 1--4 in the Supplementary Material. As before, Theorem \ref{thm:cv} readily provides conditions under which $v_{n,s}^*$ is an asymptotically linear and nonparametric efficient estimator of $v_{0,s}$, and so, under which $\psi_{n,s}^*$ is an asymptotically linear and nonparametric efficient estimator of the VIM value $\psi_{0,s}$. Based on these theoretical results as well as numerical experiments, we recommend this implementation whenever machine learning tools are used to estimate $f_0$ and $f_{0,s}$.
\vspace{.1in}
\begin{algorithm}
\caption{\revision{Cross-fitted inference on VIM value $\psi_{0,s}$ (valid in non-null settings)}}
\label{alg:cv_est}
\begin{spacing}{1.1}
\revision{
\begin{algorithmic}[1]
\vspace{.1in}
\State generate $B_n \in \{1, \ldots, K\}^n$ by sampling uniformly from $\{1, \ldots, K\}$ with replacement, and for $j=1,\ldots,K$, denote by $D_j$ the subset of observations with index in $S_j:= \{i: B_{n,i}=j\}$;
\For{$k = 1, \ldots, K$}
\State using only data in $\cup_{j\neq k}D_j$, construct estimators $f_{k,n}$ of $f_0$ and $f_{k,n,s}$ of $f_{0,s}$;
\State using only data in $D_k$, construct empirical distribution estimator $P_{k,n}$ of $P_0$;
\State with $n_k := \sum_{i=1}^n I\{i \in S_k\}$, compute ${\psi}_{k,n,s}:=V(f_{k,n}, P_{k,n})-V(f_{k,n,s}, P_{k,n})$ and \vspace{-.08in}\[\tau^2_{k,n,s}:=\frac{1}{n_k}\sum_{i\in S_k}\{\dot{V}(f_{k,n}, P_{k,n}; \delta_{Z_i} - P_{k,n})-\dot{V}(f_{k,n,s}, P_{k,n}; \delta_{Z_i} - P_{k,n})\}^2;\vspace{-.1in}\]
\EndFor
\State compute estimator ${\psi}_{n,s}^*:= \frac{1}{K}\sum_{k=1}^K {\psi}_{k,n,s}$ of $\psi_{0,s}$;
\State compute estimator $\tau_{n,s,*}^2:= \frac{1}{K}\sum_{k=1}^K \tau_{k,n,s}^2$ of the asymptotic variance $\tau^2_{0,s}$ of $n^{1/2}\,(\psi^*_{n,s}-\psi_{0,s})$.
\vspace{.1in}
\end{algorithmic}
}
\end{spacing}
\end{algorithm}

\subsection{Inference under the zero-importance null hypothesis}\label{sec:test}

When $\psi_{0,s}=0$, in which case the variable group considered has null importance, the influence function of $\psi_{n,s}$ is identically zero. In these cases, even after standardization, $\psi_{n,s}$ generally does not tend to a non-degenerate law. As such, deriving an implementable test of the null hypothesis $\psi_{0,s}=0$ or a confidence interval valid even when $\psi_{0,s}=0$ is difficult. In such cases, standard Wald-type confidence intervals and tests based on $\tau_{n,s}^2$ will typically have incorrect coverage or type I error, as illustrated in numerical simulations reported in \cite{williamson2020a}. While in parametric settings $n$-rate inference is possible under this type of degeneracy, this is not expected to be the case in nonparametric models, because the second-order contribution from estimation of $f_0$ and $f_{0,s}$ will generally have a rate slower than $n^{-1}$.

We note that, although $\psi_{n,s}$ has degenerate behavior under the null, each of $v_n$ and $v_{n,s}$ are asymptotically linear with non-degenerate (but possibly identical) influence functions. Except for extreme cases in which the entire set of covariates has null predictiveness, we may leverage this fact to circumvent null degeneracy via sample-splitting. Indeed, if $v_n$ and $v_{n,s}$ are constructed using different subsets of the data, then the resulting estimator $\psi_{n,s}$ is asymptotically linear with a non-degenerate influence function even if $\psi_{0,s}=0$, so that a valid Wald test of the strict null $H_0:\psi_{0,s}=0$ versus $H_1:\psi_{0,s}>0$ can be constructed using $\psi_{n,s}$ and an estimator of the standard error of $\psi_{n,s}$. Of course, the same holds for the corresponding cross-fitted procedures, as we consider below. We emphasize here that sample-splitting and cross-fitting are distinct operations with distinct goals. Sample-splitting is used to ensure valid inference under the zero-importance null hypothesis, whereas cross-fitting is used to eliminate the need for Donsker class conditions, which otherwise limit how flexible the learning strategies for estimating the oracle prediction functions can be. Sample-splitting and cross-fitting can be used simultaneously --- Figure~\ref{fig:samplesplit} provides an illustration of the subdivision of a dataset when equal subsets are used for sample-splitting and six splits are used for cross-fitting.

In practice, a group of variables may be considered scientifically unimportant even when $\psi_{0,s}$ is nonzero but small, yet such grouping would be deemed statistically significant in large enough samples. For this reason, given a threshold $\beta>0$, it may be more scientifically appropriate to consider testing the $\beta$-null $H_0: \psi_{0,s} \in[0, \beta]$ versus its complement alternative $H_1: \psi_{0,s} > \beta$. The $\beta$-null approaches the strict null as $\beta$ decreases to 0. The idea of sample-splitting also allows us to tackle $\beta$-null testing. Suppose that mutually exclusive portions of the dataset, say of respective sizes $n-n_s$ and $n_s$, are used to construct $v_n^*$ and $v_{n,s}^*$. Suppose further that $\eta^2_n$ and $\eta^2_{n,s}$ are consistent estimators of $\eta_0^2:=E_0\{\phi_0(Z)\}^2$ and $\eta_{0,s}^2:=E_0\{\phi_{0,s}(Z)\}^2$, respectively. Then, provided $v_0>0$, we may consider rejecting the $\beta$-null hypothesis $H_0$ in favor of its complement $H_1$ if and only if
\begin{equation}\label{eq:null_test}
t_n := \omega_{n,s}^{-1/2}\left(v_n^* - v_{n,s}^* - \beta\right)>z_{1-\alpha}\ ,
\end{equation}where $\omega_{n,s}:=\eta_n^2/(n-n_s)+\eta^2_{n,s}/n_s$ and $z_{1-\alpha}$ is the $(1-\alpha)^\text{th}$ quantile of the standard normal distribution. The implementation of the resulting test, including computation of the corresponding $p$-value, is summarized in Algorithm~\ref{alg:sscf}. Its validity is guaranteed under conditions of Theorem 2 directly applied on the split used to estimate $v_0$ and modified appropriately (by replacing instances of $f_n$, $f_0$ and $\mathcal{F}$ by $f_{n,s}$, $f_{0,s}$ and $\mathcal{F}_s$ in all conditions) for the split used to estimate $v_{0,s}$. We note that, although there is no degeneracy under the $\beta$-null whenever $\psi_{0,s}\in(0,\beta)$, sample-splitting is still required for proper type I error control since the strict null $\psi_{0,s}=0$ is contained in the $\beta$-null and must therefore be guarded against. We emphasize here that the use of distinct subsets of the data is critical for constructing $v_n^*$ and $v_{n,s}^*$. If instead $\psi_{0,s} = 0$ and $v_n^*$ and $v_{n,s}^*$ were constructed using the same data, the behavior of any testing procedure based on an estimator $\kappa_{n,s}$ of the standard error of $\psi_{n,s}$ would depend on the relative rates of convergence to zero of both $\psi_{n,s}$ and $\kappa_{n,s}$. In particular, this would lead to either uncontrolled type I error or type I error tending to zero depending on the procedures used to obtain $f_n$ and $f_{n,s}$. \revision{Inference based on a cross-fitted version of this sample-split procedure is described in Algorithm~\ref{alg:sscf}.}

\begin{figure}
\centering
\includegraphics[width = 0.9\textwidth]{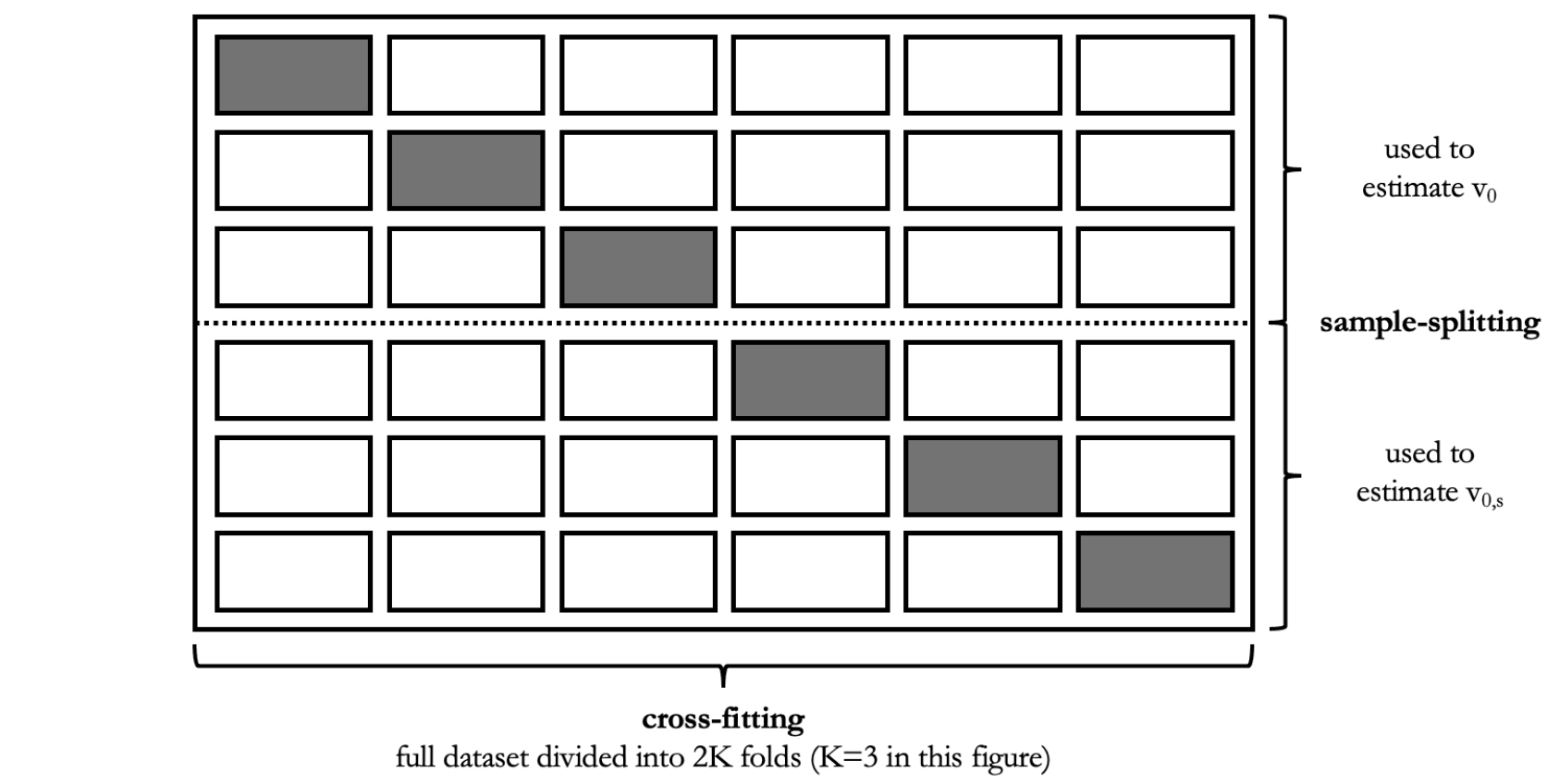}
\caption{\revision{Illustration of dataset subdivision when sample-splitting and cross-fitting are used simultaneously for valid inference under the zero-importance hypothesis (sample-splitting) without requiring Donsker class conditions (cross-fitting). Each row represents the entire dataset with a different subset singled out (in grey) as testing set. To estimate $v_0$, the top three rows are used. In each such row, $f_0$ is estimated using data in the white cells, and $v_0$ is estimated using the resulting estimate of $f_0$ and data in the grey cells. Row-specific estimates of $v_0$ are then averaged. The process is repeated for estimating $v_{0,s}$ but instead using the bottom three rows and estimating $f_{0,s}$ rather than $f_0$.}}
\label{fig:samplesplit}
\end{figure}

\vspace{.1in}

\begin{algorithm}
\caption{\revision{Sample-split, cross-fitted inference on VIM value $\psi_{0,s}$}}
\label{alg:sscf}
\begin{spacing}{1.1}
\revision{
\begin{algorithmic}[1]
\vspace{.1in}
\State generate $B_n \in \{1, \ldots, 2K\}^2$ by sampling uniformly from $\{1,2\}$ with replacement, and for $j = 1, \ldots, 2K$, denote by $D_j$ the set of observations with index in $S_j := \{i: B_{n,i} = j\}$ and $n_j:=|D_j|$;
\For{$k = 1, \ldots, 2K$}
    \State using only data in $\cup_{j \neq k}D_j$, construct estimators $f_{k,n}$ of $f_0$ and $f_{k,n,s}$ of $f_{0,s}$;
    \State using only data in $D_k$, construct estimator $P_{n,k}$ of $P_0$;
    \State if $k$ is odd, compute $\eta_{k,n}^2:=\frac{1}{n_k}\sum_{i\in S_k}\dot{V}(f_{k,n}, P_{k,n}; \delta_{Z_i} - P_{k,n})^2$ and $v_{k,n}:=V(f_{k,n},P_{k,n})$;
    \State if $k$ is even, compute $\eta_{k,n,s}^2:=\frac{1}{n_k}\sum_{i\in S_k}\dot{V}(f_{k,n,s}, P_{k,n}; \delta_{Z_i} - P_{k,n})^2$ and $v_{k,n,s}:=V(f_{k,n,s},P_{k,n})$;
\EndFor
\State compute $v_n^* := \frac{1}{K}\sum_{k = 1}^K v_{2k-1,n}$, $v_{n,s}^* := \frac{1}{K}\sum_{k = 1 }^K v_{2k,n,s}$ and estimator $\psi^*_{n,s}:=v_n^*-v_{n,s}^*$ of $\psi_{0,s}$;
\State compute $\eta_n^2 := \frac{1}{K}\sum_{k=1}^{K}\eta_{2k-1,n}^2$, $\eta_{n,s}^2 := \frac{1}{K}\sum_{k=1}^{K}\eta_{2k,n,s}^2$ and estimator $\omega_{n,s}:=\eta^2_n/(n-n_s)+\eta_{n,s}^2/n_s$ of the variance of $\psi^*_{n,s}$;
\State to test $H_0: \psi_{0,s} \in [0,\beta]$ vs $H_1:\psi_{0,s}>\beta$ at level $1-\alpha$, reject $H_0$ in favor of $H_1$ iff $p_n:=1-\Phi(t_n) < \alpha$ with $t_n:=\omega^{-1/2}_{n,s}(\psi_{n,s}^*-\beta)$ and $\Phi$ the standard normal distribution function.
\vspace{.1in}
\end{algorithmic}
}
\end{spacing}
\end{algorithm}

The above testing procedure can be readily inverted to yield a one-sided confidence interval for $\psi_{0,s}$. Specifically, under regularity conditions and provided $v_0>0$, the random interval $(v_n^*-v^*_{n,s}-z_{1-\alpha}\omega_{n,s}^{1/2},+\infty)$ contains $\psi_{0,s}$ with probability no less than $1-\alpha$ asymptotically, even when $\psi_{0,s}=0$. Then, rejecting the null hypothesis $H_0$ is equivalent to verifying that zero is contained in this one-sided interval. A two-sided confidence interval is instead given by $(v_n^*-v_{n,s}^*-z_{1-\alpha/2}\omega_{n,s}^{1/2},v_n^*-v_{n,s}^*+z_{1-\alpha/2}\omega_{n,s}^{1/2})$. While the latter interval has the advantage of giving both a lower and upper bound on possible values for $\psi_{0,s}$ supported by the data, using it for testing purposes  necessarily results in a reduction in power since the null value of $\psi_{0,s}$ is at the edge of the parameter space.

\section{Extensions to more complex settings}\label{sec:complex}

In all examples studied thus far, the primary role $P$ plays in $V(f,P)$ is to indicate the  population with respect to which a particular measure of prediction performance should be averaged. In these cases, $P\mapsto V(f,P)$ is well-defined on discrete probability measures and sufficiently smooth so that $V(f_0,P_n)-V(f_0,P_0)$ is in first order a linear estimator in view of the functional delta method. However, there are other examples in which this requirement may not be true. In these examples, $V(f,P)$ involves $P$ in a complex manner beyond some form of averaging, rendering $V(f,P)$ undefined for discrete $P$, let alone Hadamard differentiable. Complex predictiveness measures often arise when the sampling mechanism precludes from observation the ideal data unit on which a (possibly simpler) predictiveness measure is defined, and identification formulas must therefore be established to express predictiveness in terms of the observed data-generating distribution.

As a concrete illustration, we begin with an example from the causal inference literature.
As before, we denote by $Y$ and $X$ the outcome of interest and a covariate vector, respectively. We suppose that larger values of $Y$ correspond to better clinical outcomes, and consider a binary intervention $A\in\{0,1\}$. A given treatment rule $f:\mathcal{X}\rightarrow\{0,1\}$ for assigning the value of $A$ based on $X$ can be adjudicated, for example, on the basis of the population mean outcome that would arise if everyone in the population were treated according to $f$. We can consider the ideal data structure to be $\mathbbs{Z}:=(X,A,Y(0),Y(1))\sim \mathbbs{P}_0$, where for each $a\in\{0,1\}$,  $Y(a)$ denotes the counterfactual outcome corresponding to the intervention that deterministically sets $A=a$. The ideal-data predictiveness of $f$ is then $\mathbbs{V}(f,\mathbbs{P}_0):=E_{\mathbbs{P}_0}\left\{Y(f(X))\right\}$. In contrast, the observed data structure is $Z:=(X,A,Y)\sim P_0$, and we must find some observed-data predictiveness measure $V$ such that $V(f,P_0)=\mathbbs{V}(f,\mathbbs{P}_0)$ to establish identification and proceed with estimation and inference. Defining the outcome regression $Q_P(a,x):=E_P\left(Y\mid A=a,X=x \right)$, it is not difficult to verify that \[V(f,P):=E_P\left[Q_P(f(X),X)\right]\] provides a valid identification of $\mathbbs{V}(f,\mathbbs{P})$ under standard causal identification conditions. We note that this predictiveness measure involves $P$ through more than simple averaging, as the outcome regression $Q_P$ also appears in the definition of $V(f,P)$. Unless the distribution of $X$ is discrete under $P$, $Q_P$ is ill-defined on the empirical distribution $P_n$, thus \revision{violating conditions (A1) and (A4) defined in Section~\ref{sec:est}}. We also remark that, in this example, the moniker `prediction function' is not entirely fitting for $f$, which represents a treatment rule and maps into the treatment (rather than outcome) space. Nevertheless, the proposed framework for variable importance remains applicable, underscoring the fact that it is sufficiently flexible to unify a large swath of variable importance problems. Restrictions imposed on the data structure and on the properties of the prediction function in Section~\ref{sec:vimp} were largely for the sake of concreteness.

The simple plug-in approach described in Section~\ref{sec:est} may fail in applications with more complex predictiveness measures. In such cases, we can instead employ a more general strategy based on nonparametric debiasing techniques to make valid inference about $V(f_0,P_0)$. For each $P\in\mathcal{M}$, we denote by $f_P$ any optimizer of $f\mapsto V(f,P)$ over $\mathcal{F}$, and define the parameter mapping $V^*:P\mapsto V(f_P,P)$ so that $v_0$ can be expressed as $V^*(P_0)$. If $\widehat{P}_n\in\mathcal{M}$ is an estimator of $P_0$, the plug-in estimator $V^*(\widehat{P}_n)$ generally fails to be asymptotically linear unless $\widehat{P}_n$ was purposefully constructed to ensure that it is indeed so. This happens because the plug-in estimator $V^*(\widehat{P}_n)$ generally suffers from excessive bias whenever flexible learning techniques have been used, for example, because $V^*(P_0)$ involves local features of $P_0$ (e.g., the conditional mean or density function) --- see \citet{pfanzagl1982} and \citet{vanderlaan2011}. This fact renders the use of debiasing approaches necessary. In contrast, the one-step  estimator \[v_{n,OS}:=V^*(\widehat{P}_n)+\frac{1}{n}\sum_{i=1}^{n}\phi_n(Z_i)\ ,\] where $\phi_n$ is the nonparametric EIF of $V^*$ at $\widehat{P}_n$, is nonparametric efficient \citep{pfanzagl1982} under regularity conditions. Alternatively, the framework of  targeted minimum loss-based estimation describes how to convert $\widehat{P}_n$ into a revised estimator $\widehat{P}_n^*$ such that $V^*(\widehat{P}^*_n)$ is itself nonparametric efficient without the need for further debiasing \citep{vanderlaan2011}. Similarly as in Section~\ref{sec:est}, cross-fit versions of these debiasing procedures (see, e.g., \citealp{zheng2011,chernozhukov2018}) can be used to improve performance when flexible estimation algorithms are used.

The generic approach above relies on deriving the nonparametric EIF of $V^*$. The definition of $V^*$ involves $P$ in various ways, including through the $P$-optimal prediction function $f_P$. While in our examples $f_P$ has a simple closed-form expression, this may not always be so. This fact can greatly complicate the derivation of the required EIF. However, as we shall see, the optimality of $f_P$ often implies that  $f_P$ does not contribute to the nonparametric EIF of $P\mapsto V(f_P,P)$.

Before stating a formal result to this effect, we introduce a regularity condition. Below, $L_2^0(P_0)$ refers to the subset of all functions in $L_2(P_0)$ that have mean zero under $P_0$.

\begin{enumerate}
\item[(A5)] There exists a dense subset $\mathcal{H}$ of $L_2^0(P_0)$ such that, for each $h\in\mathcal{H}$ and regular univariate parametric submodel $\{P_{0,\epsilon}\}\subset\mathcal{M}$ through $P_0$ at $\epsilon=0$ and with score for $\epsilon$ equal to $h$ at $\epsilon=0$ (see, e.g., \citealp{bkrw}), the following conditions hold, with $f_{0,\epsilon}$ denoting $f_{P_{0,\epsilon}}$: \begin{itemize}
\item[(A5a)] (\emph{second-order property of predictiveness perturbations})\\$V(f_{0,\epsilon},P_{0,\epsilon})-V(f_{0,\epsilon},P_0)=V(f_{0},P_{0,\epsilon})-V(f_0,P_0)+o(\epsilon)$ holds;
\item[(A5b)] (\emph{differentiability})  the mapping $\epsilon\mapsto V(f_{0,\epsilon},P_0)$ is differentiable in a neighborhood of $\epsilon=0$;
\item[(A5c)] (\emph{richness of function class}) the optimizer $f_{0,\epsilon}$ is in $\mathcal{F}$ for small enough $\epsilon$.
\end{itemize}
\end{enumerate}

Condition (A5a) essentially requires the pathwise derivative of $P\mapsto V(f,P)$ at $P_0$ to be insensitive to infinitesimal perturbations of $f$ around $f_0$. In such case, the difference-in-differences term appearing in the condition can indeed be expected to be second-order in $\epsilon$. Condition (A5b) will generally hold provided the functionals $f\mapsto V(f,P_0)$ and $P\mapsto f_P$ are sufficiently smooth around $f_0$ and $P_0$, respectively. Finally, condition (A5c) requires that $\mathcal{F}$ be sufficiently rich around $f_0$ so that, for a dense collection of paths through $P_0$, $\mathcal{F}$ contains $f_{0,\epsilon}$ for small enough $\epsilon$ .

\begin{thm}\label{thm:same_eif}
Provided condition (A5) holds, if $P\mapsto V(f_0,P)$ is pathwise differentiable at $P_0$ relative to the nonparametric model $\mathcal{M}$, then so is $P\mapsto V(f_P,P)$, and the two parameters have the same EIF.
\end{thm}

This theorem indicates that, under a regularity condition, the computation of the nonparametric EIF $\phi_0$ can be done treating $f_P$ as fixed at $f_0$, thereby simplifying considerably this calculation. This fact is also useful because for a fixed prediction function $f$ the parameter $P\mapsto V(f,P)$  will often have already been studied in the literature, thereby circumventing the need for any novel derivation. Armed with this observation, we revisit the motivating example we presented in this section, and also consider an additional example involving missing data.\vspace{.15in}

\noindent\emph{Example 5: mean outcome under a binary intervention rule}\\
\noindent As described above, in this example, the ideal-data parameter of interest, $\mathbbs{V}(f,\mathbbs{P}):=E_{\mathbbs{P}}\left\{Y(f(X))\right\}$, can be identified by the observed-data parameter $V(f,P)=E_{P}\left\{Q_P(f(X),X)\right\}$ when the observed data unit consists of $Z=(X,A,Y)\sim P$. The map $f\mapsto V(f,P_0)$ is maximized over the unrestricted class $\mathcal{F}$ by the intervention rule  $f_0:x\mapsto I\{Q_0(1,x)>Q_0(0,x)\}$, and over its subset $\mathcal{F}_s$ by $f_{0,s}:x\mapsto I\{Q_{0,s}(1,x)>Q_{0,s}(0,x)\}$, where we define $Q_{0,s}$ pointwise as $Q_{0,s}(a,x):=E_0\left\{Q_0(a,X)\mid X_{-s}=x_{-s}\right\}$. Furthermore, the parameter $P\mapsto V(f_0,P)$ is pathwise differentiable at a distribution $P_0$ if, for example, $Q_0(1,W)-Q_0(0,W)\neq 0$ occurs $P_0$-almost surely. The nonparametric EIF of $P\mapsto V(f_0,P)$ at $P_0$ is given by \[\phi_0:z\mapsto \frac{I\{a=f_0(x)\}}{g_0(f_0(x),x)}\left\{y-Q_0(f_0(x),x)\right\}+Q_0(f_0(x),x)-V(f_0,P_0)\ ,\] where we define the propensity score $g_0(a,x):=P_0\left(A=a\mid X=x\right)$ for each $a\in\{0,1\}$. Thus, under regularity conditions, the one-step debiased estimator \[v_{n,OS}:=\frac{1}{n}\sum_{i=1}^{n}\left[\frac{I\{A_i=f_n(X_i)\}}{g_n(f_n(X_i),X_i)}\left\{Y_i-Q_n(f_n(X_i),X_i)\right\}+Q_n(f_n(X_i),X_i)\right]\] of $v_0$ is nonparametric efficient, where $Q_n$ and $g_n$ are estimators of $Q_0$ and $g_0$, respectively, and  $f_n$ is defined pointwise as  $f_n(x):=I\{Q_n(1,x)>Q_n(0,x)\}$. The one-step debiased estimator of $v_{0,s}$ is defined similarly, with $f_n$ replaced by any appropriate estimator of $f_{0,s}$, such as $f_{n,s}(x):=I\{Q_{n,s}(1,x)>Q_{n,s}(0,x)\}$ with $Q_{n,s}(a,x)$ obtained by flexibly regressing outcome $Q_n(a,X)$ onto $X_{-s}$ for each $a\in\{0,1\}$.\vspace{.15in}

\noindent\emph{Example 6: Classification accuracy under outcome missingness}\\
Suppose the ideal-data structure consists of $\mathbbs{Z}:=(X,Y)\sim\mathbbs{P}$ and the predictiveness measure of interest based on this ideal data structure is the classification accuracy measure, $\mathbbs{V}(f,\mathbbs{P}):=\mathbbs{P}\{Y = f(X)\}$, described in Example 3. Suppose that the outcome $Y$ is subject to missingness, so that the observed data structure is  $Z:=(X,\Delta,U)$, where $\Delta$ is the indicator of having observed the outcome $Y$, and we have defined $U:=\Delta Y$.  The observed-data predictiveness measure
\[V(f,P):=E_P\left[P\{U = f(X) \mid \Delta = 1, X\}\right]\]
equals the ideal-data accuracy measure provided that (a) $\Delta$ and $Y$ are independent given $X$, and (b) $P\left(\Delta=1\mid X=x\right)>0$ for $P$-almost every value $x$. In other words, the provided identification holds provided the outcome is missing at random (relative to $X$), and there is no subpopulation of patients (as defined by the value of $X$) for which the outcome can never be observed. Defining $\pi_0(x):=P_0\left(U=1\mid\Delta=1,X=x\right)$, the unrestricted optimizers $f_0$ and $f_{0,s}$  are given pointwise by $f_0(x)=I\left\{\pi_0(x)>0.5\right\}$ and  $f_{0,s}(x)=I\left\{E_0\left\{\pi_0(X)\mid X_{-s}=x_{-s}\right\}>0.5\right\}$. Finally, the nonparametric EIF of $P\mapsto V(f_0,P)$ at $P_0$ is given by
\begin{align*}
\phi_0:z\mapsto\frac{\delta}{g_0(x)}\left[I\{u=f_0(x)\}-Q_0(x)\right]+Q_0(x)-V(f_0,P_0)\ ,
\end{align*} where we now have defined the nuisance parameters $g_0(x):=P_0\left(\Delta = 1 \mid X = x\right)$ and $Q_0(x):=P_0\left\{Y=f_0(x)\mid \Delta=1,X=x\right\}=f_0(x)\pi_0(x)+\{1-f_0(x)\}\{1-\pi_0(x)\}$, so that $Q_0(x)$ is no more than a simple transformation of $\pi_0(x)$.  Under regularity conditions, the one-step debiased estimator \[v_{n,OS}:=\frac{1}{n}\sum_{i=1}^{n}\frac{\Delta_i}{g_n(X_i)}\left[I\{U_i=f_n(X_i)\}-Q_n(X_i)\right]+Q_n(X_i)\] of $v_0$ is nonparametric efficient, where $g_n$ and $\pi_n$ are consistent estimators  of $g_0$ and $\pi_0$, and we define $f_n$ and $Q_n$  pointwise as $f_n(x):=I\{\pi_n(x)>0.5\}$ and $Q_n(x):= f_n(x)\pi_n(x)+\{1-f_n(x)\}\{1-\pi_n(x)\}=\max\{\pi_n(x),1-\pi_n(x)\}$. The one-step debiased estimator of $v_{0,s}$ is defined identically except that all instances of $f_n$ are replaced by $f_{n,s}$, which we define pointwise as $f_n(x):=I\{\pi_{n,s}(x)>0.5\}$, with $\pi_{n,s}$ representing an appropriate estimator of $\pi_{0,s}:=E_0\{\pi_0(X)\mid X_{-s}=x_{-s}\}$, obtained, for example, by  flexibly regressing outcome $\pi_n(X)$ onto $X_{-s}$.

\section{Numerical experiments}\label{sec:sims}

\subsection{Simulation setup}

\revision{We now present empirical results describing the performance of our proposed plug-in VIM estimator. In all cases, our simulated dataset included independent replicates of $(X,Y)$, where $X$ is a covariate vector with independent  components $X_1,\ldots,X_p$ each following a standard normal distribution and a binary outcome $Y$ following a Bernoulli distribution with success probability $\Phi(\beta_{01}x_1+\ldots+\beta_{0p}x_p)$ conditional on $X=x$, where $\Phi$ is the standard normal distribution function. In Scenario 1, we set $p = 2$ and $\beta_0 = (2.5,3.5)$, whereas in Scenario 2, we took $p = 4$ and $\beta_0=(2.5, 3.5, 0,0)$; thus, in all cases, the first two features had nonzero importance and the remaining features (if any) had zero importance. In this specification, $Y$ follows a probit model. For each scenario considered, we generated 1000 random datasets of size $n \in \{100, 500, 1000, \dots, 4000\}$, and considered the importance of both $X_1$ and $X_2$ in Scenario 1 and the importance of $X_2$ and $X_3$ in Scenario 2. In each scenario, we considered VIMs based on classification accuracy (\emph{Example 3}) and the area under the ROC curve (\emph{Example 4}). The true values of these VIMs implied by the data-generating mechanisms considered under Scenarios 1 and 2 are provided in Table~\ref{tab:truths}. All analyses were performed using our R package \texttt{vimp} and may be reproduced using code available online (see details in the Supplementary Material). Since results were similar for accuracy and AUC, we only display results for accuracy here but provide results for AUC  in the Supplementary Material.}

\begin{table}
\centering
\caption{Approximate values of $\psi_{0,s}$ in the numerical experiments.}
\label{tab:truths}
\begin{tabular}{l|cccc}
    \hline
                   & \multicolumn{4}{c}{Feature of interest} \\
Importance measure & $X_1$ & $X_2$ & $X_3$ & $X_4$ \\
\hline
Accuracy & 0.136 & 0.236 & 0 & 0\\
Area under the ROC curve & 0.105 & 0.221 & 0 & 0\\
\hline
\end{tabular}
\end{table}

\revision{In Scenario 1, we investigate the finite-sample properties of our proposal in a setting in which all features are truly important. We also use this setting to explore the effect of cross-fitting when using flexible estimators of $f_0$ and $f_{0,s}$. Specifically, we compare the performance of our estimation procedure with and without five-fold cross-fitting when using the following estimators of $f_0$ and $f_{0,s}$: a correctly specified (parametric) probit regression model; a generalized additive model \citep[GAM;][implemented in the R package \texttt{mgcv}]{hastie1990}; random forests \citep[RF;][implemented in the R package \texttt{ranger}]{breiman2001}; and the Super Learner \citep[SL;][implemented in the R package \texttt{SuperLearner}]{vanderlaan2007}. The latter estimator is a particular implementation of stacking \citep{wolpert1992} with favorable finite-sample and asymptotic performance guarantees \citep{vanderlaan2007}. For the Super Learner, we used a library consisting of gradient boosted trees \citep[][implemented in the R package \texttt{xgboost}]{friedman2001}, GAMs (implemented in the R package \texttt{gam}), and random forests, each with the default tuning parameter choices, in addition to parametric probit regression, with five-fold cross-validation to determine the optimal convex combination of these learners that minimizes the cross-validated negative log-likelihood risk. The resulting optimal convex combination of these individual algorithms is the Super Learner-based conditional mean estimator we adopt in any case where the Super Learner was fit. We do not use sample-splitting, since the results of Section~\ref{sec:est} are valid under the alternative. We use Algorithm~\ref{alg:cv_est} to compute the cross-fitted point and standard error estimators for the importance of $X_1$ and $X_2$, from which we computed nominal 95\% Wald-type confidence intervals. We then computed the empirical bias scaled by $n^{1/2}$, the empirical variance scaled by $n$, the empirical coverage of confidence intervals, and the width of these intervals.}

\revision{In Scenario 2, we study the properties of our proposal under the null hypothesis. In this case, we used sample-splitting since the importance of $X_3$ and $X_4$ is zero. We again ran both cross-fitted  and non-cross-fitted implementations, and considered the same learning strategies as in Scenario 1, with one exception: in this case, we added the lasso \citep[][implemented in the R package \texttt{glmnet}]{tibshirani1996} to the library of candidate learners in the Super Learner. As before, we computed point estimates and nominal 95\% Wald-type confidence intervals but also obtained $p$-values for the null hypothesis using the sample-splitting procedure of Algorithm~\ref{alg:sscf}. We then computed the empirical bias scaled by $n^{1/2}$, the empirical variance scaled by $n$, the empirical coverage of confidence intervals, and the rejection probability for the proposed hypothesis test.}

\subsection{Primary empirical results}\label{sec:primary-sims}

In Figure~\ref{fig:main_alt_performance}, we display the results of the experiment conducted under Scenario 1, in which both features have nonzero importance. \revision{For ease of visualization, we only display the results for $X_2$; the results for $X_1$ are similar and available in the Supplementary Material. In the top-left panel, we observe that the bias of the proposed estimators decreases to zero at rate faster than $n^{1/2}$ for all non-cross-fitted estimators except those based on random forests and Super Learner, whereas it does so for all cross-fitted estimators. This reflects the need for cross-fitting in cases where the Donsker class conditions of Theorem~\ref{thm:general_vim} may fail to hold. The top-right panel shows that the variance of all estimators is approximately proportional to $n$. In the bottom-left panel, we observe that coverage of nominal 95\% confidence intervals increases to the nominal level with increasing sample size for all cases except the non-cross-fitted estimators based on random forests and Super Learner. In the bottom-right panel, we see that the width of these intervals decreases with increasing sample size, as expected.}

\revision{In Figure~\ref{fig:main_null_performance}, we display the results pertaining to null feature $X_3$ in the experiments conducted under Scenario 2. Here, it appears that the bias vanishes at a rate faster than $n^{-1/2}$ for both the cross-fitted and non-cross-fitted estimators (top-left panel), but that the variance of the non-cross-fitted estimators tends to increase with increasing sample size, especially for the more flexible learning algorithms (top-right panel). We observe that empirical coverage is near the nominal level at all sample sizes (bottom-left panel). Finally, we see that the type I error of the proposed hypothesis test is controlled at the nominal level for all cross-fitted procedures, but not so for their non-cross-fitted counterparts, yielding an inflated type I error in that case (bottom-right panel). In the Supplementary Material, we present results for the non-null feature $X_2$, which show that power of the proposed test is large for all sample sizes considered here.}

\begin{figure}
\centering
\includegraphics[width = 1\textwidth]{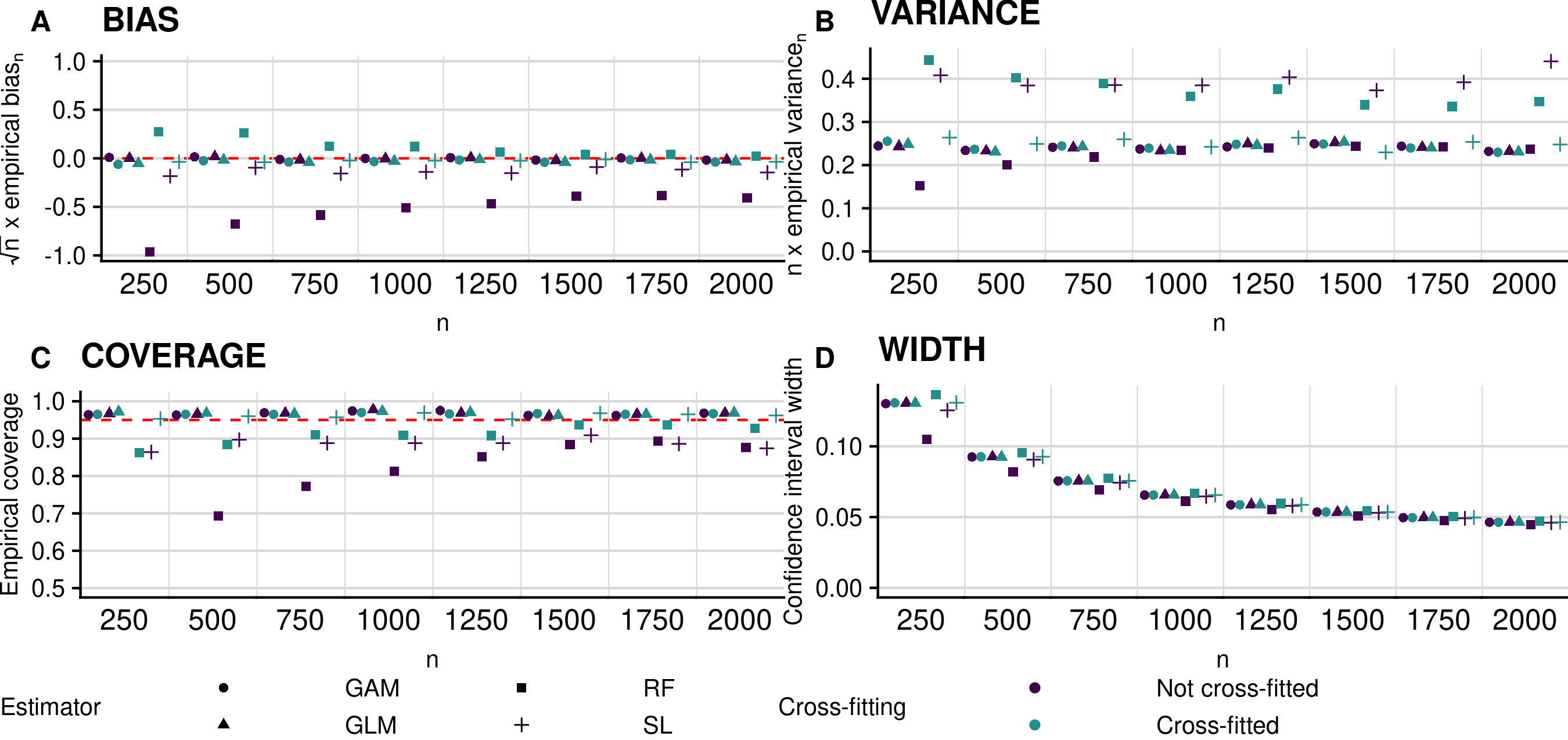}
\caption{Performance of plug-in estimators for estimating (non-zero) importance of $X_2$ in terms of accuracy under Scenario 1 (all features have non-zero importance). Clockwise from top left: empirical bias of the proposed plug-in estimator scaled by $n^{1/2}$; empirical variance scaled by $n$; empirical coverage of nominal 95\% confidence intervals; and average width of these intervals. Circles, triangles, squares and plus symbols denote estimators based on the use of generalized additive models (GAMs), probit regression (GLM), random forests (RF), and the Super Learner (SL), respectively. Blue and green symbols denote non-cross-fitted and cross-fitted estimators, respectively.}
\label{fig:main_alt_performance}
\end{figure}

\begin{figure}
\centering
\includegraphics[width = 1\textwidth]{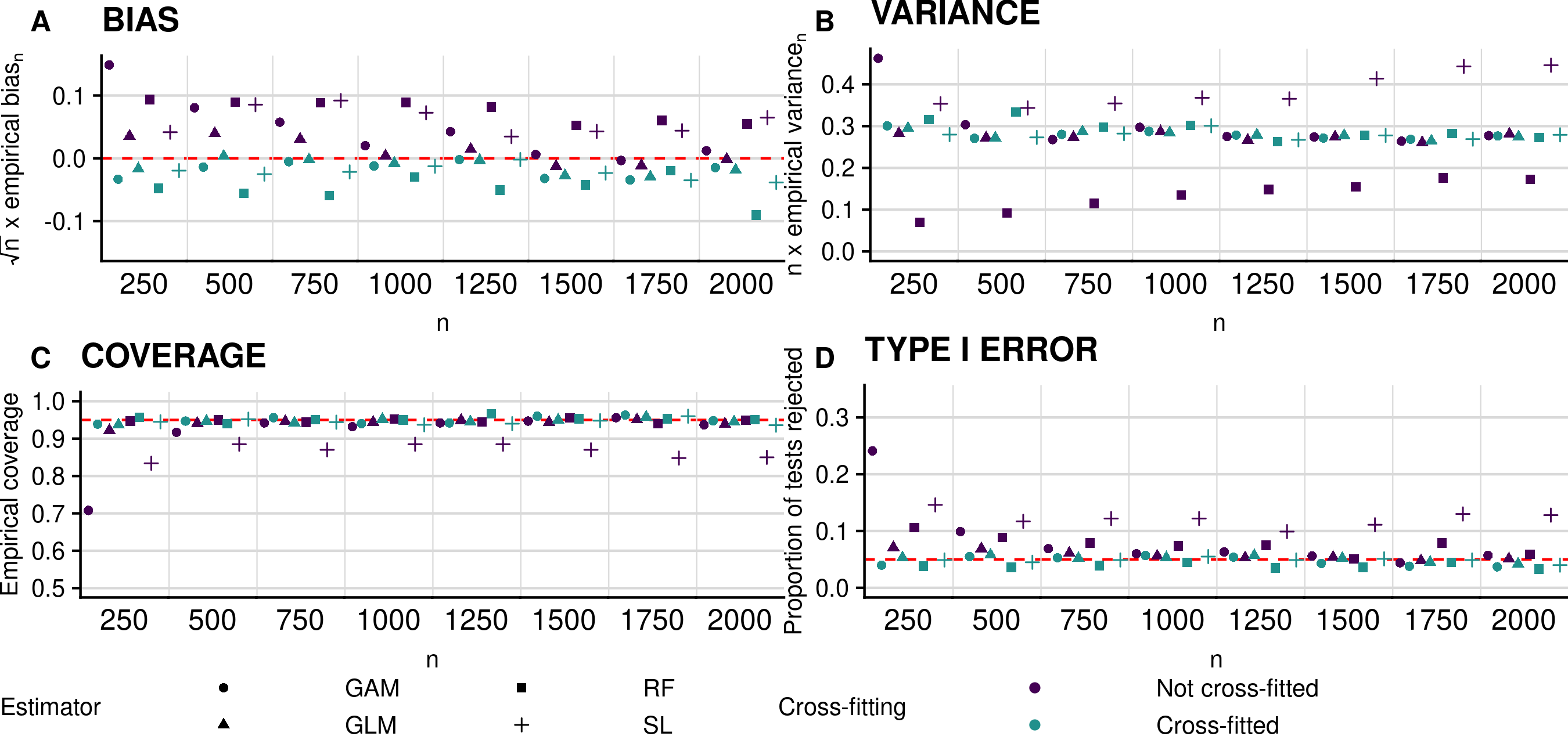}
\caption{Performance of plug-in estimators for estimating (zero) importance of $X_3$ in terms of accuracy under Scenario 2. Clockwise from top left: empirical bias of the proposed plug-in estimator scaled by $n^{1/2}$; empirical variance scaled by $n$; empirical coverage of nominal 95\% confidence intervals; and empirical type I error of the proposed hypothesis test. Circles, triangles, squares and plus symbols denote estimators based on the use of generalized additive models (GAMs), probit regression (GLM), random forests (RF), and the Super Learner (SL), respectively. Blue and green symbols denote non-cross-fitted and cross-fitted estimators, respectively.}
\label{fig:main_null_performance}
\end{figure}

This simulation study suggests that the estimation and inferential procedures proposed, including our null testing approach, have good practical performance and are properly calibrated, as suggested by theory. \revision{Our findings suggest that cross-fitting is critical when flexible algorithms are used, in which case the estimation procedure without cross-fitting performs poorly while its cross-fitted counterpart instead shows good performance. This is the case both for point and interval estimation, as we explicitly show in the Supplementary Material. When correctly-specified parametric regression models are implemented, both procedures (with and without cross-fitting) perform similarly well. This reflects the fact that when parametric estimators are used, condition (B3) is typically satisfied and cross-fitting is then not needed.}

\subsection{Additional empirical results}

\revision{In the Supplementary Material, we present results for additional features under Scenarios 1 and 2, observing similar patterns to those presented in Figures~\ref{fig:main_alt_performance} and \ref{fig:main_null_performance}. We also consider pairing a non-cross-fitted standard error estimator with the cross-fitted estimation procedure, observing reduced coverage compared to the cross-fitted standard error estimator of Algorithm~\ref{alg:cv_est}. Finally, we present results from additional investigations scrutinizing the performance of our proposal in higher dimensions, both with and without correlated features. We found, in small samples, that the presence of many independent null features results in an increased bias in the estimation of the importance of non-null features, with a corresponding decrease in empirical interval coverage. However, this inflated bias and undercoverage dissipate as the sample size increases. A similar pattern was seen in the presence of correlated null features. This suggests that greater dimensionality indeed increases the difficulty of the statistical problem at hand, but that correlation between features does not exarcerbate this challenge beyond rendering more difficult the interpretation of the population VIM values.}

\section{Studying an antibody against HIV-1 infection}\label{sec:data}

Broadly neutralizing antibodies (bnAbs) against HIV-1 neutralize a large fraction of genetic variants of HIV-1. Two harmonized, placebo-controlled randomized trials were conducted to evaluate VRC01, a promising bnAb, for its ability to prevent HIV-1 infection \citep{corey2021}. A secondary objective was to assess how VRC01 prevention efficacy depends on amino acid (AA) sequence features of HIV-1.  Because there are thousands of AA features, the statistical analysis plan for addressing this objective requires first restricting attention to a subset of AA features that putatively affect prevention efficacy. Given the underlying assumption that VRC01 prevents infection via \textit{in vivo} neutralization, a useful approach may be to rank AA features based on their estimated VIM for predicting \textit{in vitro} neutralization --- whether or not an HIV-1 virus is sensitive to neutralization by VRC01 --- and select only the top-ranked features \revision{for further analyses}.

% In an effort to determine these important AA features, we analyzed the HIV-1 envelope (Env) AA sequence features of 611 publicly-available HIV-1 Env pseudoviruses made from blood samples of HIV-1 infected individuals \citep{magaret2019}
\revision{In an effort to determine these important AA features, \citet{magaret2019} analyzed the HIV-1 envelope (Env) AA sequence features of 611 publicly-available HIV-1 Env pseudoviruses made from blood samples of HIV-1 infected individuals. All analyses accounted for the geographic region of the infected individuals. Among AA sequence features, approximately 800 individual features and 13 groups of features were of interest, e.g., polymorphic AA positions in Env AA that comprise the VRC01 antibody footprint to which VRC01 binds. These groups of features are described more fully in the Methods section of \cite{magaret2019}.}
% There, we focused on a definition of variable importance as the difference in nonparametric $R^2$, and used as outcome an indicator of whether or not the 50\% inhibitory concentration IC$_{50}$ (defined as the concentration of VRC01 necessary to neutralize 50\% of viruses \textit{in vitro}, with large values of the IC$_{50}$ indicating that the virus was resistant to neutralization; \citealp{montefiori2009}) was right-censored
\revision{There, the authors focused on a definition of variable importance as the difference in nonparametric $R^2$, and used as outcome an indicator of whether or not the 50\% inhibitory concentration, IC$_{50}$ (defined as the concentration of VRC01 necessary to neutralize 50\% of viruses \textit{in vitro}, with large values of the IC$_{50}$ indicating that the virus was resistant to neutralization; \citealp{montefiori2009}), was right-censored. However, the AMP trials have identified the 80\% inhibitory concentration (IC$_{80}$) as a possible biomarker of prevention efficacy, with 75.4\% estimated efficacy against the most sensitive viruses (IC$_{80} < 1$). Since many observations in our dataset are missing IC$_{80}$ values, we use as outcome the binary indicator that $\text{IC}_{50} < 1$. Here, analyzing the same data set, we compare results based on the outcome of \cite{magaret2019} with a variable importance analysis based on classification accuracy and AUC and the AMP-based outcome $\text{IC}_{50} < 1$. We consider a \textit{marginal} VIM value, evaluating the intrinsic importance of each feature group of interest relative to geographic confounding variables -- this can be achieved by considering the full feature vector in \eqref{eq:vimp_functional} to be simply the geographic confounders plus the feature group of interest. We provide a replication of \cite{magaret2019} using a harmonized outcome in the Supplementary Material.}

We used the Super Learner with a large library of candidate learners to estimate the involved regression functions. These learners included the lasso, random forests, and boosted decision trees, each with varying tuning parameters. Details on our library of learners are described in the Supplementary Material. Our resulting estimator is the convex combination of the candidate estimators, where we used five-fold cross-validation to determine the convex combination that minimized the negative log-likelihood risk. Finally, to make inference on the VIM values considered, we used the sample-split cross-fitted method \revision{(Algorithm~\ref{alg:sscf}) studied in the simulations under Scenario 2}.

In Figure~\ref{fig:amp_results}, we display the results of this analysis and the feature groups of interest. The top-ranked feature groups do not differ much between different VIMs but the magnitude of both importance and $p$-values depends greatly on the measure chosen. Both VIMs result suggest that the CD4 binding sites, the VRC01 binding footprint, \revision{sites with sufficient exposed surface area (ESA sites), sites with residues that co-vary with the VRC01 binding footprint (co-varying sites), and sites for indicating N-linked glycosylation (glycosylation sites)} are the \revision{five} most important groups. \revision{The finding that CD4 binding sites are in the most important groups across VIMs matches our expectations from basic science experiments that have identified AA substitutions at CD4 binding sites that altered VRC01 neutralization sensitivity}. This result is in line with \cite{magaret2019}. Based on our proposed hypothesis test, we computed $p$-values for a test of the strict null hypothesis (that is, $\beta = 0$) for each group. We found that AA features in the CD4 binding sites (group 2), VRC01 binding footprint (group 1), ESA sites (group 3), co-varying sites (group 5), and glycosylation sites (group 8) had $p$-values of $6.98\times 10^{-9}$, $8.14\times 10^{-9}$, $1.69\times 10^{-7}$, $4.66\times 10^{-6}$, and $8.07\times 10^{-6}$, respectively, based on AUC (denoted by stars in Figure~\ref{fig:amp_results}). Based on these analyses, AA features in these groups may be prioritized for the forthcoming trial data analyses. Additionally, taking the set of top-ranked features above a minimum threshold may help to narrow the set of gp160 AA sequence features to pre-specify for the analysis of the AMP trial data sets. \revision{Our recommendation, nonetheless, is to analyze all feature sets in secondary or supporting analysis of the AMP trial data sets to ensure that the results generated are comprehensive.}

\begin{figure}
\centering
\includegraphics[width=1\textwidth]{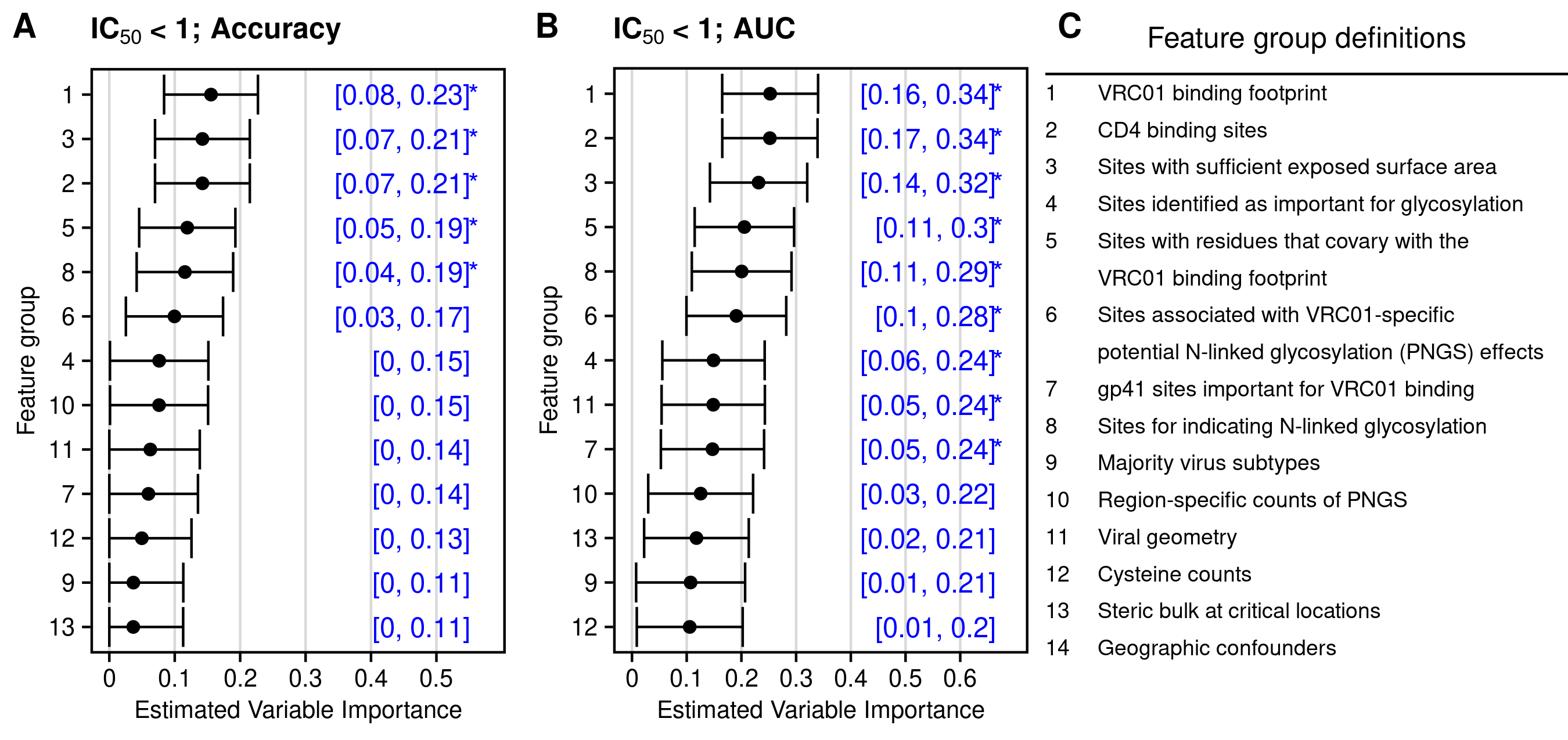}
\caption{Variable importance measured by accuracy (panel A) and AUC (panel B) for the groups defined in panel C. Stars denote importance deemed statistically significantly different from zero at the 0.0038 (0.05 / 13) level.}
\label{fig:amp_results}
\end{figure}

\section{Discussion}\label{sec:conclusions}

We have proposed a \revision{general} model-agnostic framework for statistical inference on population-level VIMs. These measures are summaries of the true data-generating mechanism, defined as a contrast between the predictiveness of the best possible prediction function based on all available features versus all features but those under consideration. We found that plug-in estimators of these VIMs are asymptotically linear and nonparametric efficient under regularity conditions. Through examples, we showed that many simple and commonly used VIMs fall within this framework. We found in numerical experiments that our proposed cross-fitted VIM estimator enjoys good operating characteristics, and that these characteristics match our theoretical expectations. More complex predictiveness measures and sampling scenarios, including missing data, may also be analyzed within our proposed framework, though these cases typically require more effort, including the computation of an influence function. Interpretation of the estimated VIMs depends on the application, and may include considering the ranked VIM values, or considering features with VIM values above some scientifically meaningful threshold.

\revision{Defining the importance of individual features in cases with large amounts of correlation is challenging. In practice, we recommend making use of any available background scientific knowledge either to group variables that are expected to be highly correlated or to develop an appropriate causal model. In settings where this knowledge is lacking, it may be useful to consider, for example, unsupervised methods to cluster variables before assessing variable importance; however, further work is needed to determine how to preserve inferential validity with any such procedure. One alternative approach to handling correlated features is to consider \emph{marginal} importance, wherein each feature in turn could be considered as the `full set of covariates' and its importance could be assessed relative to the null feature vector; if there are concerns about confounding factors, these can constitute the `null feature vector' and each feature could be added to the potential confounders. A second alternative approach is to use measures like the Shapley Population VIM \citep[SPVIM; ][]{williamson2020c}. Since SPVIM is defined as the average increase in predictive power from including a particular feature in \emph{all possible subsets} of the remaining features, use of this approach comes at the cost of significantly increased complexity.}

The inferential procedures following Theorems 1 and 2 can be used whenever it is known a priori that the features of interest have non-zero importance. \revision{We note that, as an alternative, a nonparametric bootstrap scheme could be used in which $f_0$ and $f_{0,s}$ are not re-estimated over bootstrap samples but rather fixed at their original estimates. The use of this bootstrap is illustrated in the Supplementary Material, where it is shown to yield similar results as the inferential procedures described in this paper. If the features of interest may have zero importance, inference should generally be conducted using sample-splitting, as described in Section~\ref{sec:test}. There, we propose confidence intervals valid even when a feature of interest has zero importance and a test of the zero-importance hypothesis.} Our numerical results suggest that the resulting test controls type I error rate at the desired level. However, since our procedure involves sample-splitting without data reuse, it does not fully exploit the information available in the data\revision{, and may possibly be improved upon.} \revision{Use of the bootstrap in this context is complicated by the need to re-estimate $f_0$ and $f_{0,s}$.} Developing a more powerful test of the null importance hypothesis is an important unresolved need. This objective could be achieved, on one hand, by considering modifications of our current approach, including averaging results over multiple splits of the dataset or choosing split sizes more judiciously, or on the other hand, by utilizing more complex analytical tools, including approximate higher-order influence functions. These ideas are being pursued in ongoing research.

\section*{Software and supplementary material}\label{sec:software}
% We implement the methods discussed above in the R package \texttt{vimp} and the Python package \texttt{vimpy}, both available on \href{https://cran.r-project.org/web/packages/vimp/index.html}{CRAN} and \href{https://pypi.org/project/vimpy/}{PyPI}, respectively. Additional technical details are available in the supplementary material. All results may be reproduced using code available  \href{https://github.com/bdwilliamson/vimp_supplementary}{on GitHub}.
We implement the methods discussed above in R and Python packages. Additional technical details are available in the Supplementary Material. All results may be reproduced using code available online.
The data from Section~\ref{sec:data} are available at \href{https://github.com/benkeser/vrc01/tree/1.0}{https://github.com/benkeser/vrc01/tree/1.0}.

% \section*{Acknowledgements}\label{sec:acknowledgements}
% This work was supported by NIH grants F31AI140836, R01AI029168, R01HL137808 UM1AI068635, and S10OD028685. The opinions expressed in this article are those of the authors and do not necessarily represent the official views of the NIH.

\begin{appendices}

\section{Special case: standardized V-measures}\label{sec:app-a}

\revision{Beyond smoothness requirements, the results  presented in Section 3 do not impose much structure on the predictiveness measure.} However, it is often the case that the predictiveness measure has the form $V(f,P)=a+V_1(f,P)/V_2(P)$ with
\begin{align*}
    V_1(f, P) := E_P\left\{G\left((Y_1,f(X_1)), \dots, (Y_m,  f(X_m))\right)\right\}
\end{align*}
for some symmetric function $G: (\mathcal{Y}\times \mathcal{Y})^m \to \mathbb{R}$, where $a\in\mathbb{R}$ is a fixed constant, $V_2:\mathcal{M}\rightarrow\mathbb{R}$ is  Hadamard differentiable, and the expectation defining $V_1$ is over the distribution of independent draws $(X_1,Y_1),\ldots,(X_m,Y_m)$ from $P$. In this case, the plug-in estimator $V_1(f_n,P_n)$ of $V_1(f_0,P_0)$ is a $V$-statistic of degree $m$ \citep{hoeffding1948}, whereas the denominator $V_2(P_0)$ does not depend on $f_0$ and typically serves as a normalization constant. As such, we refer to any predictiveness measure of this form as a \emph{standardized V-measure}. We note that each example \revision{presented in Section~\ref{subsec:examples}} is a standardized V-measure, defined respectively by: \begin{flalign*}\mbox{[1]}\ &\ a=1,\ G((u,v))=-(u-v)^2,\ V_2(P)=var_P(Y),\ m=1;&&\\
\mbox{[2]}\ &\ a=1,\ G((u,v))=-\{u\log v+(1-u)\log(1-v)\},&&\\
 &\ V_2(P)=P\left(Y=1\right)\log P\left(Y=1\right)+P\left(Y=0\right)\log P\left(Y=0\right),\ m=1;&&\\
\mbox{[3]}\ &\ a=0,\ G((u,v))=I(u=v),\ V_2(P)=1,\ m=1;&&\\
\mbox{[4]}\ &\ a=0,\ G((u_1,v_1),(u_2,v_2))=\{I(u_1=0,u_2=1,v_1<v_2)+I(u_2=0,u_1=1,v_2< v_1)\}/2,&&\\
&\ V_2(P)=P\left(Y=1\right)P\left(Y=0\right),\ m=2.&&
\end{flalign*} This is useful to note because whenever $V$ is a standardized $V$-measure, the influence function $\phi_0$ of $V(f_n,P_n)$ can be described more explicitly. Specifically, its pointwise evaluation $\phi_0(z)$ at a given observation value $z=(x,y)$ is given by \[m\left[\frac{E_0\left\{G\left((y,f_0(x)),(Y_2,f_0(X_2)),\ldots,(Y_m,f_0(X_m))\right)\right\}}{V_2(P_0)}-V(f_0,P_0)\right]-\frac{\dot{V}_2(P_0;\delta_z-P_0)}{V_2(P_0)}V(f_0,P_0)\] with $\dot{V}_2(P_0;\delta_z-P_0)$ denoting the G\^{a}teaux derivative of $V_2$ at $P_0$ in the direction $h=\delta_z-P_0$.  Except for the influence function of the normalization estimator $V_2(P_n)$, which is typically straightforward to compute, this is an explicit form. In Examples 1--4, the influence function of $V(f_n,P_n)$ can thus be derived respectively as: \begin{flalign*}
\mbox{[1]}\ &\ \phi_0(z)=&&\hspace{-.1in}-\left\{y-\mu_0(x)\right\}^2/\sigma^2_0+v_0\left\{2-(y-\mu_0)^2/\sigma_0^2\right\};&&\\
\mbox{[2]}\ &\ \phi_0(z)=&&\hspace{-.1in}-2\left[y\log\mu_0(x) + (1-y)\log\left\{1 - \mu_0(x)\right\}\right]/\overline{\pi}_0+
v_0\left[2\log\left\{\pi_0/(1-\pi_0)\right\}(y - \pi_0)/\overline{\pi}_0 - 1\right];&&\\
\mbox{[3]}\ &\ \phi_0(z)=&&\hspace{-.1in}\ y I\left\{\mu_0(x)>0.5\right\}+(1-y)I\left\{\mu_0(x)\leq 0.5\right\}-v_0;&&\\
\mbox{[4]}\ &\  \phi_0(z)=&&\hspace{-.1in}\ (1-y)P_0\left\{\mu_0(X)>\mu_0(x)\mid Y=1\right\}/(1-\pi_0)+yP_0\left\{\mu_0(x)>\mu_0(X)\mid Y=0\right\}/\pi_0&&\\
& && -v_0\left[2+(1-2\pi_0)(y-\pi_0)/\{\pi_0(1-\pi_0)\}\right]\, ,&&
\end{flalign*}
where here we have used the shorthand notation $\mu_0(x):=E_0\left(Y\mid X=x\right)$, $\mu_0:=E_0\left(Y\right)$, $\sigma^2_0:=var_0\left(Y\right)$, $\pi_0:=P_0\left(Y=1\right)$, and $\overline{\pi}_0:=\pi_0\log\pi_0 + (1 - \pi_0)\log(1 - \pi_0)$. Furthermore, for standardized $V$-measures, condition (A2) is often easier to verify. For example, if $m=1$, then it holds trivially since $V_1(f,P)$, the only component of $V(f,P)$ involving $f$,  is linear in $P$.

\end{appendices}

\vspace{0.1in}

{\small
\bibliographystyle{chicago}
\bibliography{brian-papers}
}

\newpage

\section*{SUPPLEMENTARY MATERIAL}

\section{Proof of theorems}\label{sec:proofs}

\subsection{Proof of Theorem \ref{thm:general_vim}}\label{sec:pf_general_vim}

Writing $r_n:=\{V(f_n,P_n)-V(f_n,P_0)\}-\{V(f_0,P_n)-V(f_0,P_0)\}$, we first decompose
\begin{align*}
    v_n - v_0 =& \ \{V(f_0, P_n) - V(f_0, P_0)\} + \{V(f_n, P_0) - V(f_0, P_0)\}+r_n\ .
\end{align*} In view of condition (A2), the functional delta method is applicable and yields that
\begin{align*}
    V(f_0,P_n) - V(f_0,P_0)\ &=\ \dot{V}(f_0, P_0; P_n - P_0) + o_P(n^{-1/2}) \\
    &=\ \frac{1}{n}\sum_{i=1}^n \dot{V}(f_0, P_0; \delta_{Z_i} - P_0) + o_P(n^{-1/2})\ ,
\end{align*}
where $\dot{V}(f_0, P_0; h)$ is the G\^ateaux derivative of the mapping $P \mapsto V(f_0, P)$ at $P_0$ in the direction $h$ and $\delta_z$ is the degenerate distribution on ${z}$. Under condition (A1), we have that $\lvert V(f_n,P_0) - V(f_0,P_0) \rvert \leq C\lVert f_n - f_0 \rVert^2_\mathcal{F} = o_P(n^{-1/2})$ under condition (B1). It remains to show that $r_n=o_P(n^{-1/2})$ as well. For any given $\epsilon>0$, $h\in\mathcal{Q}$ and $f\in\mathcal{F}$, we define
    \[R_0(f,\epsilon,h):=\frac{V(f,P_0+\epsilon h)-V(f,P_0)}{\epsilon}-\dot{V}(f,P_0;h)\ .\]  Setting $\epsilon_n := n^{-1/2}$ and $h_n := n^{1/2}(P_n - P_0)$, we have that
\begin{align*}
    n^{1/2}r_n\ &=\ \frac{[\{V(f_n, P_n) - V(f_n, P_0)\} - \{V(f_0, P_n) - V(f_0, P_0)\}]}{\epsilon_n} \\
    &=\ \{\dot{V}(f_n,P_0;h_n)+R_0(f_n,\epsilon_n,h_n)\}-\{\dot{V}(f_0,P_0;h_n)+R_0(f_0,\epsilon_n;h_n))\}\ =A_n+B_n\ ,
\end{align*}where $A_n:=\dot{V}(f_n,P_0;h_n)-\dot{V}(f_0,P_0;h_n)$ and $B_n:=R_0(f_n,\epsilon_n,h_n)-R_0(f_0,\epsilon_n,h_n)$, and so, we can write that $
P_0\left(n^{1/2}|r_n|>\epsilon\right)\leq P_0\left(|A_n|>\epsilon/2\right)+P_0\left(|B_n|>\epsilon/2\right)$. On one hand, since we can rewrite $A_n=\dot{V}(f_n,P_0;h_n)-\dot{V}(f_0,P_0;h_n)=n^{1/2}\int g_n(z)d(P_n-P_0)(z)$, under conditions (B2) and (B3), an application of Lemma 19.24 of van der Vaart (2000) yields that $A_n=o_P(1)$ under $P_0$, and so, $P_0(|A_n|>\epsilon/2)\longrightarrow 0$. On the other hand, we can write \begin{align*}
P_0\left(|B_n|>\epsilon/2\right)\ &=\ P_0\left(|B_n|>\epsilon/2,\|f_n-f_0\|_\mathcal{F}<\delta\right)+P_0\left(|B_n|>\epsilon/2,\|f_n-f_0\|_\mathcal{F}\geq \delta\right)\\
&\hspace{-0.7in}\leq\ P_0\left(\sup\textstyle{_{f\in\mathcal{F}:\|f-f_0\|_\mathcal{F}<\delta}}R_0(f,\epsilon_n,h_n)>\epsilon/4,\|f_n-f_0\|_\mathcal{F}<\delta\right)+P_0\left(\|f_n-f_0\|_\mathcal{F}\geq \delta\right)\\
&\hspace{-0.7in}\leq\ P_0\left(\sup\textstyle{_{f\in\mathcal{F}:\|f-f_0\|_\mathcal{F}<\delta}}R_0(f,\epsilon_n,h_n)>\epsilon/4\right)+P_0\left(\|f_n-f_0\|_\mathcal{F}\geq \delta\right).
\end{align*} Since the first and second summands tend to zero by conditions (A2) and (B1), respectively, it follows that $P_0\left(|B_n|>\epsilon/2\right)\longrightarrow 0$. In summary, under conditions (A1)--(A2) and (B1)--(B3), we find that \[v_n-v_0\ =\ \frac{1}{n}\sum_{i=1}^{n}\dot{V}(f_0,P_0;\delta_{Z_i}-P_0)+o_P(n^{-1/2})\] under sampling from $P_0$, as claimed.

Now, we verify the claim of asymptotic efficiency. Let $s$ be any bounded element of $L_2^0(P_0)$. We construct the parametric submodel $\{P_{0,\epsilon}\}$ with univariate index $\epsilon$ defined in a neighborhood of zero and with corresponding distribution function defined pointwise as $F_{0,\epsilon}(z):= F_0(z)+\epsilon \int_{(-\infty,z]}s(u)F_0(du)$, where $F_0$ denotes the distribution function of $P_0$ and $(-\infty,z]$ is interpreted as an orthant in the dimension of $\mathcal{Z}$. We note that $z\mapsto \int_{(-\infty,z]}s(u)F_0(du)$ induces a finite signed measure $h(s)\in\mathcal{R}$ since it is cadlag and has finite total variation norm. We then write that \begin{align*}
&\left|\frac{V(f_{0,\epsilon},P_{0,\epsilon})-V(f_0,P_0)}{\epsilon}-\dot{V}(f_0,P_0;h(s))\right|\\
&\leq\ \left|\frac{V(f_{0,\epsilon},P_{0,\epsilon})-V(f_{0,\epsilon},P_0)+V(f_{0,\epsilon},P_0)-V(f_0,P_0)}{\epsilon}-\dot{V}(f_{0,\epsilon},P_0;h(s))+\dot{V}(f_{0,\epsilon},P_0;h(s))-\dot{V}(f_0,P_0;h(s))\right|\\
&\leq\ U_1(\epsilon)+U_2(\epsilon)+U_3(\epsilon)\ ,
\end{align*} where we have defined the summands \[U_1(\epsilon):=\left|\frac{V(f_{0,\epsilon},P_{0,\epsilon})-V(f_{0,\epsilon},P_0)}{\epsilon}-\dot{V}(f_{0,\epsilon},P_0;h(s))\right|\mbox{,\ \ }U_2(\epsilon):=\left|\frac{V(f_{0,\epsilon},P_0)-V(f_0,P_0)}{\epsilon}\right|\] and $U_3(\epsilon):=|\dot{V}(f_{0,\epsilon},P_0;h(s))-\dot{V}(f_{0},P_0;h(s))|$. By conditions (A1) and (A3), we can bound $U_2(\epsilon)$ above by $C\|f_{0,\epsilon}-f_0\|^2_\mathcal{F}/\epsilon=O(\epsilon)$. By condition (A4), we have that $U_3(\epsilon)=O(\epsilon)$. Since $\|f_{0,\epsilon}-f_0\|_{\mathcal{F}}=O(\epsilon)$ by condition (A3), then for small enough $\epsilon$ we have that \[U_1(\epsilon)\ \leq\ \sup_{f\in\mathcal{F}:\|f-f_0\|_\mathcal{F}\leq \delta}\left|\frac{V(f,P_{0,\epsilon})-V(f,P_0)}{\epsilon}-\dot{V}(f,P_0;h(s))\right|,\] where the right-hand side of the inequality itself tends to zero as $\epsilon\rightarrow 0$ in view of condition (A2). In other words, we find that $U_1(\epsilon)=O(\epsilon)$. Thus, we find that \[\left|\frac{V(f_{0,\epsilon},P_{0,\epsilon})-V(f_0,P_0)}{\epsilon}-\dot{V}(f_0,P_0;h(s))\right|=O(\epsilon)\ ,\]which implies that the derivative $\epsilon\mapsto V(f_{0,\epsilon},P_{0,\epsilon})$ at $\epsilon=0$ equals $\dot{V}(f_0,P_0;h(s))$.  In view of \cite{frangakis2015biometrics} and \cite{luedtke2015biometrics}, the evaluation of the nonparametric efficient influence function at observation value $z$ is obtained by choosing $s$ so that $h(s)=\delta_z-P_0$, establishing that $v_n$ is indeed asymptotically efficient relative to a nonparametric model.

\subsection{Proof of Theorem \ref{thm:cv}}\label{sec:pf_cv}

As before, we denote by $B_n \in \{1, \ldots, K\}^n$ a random vector generated by sampling uniformly from $\{1, \ldots, K\}$ with replacement, and by $D_k$ the subset of observations with index in $\{i:B_{n,i}=k\}$ for $k=1,\ldots,K$. Additionally, we denote by ${f}_{k,n}$ an estimator of $f_0$ constructed using the data in $\cup_{j\neq k}D_j$, and we write $P_{k,n}$ for the empirical distribution estimator of $P_0$ based on the data in $D_k$. Recalling that $v_n^* = \tfrac{1}{K}\sum_{k=1}^K V(f_{k,n}, P_{k,n})$,  we note that $v_n^* - v_0=A_{1,K,n}+A_{2,K,n}+A_{3,K,n}$, where
$A_{1,K,n}:=\frac{1}{K}\sum_{k=1}^{K}\left\{V(f_0,P_{k,n}) - V(f_0,P_0)\right\}$, $A_{2,K,n}:=\frac{1}{K}\sum_{k=1}^{K}\left\{V(f_{k,n},P_0) - V(f_0,P_0)\right\}$ and $A_{3,K,n}:=\frac{1}{K}\sum_{k=1}^{K}r_{k,n}$ with $r_{k,n}:=\{V(f_{k,n},P_{k,n})-V(f_{k,n},P_0)\}-\{V(f_0,P_{k,n})-V(f_0,P_0)\}$. We will study separately each of these three summands.

Under condition (A2), the functional delta method can be used to establish the representation $V(f_0,P_{k,n})-V(f_0,P_0)=\dot{V}(f_0,P_0;P_{k,n}-P_0)+o_P(n^{-1/2}_k)=\frac{1}{n_k}\sum_{i\in D_k}\dot{V}(f_0,P_0;\delta_{Z_i}-P_0)+o_P(n^{-1/2}_k)$
 for each $k\in\{1,\ldots,K\}$, from which it follows that \begin{align*}
\left|A_{1,K,n}-\frac{1}{n}\sum_{i=1}^{n}\dot{V}(f_0,P_0;\delta_{Z_i}-P_0)\right|\ &\leq\ \max_{k}\left|\frac{n}{Kn_k}-1\right|\cdot\frac{1}{n}\sum_{i=1}^{n}\dot{V}(f_0,P_0;\delta_{Z_i}-P_0)+\frac{1}{K}\sum_{k=1}^{K}o_P(n_k^{-1/2})\\
&=\ O_P(n^{-1})+o_P(n^{-1/2})\ =\ o_P(n^{-1/2})\ .
\end{align*} Under conditions (A1) and (B1'), we have that \begin{align*}
\left|A_{2,K,n}\right|\ &\leq\ \max_k\left|V(f_{k,n},P_0)-V(f_0,P_0)\right|\ \leq\ C\max_k \|f_{k,n}-f_0\|^2 _\mathcal{F}\ =\ o_P(n^{-1/2})\ .
\end{align*} Finally, we show that $|A_{3,K,n}|=o_P(n^{-1/2})$ by showing that $|r_{k,n}|=o_P(n^{-1/2})$ for each $k$. Similarly as in the proof of Theorem \ref{thm:general_vim}, setting $\epsilon_{k,n}:=n_k^{-1/2}$ and $h_{k,n}:=n^{1/2}_{k}(P_{k,n}-P_0)$, we can write that $
n_k^{1/2}r_{k,n}=A_{k,n}+B_{k,n}$, where  we have defined the terms $A_{k,n}:=\dot{V}(f_{k,n},P_0;h_{k,n})-\dot{V}(f_{0},P_0;h_{k,n})$ and $B_{k,n}:=R_0(f_{k,n},\epsilon_{k,n},h_{k,n})-R_0(f_0,\epsilon_{k,n},h_{k,n})$. Following the same argument made for $B_n$ in the proof of Theorem \ref{thm:general_vim}, we can show that $B_{k,n}=o_P(1)$.  \revision{We then note that $A_{k,n}=n_k^{1/2}\int g_{k,n}(z)d(P_{k,n}-P_0)(z)$. For any $\varepsilon>0$, by Chebyshev's inequality, we have that \[0\ \leq\ P_0\left(|A_{k,n}|>\varepsilon\,\middle|\, \cup_{j\neq k}D_j\right)\ \leq\  \frac{var_0\left[\,g_{k,n}(Z)\,|\,\cup_{j\neq k}D_j\,\right]}{\varepsilon^2}\ \leq\ \frac{P_0g^2_{k,n}}{\varepsilon^2}\ .\]Thus, by condition (B2'), we have that $P_0\left(|A_{k,n}|>\varepsilon\,\middle|\, \cup_{j\neq k}D_j\right)=o_P(1)$. Since  $P_0\left(|A_{k,n}|>\varepsilon\,\middle|\, \cup_{j\neq k}D_j\right)$ is uniformly bounded by virtue of being a probability, this implies that  $E_0[P_0\left(|A_{k,n}|>\varepsilon\,\middle|\, \cup_{j\neq k}D_j\right)]=o(1)$, and so, $P_0\left(|A_{k,n}|>\varepsilon\right)=o(1)$. Thus, we find that $A_{k,n}=o_P(1)$.} As such, we have found that $|r_{k,n}|=o_P(n_k^{-1/2})$, and since $n/n_k\stackrel{P}{\longrightarrow}K$, this implies that $|r_{k,n}|=o_P(n^{-1/2})$.

The proof of nonparametric asymptotic efficiency is identical to that provided for Theorem 1.

\subsection{Proof of Theorem \ref{thm:same_eif}}

Fix an arbitrary $h\in \mathcal{H}$, and let $\{P_{0,\epsilon}\}\subset\mathcal{M}$ be an arbitrary regular univariate parametric submodel through $P_0$ at $\epsilon=0$ and with score $h$ for $\epsilon$ at $\epsilon=0$. Write $f_{0,\epsilon}:=f_{P_{0,\epsilon}}$ for brevity. We note that \begin{align}
V(f_{0,\epsilon},P_{0,\epsilon})-V(f_0,P_0)\ &=\ V(f_{0,\epsilon},P_{0,\epsilon})-V(f_0,P_{0,\epsilon})+V(f_0,P_{0,\epsilon})-V(f_0,P_0)\notag\\
&=\ V(f_{0,\epsilon},P_0)-V(f_0,P_0)+V(f_0,P_{0,\epsilon})-V(f_0,P_0)+o(\epsilon)\ ,\label{sameeif}
\end{align} where the second line follows from the first in view of condition (A5a). By the nonparametric pathwise differentiability of $P\mapsto V(f_0,P)$ at $P_0$, we have that $V(f_0,P_{0,\epsilon})-V(f_0,P_0)=\epsilon\int d_0(z)h(z)dP_0(z)+O(\epsilon^2)$, where $d_0$ is the nonparametric EIF of $P\mapsto V(f_0,P)$ at $P_0$. Condition (A5b) and (A5c) together indicate that \[\left.\frac{d}{d\epsilon}V(f_{0,\epsilon},P_0)\right|_{\epsilon=0}=0\ ,\] and furthermore, that $V(f_{0,\epsilon},P_0)-V(f_0,P_0)=o(\epsilon)$. So, in view of equation \ref{sameeif}, we obtain the representation $V(f_{0,\epsilon},P_{0,\epsilon})-V(f_0,P_0)=\epsilon\int d_0(z)h(z)dP_0(z)+o(\epsilon)$, which implies that $P\mapsto V(f_P,P)$ is pathwise differentiable at $P_0$ relative to the nonparametric model $\mathcal{M}$ and has nonparametric EIF $d_0$.

\section{Explicit description of estimation procedure for Examples 1--4}

In this section, we provide the explicit form of our proposed estimator for Examples 1--4. For each example, we describe both the simple plug-in estimator and the cross-fitted estimator. When we discuss cross-fitting, recall that we generate a random partition assignment vector $B_n \in \{1, \ldots, K\}^n$ by sampling uniformly from $\{1, \ldots, K\}$ with replacement, and denote by $D_k$ the subset of observations with index in $\{i:B_{n,i}=k\}$ for $k=1,\ldots,K$. For each $k = 1,\ldots,K$, we denote by $f_{k,n}$ and $f_{k,n,s}$ estimators of $f_0$ and $f_{0,s}$, respectively, constructed on the data in $\bigcup_{j \neq k}D_j$, and we denote by $P_{k,n}$ the empirical distribution estimator of $P_0$ based on the data in $D_k$.

\noindent\emph{Example 1: $R^2$}\\
The difference in $R^2$ VIM estimator is
\begin{align*}
    \psi_{n,s} =& \ \left[1 - \frac{\sum_{i=1}^n \{Y_i - f_n(X_i)\}^2}{\sum_{i=1}^n (Y_i - \overline{Y}_n)^2}\right] - \left[1 - \frac{\sum_{i=1}^n \{Y_i - f_{n,s}(X_i)\}^2}{\sum_{i=1}^n (Y_i - \overline{Y}_n)^2}\right],
\end{align*}
where $\overline{Y}_n := \tfrac{1}{n}\sum_{i=1}^n Y_i$ is the marginal empirical mean of $Y$. In this example, $f_n = \mu_n$ and $f_{n,s} = \mu_{n,s}$, where $\mu_n$ and $\mu_{n,s}$ are estimators of $\mu_0$ and $\mu_{0,s}$, respectively.
For each $k = 1,\ldots,K$, the fold-specific difference in $R^2$ VIM estimator is
\begin{align*}
    \psi_{k,n,s} =& \ \left[1 - \frac{\tfrac{1}{n_k}\sum_{i \in D_k} \{Y_i - f_{k,n}(X_i)\}^2}{\tfrac{1}{n_k}\sum_{i \in D_k} (Y_i - \overline{Y}_{k,n})^2}\right] - \left[1 - \frac{\tfrac{1}{n_k}\sum_{i \in D_k} \{Y_i - f_{k,n,s}(X_i)\}^2}{\tfrac{1}{n_k}\sum_{i \in D_k} (Y_i - \overline{Y}_{k,n})^2}\right],
\end{align*}
where $n_k := \sum_{i=1}^nI(i \in D_k)$ is the number of observations in fold $k$, and $\overline{Y}_{k,n} := \tfrac{1}{n_k}\sum_{i\in D_k}^n Y_i$ is the marginal empirical mean of $Y$ in fold $k$. The cross-fitted estimator is  then $\psi_{n,s}^* = \tfrac{1}{K}\sum_{k=1}^K \psi_{k,n,s}$.
\vspace{.15in}

\noindent\emph{Example 2: deviance}\\
The difference in deviance VIM estimator is
\begin{align*}
    \psi_{n,s} =& \ \left[1 - \frac{\tfrac{1}{n}\sum_{i=1}^n \{Y_i\log f_n(X_i) + (1 - Y_i)\log (1 - f_n(X_i))\}}{\pi_n \log(\pi_n) + (1 - \pi_n)\log(1 - \pi_n)}\right] \\
    &\ -\left[1 - \frac{\tfrac{1}{n}\sum_{i=1}^n \{Y_i\log f_{n,s}(X_i) + (1 - Y_i)\log (1 - f_{n,s}(X_i))\}}{\pi_n \log(\pi_n) + (1 - \pi_n)\log(1 - \pi_n)}\right],
\end{align*}
where $\pi_n := \tfrac{1}{n}\sum_{i=1}^n Y_i$ is the empirical estimator of the marginal probability $P_0\left(Y=1\right)$. Again, in this example, $f_n = \mu_n$ and $f_{n,s} = \mu_{n,s}$.
For each $k = 1,\ldots,K$, the fold-specific difference in deviance VIM estimator is
\begin{align*}
    \psi_{k,n,s} =& \ \left[1 - \frac{\tfrac{1}{n_k}\sum_{i\in D_k} \{Y_i\log f_{k,n}(X_i) + (1 - Y_i)\log (1 - f_{k,n}(X_i))\}}{\pi_{k,n} \log(\pi_{k,n}) + (1 - \pi_{k,n})\log(1 - \pi_{k,n})}\right] \\
    &\ - \left[1 - \frac{\tfrac{1}{n_k}\sum_{i\in D_k} \{Y_i\log f_{k,n,s}(X_i) + (1 - Y_i)\log (1 - f_{k,n,s}(X_i))\}}{\pi_{k,n} \log(\pi_{k,n}) + (1 - \pi_{k,n})\log(1 - \pi_{k,n})}\right],
\end{align*}
where $\pi_{k,n} := \tfrac{1}{n_k}\sum_{i\in D_k} Y_i$ is the marginal estimator of $P_0\left(Y=1\right)$ in fold $k$. The cross-fitted estimator is then $\psi_{n,s}^* = \tfrac{1}{K}\sum_{k=1}^K \psi_{k,n,s}$.
\vspace{.15in}

\noindent\emph{Example 3: classification accuracy}\\
The difference in classification accuracy VIM estimator is $\psi_{n,s} = \tfrac{1}{n}\sum_{i=1}^nI\{Y_i = f_n(X_i)\} - \tfrac{1}{n}\sum_{i=1}^nI\{Y_i = f_{n,s}(X_i)\}$. Sensible estimators of $f_0$ and $f_{0,s}$ are given by \[f_n:x\mapsto I\left\{\mu_n(x)>0.5\right\}\mbox{\ \ and\ \ }f_{n,s}:x\mapsto I\left\{\mu_{n,s}(x)>0.5\right\}.\] The fold-specific difference in classification accuracy VIM estimator is
\begin{align*}
     \psi_{k,n,s} = \frac{1}{n_k}\sum_{i\in D_k}I\{Y_i = f_{k,n}(X_i)\} - \frac{1}{n_k}\sum_{i\in D_k}I\{Y_i = f_{k,n,s}(X_i)\}\ .
\end{align*}
The cross-fitted estimator is then $\psi_{n,s}^* = \tfrac{1}{K}\sum_{k=1}^K \psi_{k,n,s}$.
\vspace{.15in}

\noindent\emph{Example 4: area under the ROC curve}\\
The difference in AUC VIM estimator is
\begin{align*}
    \psi_{n,s} =& \ \frac{1}{n_0n_1}\sum_{i=1}^n \sum_{j=1}^n I\{f_n(X_i) < f_n(X_j)\}(1-Y_i)Y_j - \frac{1}{n_0n_1}\sum_{i=1}^n \sum_{j=1}^n I\{f_{n,s}(X_i) < f_{n,s}(X_j)\}(1-Y_i)Y_j\ ,
\end{align*}
where $n_1 := \sum_{i=1}^n Y_i$ is the number of observations with corresponding $Y=1$ and $n_0:=n-n_1$. As above, in this example, we can take $f_n = \mu_n$ and $f_{n,s} = \mu_{n,s}$.
The fold-specific difference in AUC VIM estimator is
\begin{align*}
    \psi_{k,n,s} =& \ \frac{1}{n_{k,0}n_{k,1}}\sum_{i\in D_k} \sum_{j \in D_k} I\{f_{k,n}(X_i) < f_{k,n}(X_j)\}(1-Y_i)Y_j \\
    &\ - \frac{1}{n_{k,0}n_{k,1}}\sum_{i\in D_k} \sum_{j\in D_k} I\{f_{k,n,s}(X_i) < f_{k,n,s}(X_j)\}(1-Y_i)Y_j,
\end{align*}
where $n_{k,1} := \sum_{i\in D_k} I(Y_i = 1)$ is the number of observations with corresponding $Y = 1$ in fold $k$ and $n_{k,0}:=n_k-n_{k,1}$. The cross-fitted estimator is then $\psi_{n,s}^* = \tfrac{1}{K}\sum_{k=1}^K \psi_{k,n,s}$.
\vspace{.15in}

\section{Additional technical details}\label{sec:more_technical_details}

\subsection{Bayes classifier maximizes classification accuracy}\label{sec:acc_maximizer_proof}
Suppose that $Y\in\{0,1\}$ is a binary random variable. Define the Bayes classifier $b_0: x \mapsto I\{\mu_0(x) > 1/2\}$ with  $\mu_0(x) = E_0(Y \mid X = x)$. For any fixed $x \in \mathcal{X}$, we have that
\begin{align*}
    P_0\{f(X) = Y \mid X = x\}\ &=\ P_0\{Y = 1, f(X) = 1 \mid X = x\} + P_0\{Y = 0, f(X) = 0 \mid X = x\} \\
    &=\ f(x)P_0(Y = 1 \mid X = x) + \{1-f(x)\}P_0(Y = 0 \mid X = x) \\
    &=\ f(x)\mu_0(x) + \{1-f(x)\}\{1 - \mu_0(x)\}\ ,
\end{align*}which allows us to write that
\begin{align*}
    &P_0\{f(X) = Y \mid X = x\} - P_0\{b_0(X) = Y \mid X = x\}\\
    &\hspace{0.8in}=\ \mu_0(x)\{f(x) - b_0(x)\} + \{1 - \mu_0(x)\}[\{1-f(x)\} - \{1-b_0(x)\}] \\
     &\hspace{0.8in}=\ \{2\mu_0(x) - 1\}\{f(x) - b_0(x)\}\ \leq\ 0
\end{align*}
by definition of $b_0$. It follows then that
\begin{align*}
	P_0\left\{f(X)=Y\right\}-P_0\left\{b_0(X)=Y\right\}\ &=\ E_0\left[ P_0\{f(X) = Y \mid X\}\right]- E_0\left[ P_0\{b_0(X) = Y \mid X\} \right]\\
    	&=\ E_0\left[P_0\{f(X) = Y \mid X \} - P_0\{b_0(X) = Y \mid X\} \right]\ \leq\ 0\ ,
\end{align*}so that $b_0$ is the maximizer of the classification accuracy $P_0\{Y = f(X)\}$.

\subsection{Conditional mean maximizes the area under the ROC curve}\label{sec:auc_maximizer_proof}

Suppose that $Y \in \{0, 1\}$ is a binary random variable. For a given function $f \in \mathcal{F}$, we define the conditional distribution functions \begin{align*}
F_{1}(P_0, f)(c):=P_0\left\{f(X) \leq c \mid Y = 1\right\} \text{\ \ and\ \ }F_{0}(P_0, f)(c):=P_0\left\{f(X) \leq c \mid Y = 0\right\}\ .
\end{align*} If $Y$ denotes the presence of a disease, then $1-F_1(P_0,f)(c)$ and $F_0(P_0,f)(c)$ denote the sensitivity and specificity of a medical test that flags the presence of disease if and only if $f(X)>c$. The AUC value corresponding to $f$ and $P_0$ can be written as
\begin{align*}
P_0\left\{f(X_1)<f(X_2)\mid Y_1=0,Y_2=1\right\}\ &=\ \int_0^\infty \left\{1 - F_{1}(P_0,f)(c)\right\}F_{0}(P_0,f)(dc) \\
&=\ \int_0^1 \left\{1 - F_{1}(P_0,f)(F_{0}^{-1}(P_0,f)(w)) \right\}dw\ .
\end{align*} For a fixed $w$, the integrand $1 - F_{1}(P_0,f)(F_{0}^{-1}(P_0,f)(w))$ is the sensitivity of a test based on $f$ and a cutoff that results in specificity $w$. By an application of the Neyman-Pearson Lemma, it is known that, for any fixed specificity level, any strictly increasing transformation of the likelihood ratio mapping $x\mapsto P_0\left(Y=1\mid X=x\right)/P_0\left(Y=0\mid X=x\right)=\mu_0(x)/\{1-\mu_0(x)\}$ gives an optimal choice of $f$. in particular, the function $f:x\mapsto \mu_0(x)$ is optimal. Since this is true irrespective of the fixed specificity level, it holds uniformly across specificity levels and hence also maximizes the AUC value, as claimed.

\subsection{Verification of conditions (A1) and (A2) for Examples 1--4}

\vspace{.15in}
\noindent\emph{Example 1: $R^2$}\\
We have that $|V(f,P_0)-V(f_0,P_0)|=E_0\{f(X)-f_0(X)\}^2/\sigma^2(P_0)$ so that $|V(f,P_0)-V(f_0,P_0)|=O(\|f-f_0\|^2_\mathcal{F})$ and condition (A1) holds. We can verify that $\dot{V}(f, P_0; h) = -\int \{y - f(x)\}^2 h(dz)/\sigma^2(P_0)$. Since $P\mapsto E_P\{Y-f(X)\}^2$ is linear and thus Hadamard differentiable uniformly in $f$, condition (A2) can be shown to hold for any $\delta>0$ provided the marginal distribution of $Y$ under $P_0$ has bounded support.

\vspace{.15in}
\noindent\emph{Example 2: deviance}\\
Using that $f_0=\mu_0$ and setting $a_0:=-2/\{\log P_0(Y=0)+\log P_0(Y=1)\}$, a standard argument based on Taylor approximations allows to write that \begin{align*}
    |V(f,P_0)-V(f_0,P_0)|\ &=\ a_0\left|E_0\left[f_0(X)\log\left\{\frac{f(x)}{f_0(x)}\right\}+\left\{1-f_0(x)\right\}\log\left\{\frac{1-f(x)}{1-f_0(x)}\right\} \right]\right|\\
    &\leq\ \frac{a_0}{2}E_0\left[\left\{f(x)-f_0(x)\right\}^2\left\{\frac{f_0(x)}{\xi_0(x)}+\frac{1-f_0(x)}{1-\xi_1(x)}\right\}\right]
\end{align*}
for some $\xi_0,\xi_1:\mathcal{X}\rightarrow \mathcal{Y}$ lying pointwise between $f$ and $f_0$. If $f(X),f_0(X)\in(\delta,1-\delta)$ almost surely under $P_0$, then we find that $|V(f,P_0)-V(f_0,P_0)|\leq a_0\left(\frac{1-\delta}{\delta}\right)\|f-f_0\|_{\mathcal{F}}^2$. Thus, condition (A1) then holds with $\alpha=2$. Since $P\mapsto E_P\left[Y\log f(X)+(1-Y)\log\{1-f(X)\}\right]$ is linear and thus Hadamard differentiable uniformly in $f$, condition (A2) can again be shown to hold for any $\delta>0$.

\vspace{.15in}
\noindent\emph{Example 3: classification accuracy}\\
Using that $f_0:x\mapsto I\{\mu_0(x)>1/2\}$ is an optimizer of accuracy, and writing any candidate prediction function $f: \mathcal{X} \to \{0,1\}$ as $f(x) = I\{\mu(x) > 1/2\}$ for some function $\mu : \mathcal{X} \to [0,1]$, we can write
\begin{align*}
    0\ &\leq\ P_0\left\{Y = f_0(X)\right\} - P_0\left\{Y = f(X)\right\}\ =\ E_0\left[I\left\{Y=f_0(X)\right\}-I\left\{Y=f(X)\right\}\right] \\
    &=\ P_0\left\{Y=f_0(X),Y\neq f(X)\right\}-P_0\left\{Y\neq f_0(X),Y=f(X)\right\}\\
    &=\ P_0\left\{f_0(X)=1,f(X)=0,Y=1\right\}+P_0\left\{f_0(X)=0,f(X)=1,Y=0\right\}\\
    &\hspace{.3in}-P_0\left\{f_0(X)=0,f(X)=1,Y=1\right\}-P_0\left\{f_0(X)=1,f(X)=0,Y=0\right\}\\
    &=\ [P_0\{Y=1\mid \mu_0(X)\geq \tfrac{1}{2}>\mu(X)\}-P_0\{Y=0\mid \mu_0(X)\geq \tfrac{1}{2}>\mu(X)\}]\,P_0\{\mu_0(X)\geq \tfrac{1}{2}>\mu(X)\}\\
    &\hspace{.3in}+[P_0\{Y=0\mid \mu(X)\geq \tfrac{1}{2}>\mu_0(X)\}-P_0\{Y=1\mid \mu(X)\geq \tfrac{1}{2}>\mu_0(X)\}]\,P_0\{\mu(X)\geq \tfrac{1}{2}>\mu_0(X)\}\\
    &=\ [2P_0\{Y=1\mid \mu_0(X)\geq\tfrac{1}{2}>\mu(X)\}-1]\,P\{\mu_0(X)\geq\tfrac{1}{2}>\mu(X)\}\\
    &\hspace{.3in}+[2P_0\{Y=0\mid \mu(X)\geq\tfrac{1}{2}>\mu_0(X)\}-1]\,P\{\mu(X)\geq\tfrac{1}{2}>\mu_0(X)\}\ .
 \end{align*} Now, on one hand, we note that \begin{align*}
 &P_0\{Y=1\mid \mu_0(X)\geq\tfrac{1}{2}>\mu(X)\}-\tfrac{1}{2}\ =\ E_0\{Y\mid \mu_0(X)\geq\tfrac{1}{2}>\mu(X)\}-\tfrac{1}{2}\\
 &\hspace{0.5in}=\ E_0\{\mu_0(X)\mid \mu_0(X)\geq\tfrac{1}{2}>\mu(X)\}-\tfrac{1}{2}\ =\ E_0\{\mu_0(X)-\tfrac{1}{2}\mid \mu_0(X)\geq\tfrac{1}{2}>\mu(X)\}\ ,
 \end{align*} and so it follows that $|P_0\{Y=1\mid \mu_0(X)\geq\tfrac{1}{2}>\mu(X)\}-\tfrac{1}{2}|\leq\|\mu-\mu_0\|_\infty$. We can similarly show that $|P_0\{Y=0\mid \mu(X)\geq\tfrac{1}{2}>\mu_0(X)\}-\tfrac{1}{2}|\leq\|\mu-\mu_0\|_\infty$. On the other hand, in view of the margin condition we impose, we have that \[P_0\{\mu_0(X)\geq\tfrac{1}{2}>\mu(X)\}\ \leq\ P_0\{|\mu_0(X)-\tfrac{1}{2}|<|\mu(X)-\mu_0(X)|\}\ \leq\ \kappa\|\mu-\mu_0\|_{\infty}\]and similarly, $P_0\{\mu(X)\geq\tfrac{1}{2}>\mu_0(X)\}\leq \kappa\|\mu-\mu_0\|_{\infty}$. Combining the inequalities we have derived, we conclude that $0\leq P_0\{Y=f_0(X)\}-P_0\{Y=f(X)\}\leq 4\kappa\|\mu-\mu_0\|_\infty$.

 \vspace{.15in}
\noindent\emph{Example 4: Area under the ROC curve}\\
We begin by writing
\begin{align*}
    0\ &\leq\ P_0\left\{f_0(X_1)<f_0(X_2),Y_1=0,Y_2=1\right\}-P_0\left\{f(X_1)<f(X_2),Y_1=0,Y_2=1\right\}\\
    &=\ E_0\left[I\left\{f_0(X_1)<f_0(X_2),Y_1=0,Y_2=1\right\}-I\left\{f(X_1)<f(X_2),Y_1=0,Y_2=1\right\}\right]\\
    &=\ \tfrac{1}{2}\,E_0\left[I\left\{f_0(X_1)<f_0(X_2),Y_1=0,Y_2=1\right\}+I\left\{f_0(X_1)\geq f_0(X_2),Y_1=1,Y_2=0\right\}\right]\\
    &\hspace{.5in}-\tfrac{1}{2}\,E_0\left[I\left\{f(X_1)<f(X_2),Y_1=0,Y_2=1\right\}+I\left\{f(X_1)\geq f(X_2),Y_1=1,Y_2=0\right\}\right]\\
    &=\ \tfrac{1}{2}\,E_0\left[(Y_2-Y_1)I\left\{f_0(X_1)<f_0(X_2),f(X_1)\geq f(X_2)\right\}\right]\\
    &\hspace{.5in}+\tfrac{1}{2}\,E_0\left[(Y_1-Y_2)I\left\{f_0(X_1)\geq f_0(X_2),f(X_1)< f(X_2)\right\}\right]\\
    &=\ \tfrac{1}{2}\,E_0\left[\{f_0(X_2)-f_0(X_1)\}I\left\{f_0(X_1)<f_0(X_2),f(X_1)\geq f(X_2)\right\}\right]\\
    &\hspace{.5in}+\tfrac{1}{2}\,E_0\left[\{f_0(X_1)-f_0(X_2)\}I\left\{f_0(X_1)\geq f_0(X_2),f(X_1)< f(X_2)\right\}\right]\\
    &\leq\ \tfrac{1}{2}\, E_0\left[|f_0(X_1)-f_0(X_2)|I\left\{[f_0(X_1)-f_0(X_2)][f(X_1)-f(X_2)]<0\right\} \right].
\end{align*} Defining $A:=\{f(X_1)-f_0(X_1)\}+\{f_0(X_2)-f(X_2)\}$, $B:=f_0(X_1)-f_0(X_2)$ and $t:x\mapsto |f(x)-f_0(x)|$, we note that
\begin{align*}
    &\{[f_0(X_1) - f_0(X_2)][f(X_1) - f(X_2)] < 0\}\ = \ \{B(A + B) < 0\}\ =\ \{(\tfrac{1}{2}A+B)^2 - \tfrac{1}{4}A^2 < 0\} \\
    &=\ \{\lvert A \rvert > \lvert B \rvert, AB < 0\}\ \subseteq\ \{\lvert A \rvert > \lvert B \rvert \}\ \subseteq\ \{\lvert f_0(X_1) - f_0(X_2) \rvert < t(X_1) + t(X_2)\}\ .
\end{align*} Using this result and the inequality derived above, and defining $\alpha_0:=\{P_0(Y=1)P_0(Y=0)\}^{-1}$, we have that
\begin{align*}
0\ &\leq\ AUC(f_0,P_0)-AUC(f,P_0)\\
&=\ \alpha_0\,\left[P_0\left\{f_0(X_1)<f_0(X_2),Y_1=0,Y_2=1\right\}-P_0\left\{f(X_1)<f(X_2),Y_1=0,Y_2=1\right\}\right]\\
&\leq\ \tfrac{1}{2}\alpha_0\,E_0\left[|f_0(X_1)-f_0(X_2)|I\left\{[f_0(X_1)-f_0(X_2)][f(X_1)-f(X_2)]<0\right\} \right]\\
&\leq\ \tfrac{1}{2}\alpha_0\,E_0\left[|f_0(X_1)-f_0(X_2)|I\left\{|f_0(X_1)-f_0(X_2)|<t(X_1)+t(X_2)\right\} \right]\\
&\leq\ \tfrac{1}{2}\alpha_0\,E_0\left[|f_0(X_1)-f_0(X_2)|I\left\{|f_0(X_1)-f_0(X_2)|<2\|t\|_{\infty}\right\} \right]\\
    & \leq\ \alpha_0\,\lVert t \rVert_\infty\, P_0\left\{\lvert f_0(X_1) - f_0(X_2) \rvert < 2 \lVert t \rVert_\infty\right\}\ \leq\ 2\alpha_0\, \kappa\, \lVert t \rVert_\infty^2\ ,
\end{align*} where the last inequality follows from the margin condition we impose.

\subsection{Derivation of the EIFs for Examples 5 and 6}

\noindent\emph{Example 5: Mean outcome under a binary intervention rule}

The nonparametric EIF for this example is derived in, for example, Sections 2 and 3 of \cite{luedtke2016} and in Section A.1 of its supplement.

\vspace{.15in}
\noindent\emph{Example 6: Classification accuracy under outcome missingness}

Recall that, in this example, the ideal-data structure consists of $\mathbbs{Z}:=(X,Y)\sim\mathbbs{P}$, and the observed data structure is  $Z:=(X,\Delta,U)$, where $\Delta$ is the indicator of having observed the outcome $Y$, and we have defined $U:=\Delta Y$. The ideal-data nonparametric EIF at $\mathbbs{P}$, following Appendix~\ref{sec:app-a}, is given by
\begin{align*}
    \phi_{\mathbbs{P}}^F(x,y) = I\{y = f_\mathbbs{P}(x)\} - V(f_\mathbbs{P}, \mathbbs{P}).
\end{align*}
Based on results in Chapter 25.5.3 of \cite{vandervaart2000}, the observed-data nonparametric EIF at $P$ is given by
\begin{align}\label{eq:supp-e6}
    \phi_P(z) = \frac{\delta}{g_P(x)}\phi_{\mathbbs{P}}^F(z) + \left\{1 - \frac{\delta}{g_P(x)}\right\}E_P\{\phi_{\mathbbs{P}}^F(Z) \mid \Delta = 1, X = x\}\ .
\end{align}
Defining the nuisance function $Q_P(x) := P\{Y = f_P(X) \mid \Delta = 1, X = x\}$, simple algebraic manipulations then yield that
$E_P\{\phi_\mathbbs{P}^F(Z) \mid \Delta = 1, X = x\} = \ Q_P(x) - V(f_P, P)$.
Plugging this into \eqref{eq:supp-e6} yields the desired form of the EIF.

\section{Additional numerical experiments}\label{sec:more_sims}
\subsection{Replicating all numerical experiments}\label{sec:rep}

All numerical experiments presented here and in the main manuscript can be replicated using code available on GitHub.
% \href{https://github.com/bdwilliamson/vimp_supplementary}{on GitHub}.
\revision{In all cases, we generate data by:
\begin{align*}
&1:\text{ drawing\ }X \sim MVN(0, \Sigma)\\
&2:\text{ drawing }\epsilon\sim N(0,1)\text{ independent of }X,\text{ and setting }Y =I\{x\beta_0 + \epsilon > 0\}\ \text{given }X=x,
\end{align*}where  $\Sigma$ is the $p \times p$ identity matrix and $\beta_0 = (2.5, 3.5, 0,\ldots,0)^\top$. The dimension $p$ is determined by the scenario. The approximate true values of variable importance based on accuracy and AUC under all scenarios considered here are provided in Table~\ref{tab:supp-truths}. The specification of each individual algorithm for estimating $f_0$ and $f_{0,s}$ is provided in Table~\ref{tab:supp-individual-algs}, while the specification of the candidate algorithms used in the Super Learner is provided in Table~\ref{tab:supp-sl-algs}.}

\begin{table}
\centering
\caption{Approximate values of $\psi_{0,s}$ in the numerical experiments.}
\label{tab:supp-truths}
\begin{tabular}{|ll|cccccc|}
    \hline
Importance measure & Scenario & $X_1$ & $X_2$ & $X_3$ & $X_4$ & $(X_1, X_3)$ & $(X_2, X_4)$ \\
\hline
\multirow{2}*{Accuracy} & (1,2,3) & 0.136 & 0.236 & 0 & 0 & 0.136 & 0.236 \\
 & 4 & 0.081 & 0.228 & 0 & 0 & 0.136 & 0.236 \\
\multirow{2}*{Area under the ROC curve} & (1,2,3) & 0.105 & 0.221 & 0 & 0 & 0.105 & 0.221 \\
 & 4 & 0.052 & 0.211 & 0 & 0 & 0.105 & 0.221 \\
\hline
\end{tabular}
\end{table}

\begin{table}
    \centering
    \begin{tabular}{c|ccc}
      Algorithm & R & Tuning Parameter(s) & Tuning parameter  \\
       & Implementation & and possible values & description\\ \hline
        Generalized linear models & \texttt{glm} & -- & -- \\ \hline
        Generalized additive models & \texttt{mgcv} & method = \texttt{"GCV.Cp"} & Smoothing parameter \\
        & \citep{mgcvpkg} & & estimation method \\ \hline
        Random forests & \texttt{ranger} & \texttt{ntree}$^{\ddagger}$ & Number of variables \\
        & \citep{rangerpkg} &  & to possibly split \\
        & & & at in each node \\
        & & \texttt{max.depth}$^{\ddagger}$ & Maximum tree depth \\
        & & \texttt{min.node.size}$^{\ddagger}$ & Minimum node size \\ \hline
    \end{tabular}
    \caption{Individual algorithms considered with their R implementation, tuning parameter values, and description of the tuning parameters. All tuning parameters besides those listed here are set to their default values. In particular, the random forests are grown with 500 trees, \texttt{mtry} = $\sqrt{p}$ ${}^{\dagger}$, and a subsampling fraction of 1; five-fold cross-validation over the grid defined by (\texttt{ntree}, \texttt{max.depth}, \texttt{min.node.size}) was used to select the tuning parameter combination that minimized log-likelihood loss. \\
    ${}^{\dagger}$: $p$ denotes the total number of predictors. \\
    ${}^{\ddagger}$: For setting 1, \texttt{ntree} $\in \{100, 500, 1000\}$, \texttt{max.depth} = 5, \texttt{min.node.size} = 1; for all other settings, \texttt{ntree} $\in \{500, 1000, 1500, 2000, 5000\}$, \texttt{max.depth} $\in \{1, 3, 5\}$, \texttt{min.node.size} = 10.}
    \label{tab:supp-individual-algs}
\end{table}

\begin{table}
    \centering
    \begin{tabular}{c|ccc}
       Candidate Learner & R & Tuning Parameter & Tuning parameter  \\
       & Implementation & and possible values & description\\ \hline
        Generalized linear models & \texttt{glm} & -- & -- \\ \hline
        Generalized additive models & \texttt{gam} & degree $= 2$ & Degree of smooth terms \\
        & \citep{gampkg} & & \\ \hline
        Random forests & \texttt{ranger} & \texttt{mtry} $= \sqrt{p}$ ${}^{\dagger}$ & Number of variables \\
        & \citep{rangerpkg} & & to possibly split \\
        & & & at in each node \\ \hline
        Gradient boosted & \texttt{xgboost} & \texttt{max.depth} $ = 1$ &  Maximum tree depth\\
        trees & \citep{xgboostpkg} & & \\ \hline
        Elastic net$^{\ddagger}$ & \texttt{glmnet} & mixing parameter $\alpha$ & Trade-off between  \\
        & \citep{glmnetpkg} & $ = 1$ & $\ell_1$ and $\ell_2$ regularization \\ \hline
    \end{tabular}
    \caption{Candidate learners in the Super Learner ensemble along with their R implementation, tuning parameter values, and description of the tuning parameters. All tuning parameters besides those listed here are set to their default values. In particular, the random forests are grown with 500 trees, a minimum node size of 5 for continuous outcomes and 1 for binary outcomes, and a subsampling fraction of 1; the boosted trees are grown with a maximum of 1000 trees, shrinkage rate of 0.1, and a minimum of 10 observations per node; and the lasso $\ell_1$ tuning parameter is chosen using 10-fold cross-validation. \\
    ${}^{\dagger}$: $p$ denotes the total number of predictors. \\
    ${}^{\ddagger}$: lasso is only included in cases where $p \geq 4$.}
    \label{tab:supp-sl-algs}
\end{table}

\subsection{Properties of our proposal under the alternative hypothesis}\label{sec:cf-alt}

\revision{In this section, we present additional results under Scenario 1. In this case, $p = 2$. For each scenario presented here, we generated 1000 random datasets of size $n \in \{100, 500, 1000, \dots, 4000\}$, and considered the importance of both $X_1$ and $X_2$. We highlight results for both features using the AUC and for $X_1$ using accuracy, and we provide the coverage of nominal 95\% confidence intervals. We assess performance in the same way as in the main manuscript.}

\revision{We present results for AUC and for the accuracy-based importance of $X_1$ in Figures~\ref{fig:supp_alt_accuracy_1}--\ref{fig:supp_alt_auc_2}. The results for both features and both importance measures are largely similar to those presented in Section 5.2 of the main manuscript. The need for cross-fitting is particularly striking in Figure~\ref{fig:supp_alt_auc_2}, where we observed coverage near zero for intervals based on a non-cross-fitted random forest estimator of the oracle prediction functions. In Figure~\ref{fig:supp_alt_nocfse}, we show the coverage of nominal 95\% intervals based on the non-cross-fitted standard error estimator. Here, we observe reduced coverage in some cases compared to the results presented above. Taken together, these results highlight that} when using simple estimators of the conditional mean functions (e.g., estimators based on correctly-specified parametric models), using cross-fitting appears to have minimal impact on the performance of the proposed inferential procedures and is therefore not needed. In contrast, when flexible nuisance estimators are used, it appears important to use cross-fitting when estimating VIM values \revision{and standard errors}. The elimination of the constraint on nuisance estimator complexity (i.e., the Donsker class condition) achieved via cross-fitting does appear to translate into substantially improved practical performance when complex nuisance estimators are used.

\begin{figure}
\centering
\includegraphics[width = 1\textwidth]{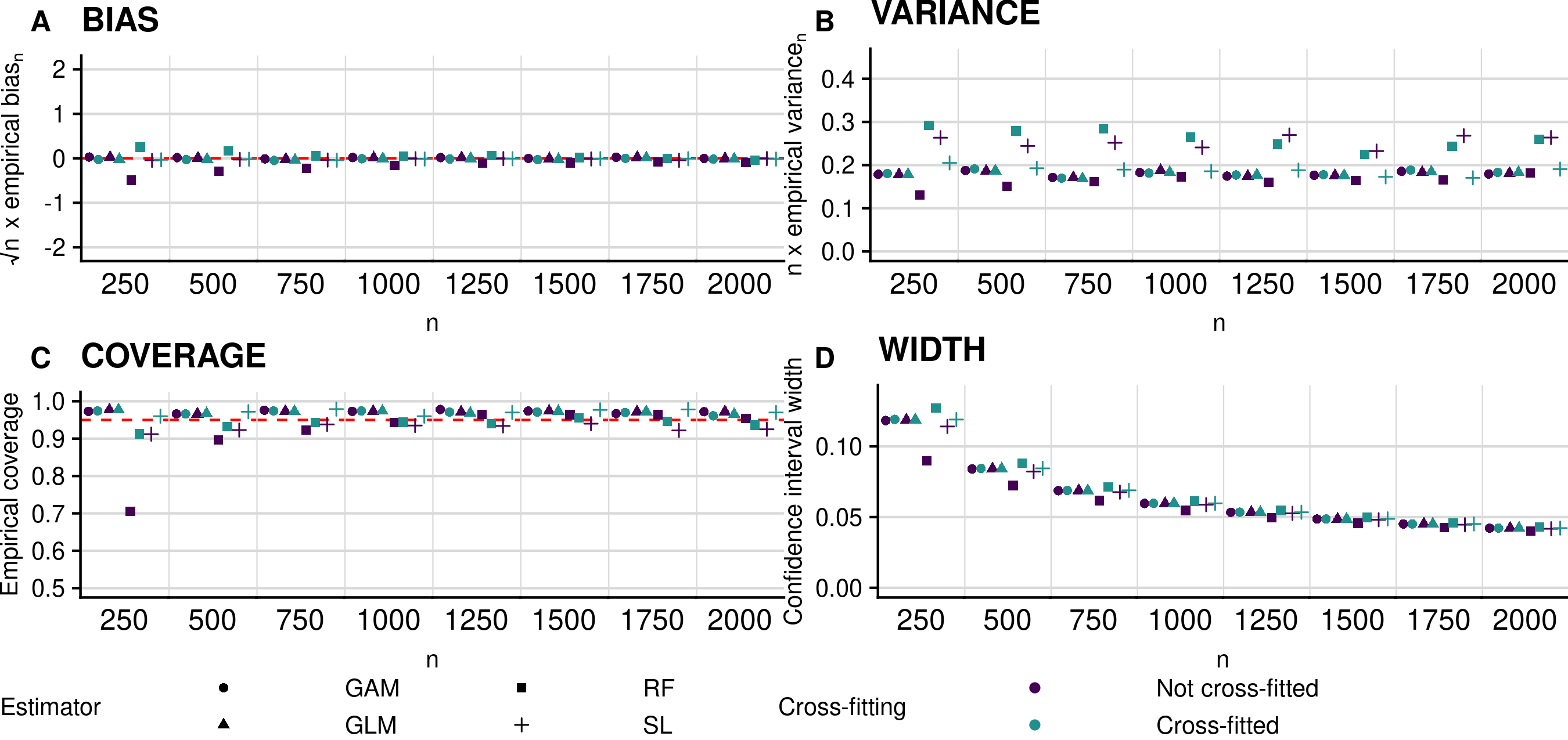}
\caption{Performance of plug-in estimators for estimating (non-zero) importance of $X_1$ in terms of accuracy under Scenario 1 (all features have non-zero importance). Clockwise from top left: empirical bias of the proposed plug-in estimator scaled by $n^{1/2}$; empirical variance scaled by $n$; empirical coverage of nominal 95\% confidence intervals; and width of these intervals. Circles, triangles, squares, and plus symbols denote estimators based on the use of generalized additive models (GAMs), probit regression (GLM), random forests (RF) or the Super Learner (SL), respectively. Blue and green symbols denote non-cross-fitted and cross-fitted estimators, respectively. This figure appears in color in the electronic version of this article.}
\label{fig:supp_alt_accuracy_1}
\end{figure}

\begin{figure}
\centering
\includegraphics[width = 1\textwidth]{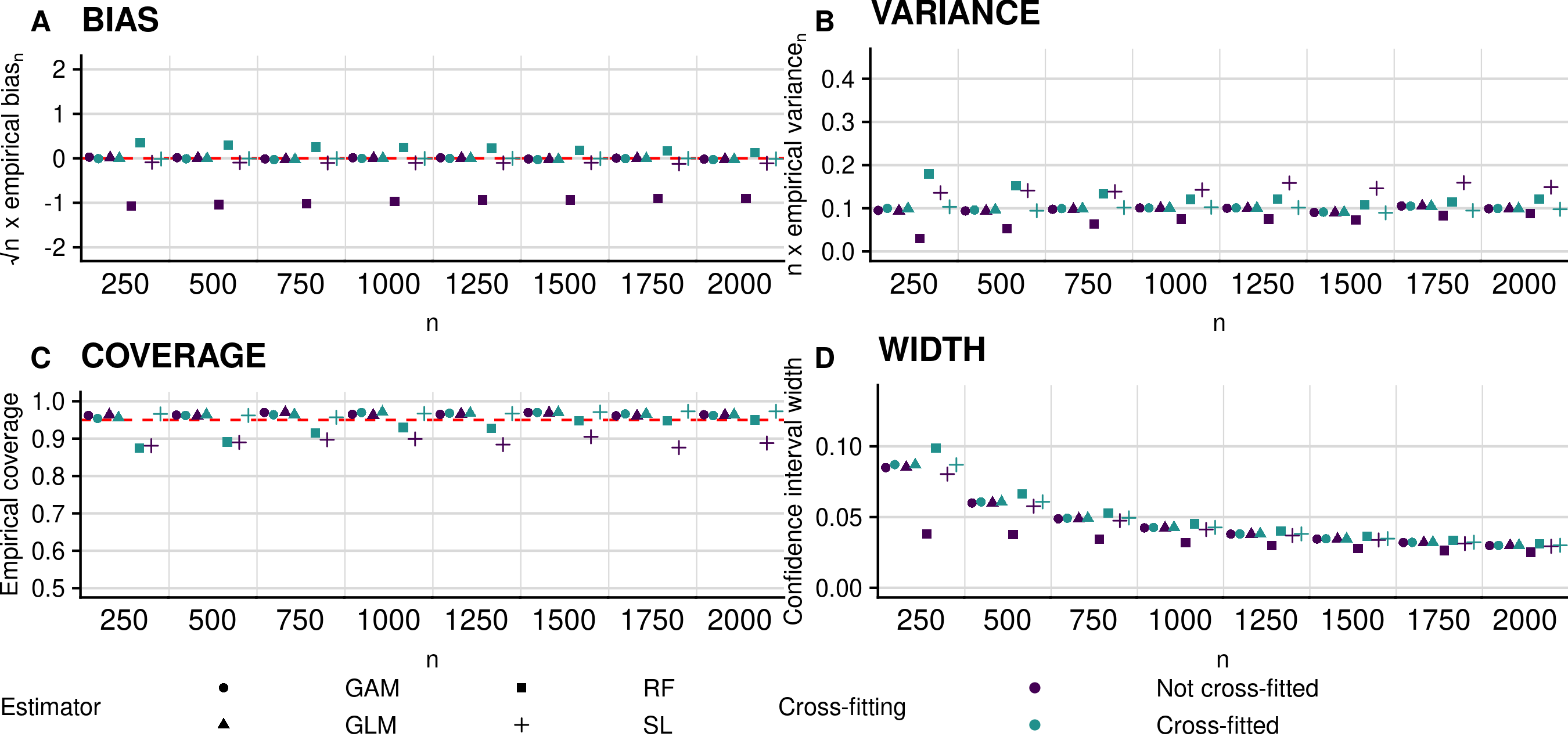}
\caption{Performance of plug-in estimators for estimating (non-zero) importance of $X_1$ in terms of AUC under Scenario 1 (all features have non-zero importance). Clockwise from top left: empirical bias of the proposed plug-in estimator scaled by $n^{1/2}$; empirical variance scaled by $n$; empirical coverage of nominal 95\% confidence intervals; and width of these intervals. Circles, triangles, squares, and plus symbols denote estimators based on the use of generalized additive models (GAMs), probit regression (GLM), random forests (RF) or the Super Learner (SL), respectively. Blue and green symbols denote non-cross-fitted and cross-fitted estimators, respectively. Coverage of intervals based on the non-cross-fitted RF-based estimator never exceeds 0.5 and is as low as zero in some cases. This figure appears in color in the electronic version of this article.}
\label{fig:supp_alt_auc_1}
\end{figure}

\begin{figure}
\centering
\includegraphics[width = 1\textwidth]{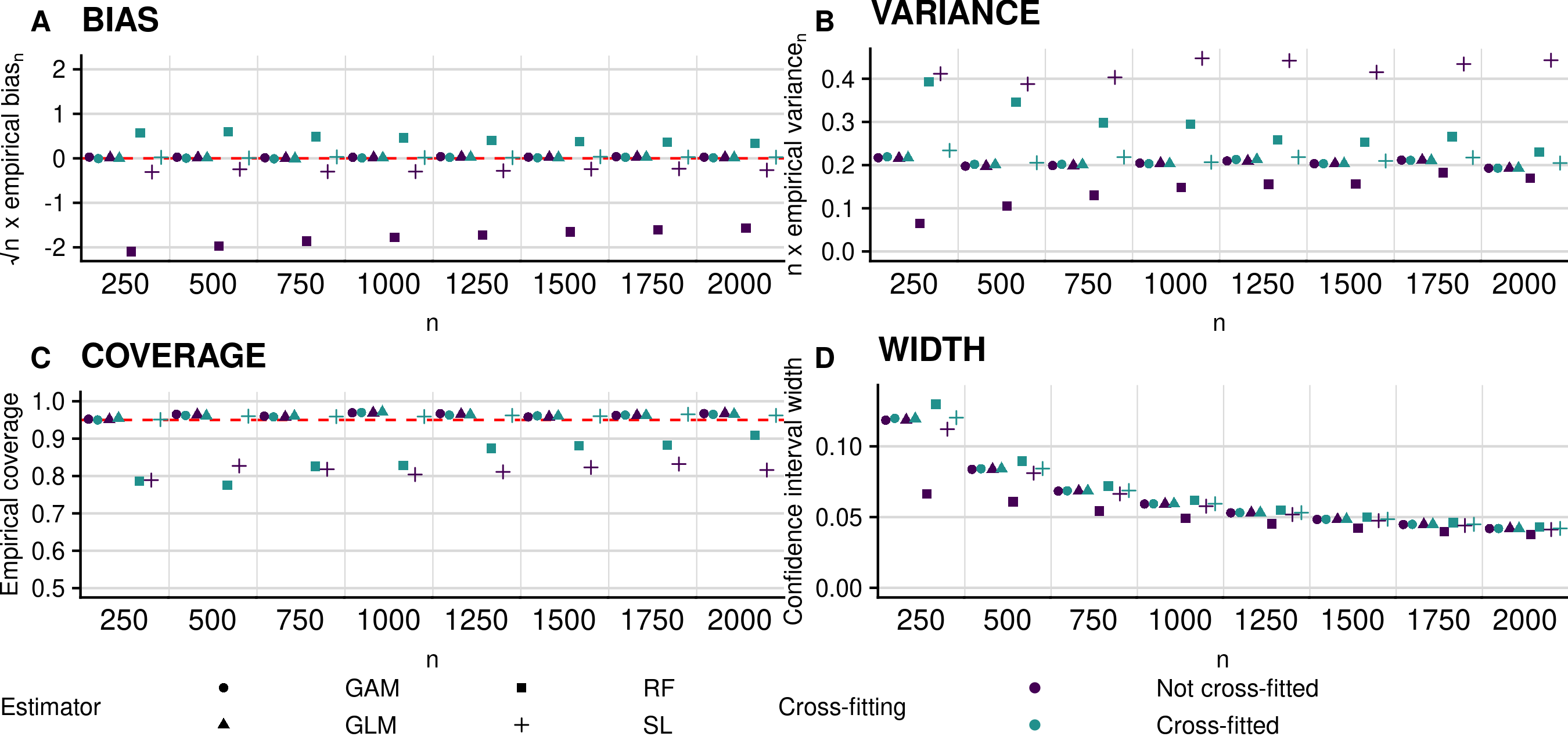}
\caption{Performance of plug-in estimators for estimating (non-zero) importance of $X_2$ in terms of AUC under Scenario 1 (all features have non-zero importance). Clockwise from top left: empirical bias of the proposed plug-in estimator scaled by $n^{1/2}$; empirical variance scaled by $n$; empirical coverage of nominal 95\% confidence intervals; and width of these intervals. Circles, triangles, squares, and plus symbols denote estimators based on the use of generalized additive models (GAMs), probit regression (GLM), random forests (RF) or the Super Learner (SL), respectively. Blue and green symbols denote non-cross-fitted and cross-fitted estimators, respectively. Coverage of intervals based on the non-cross-fitted RF-based estimator never exceeds 0.5 and is as low as zero in some cases. This figure appears in color in the electronic version of this article.}
\label{fig:supp_alt_auc_2}
\end{figure}

\begin{figure}
\centering
\includegraphics[width = 1\textwidth]{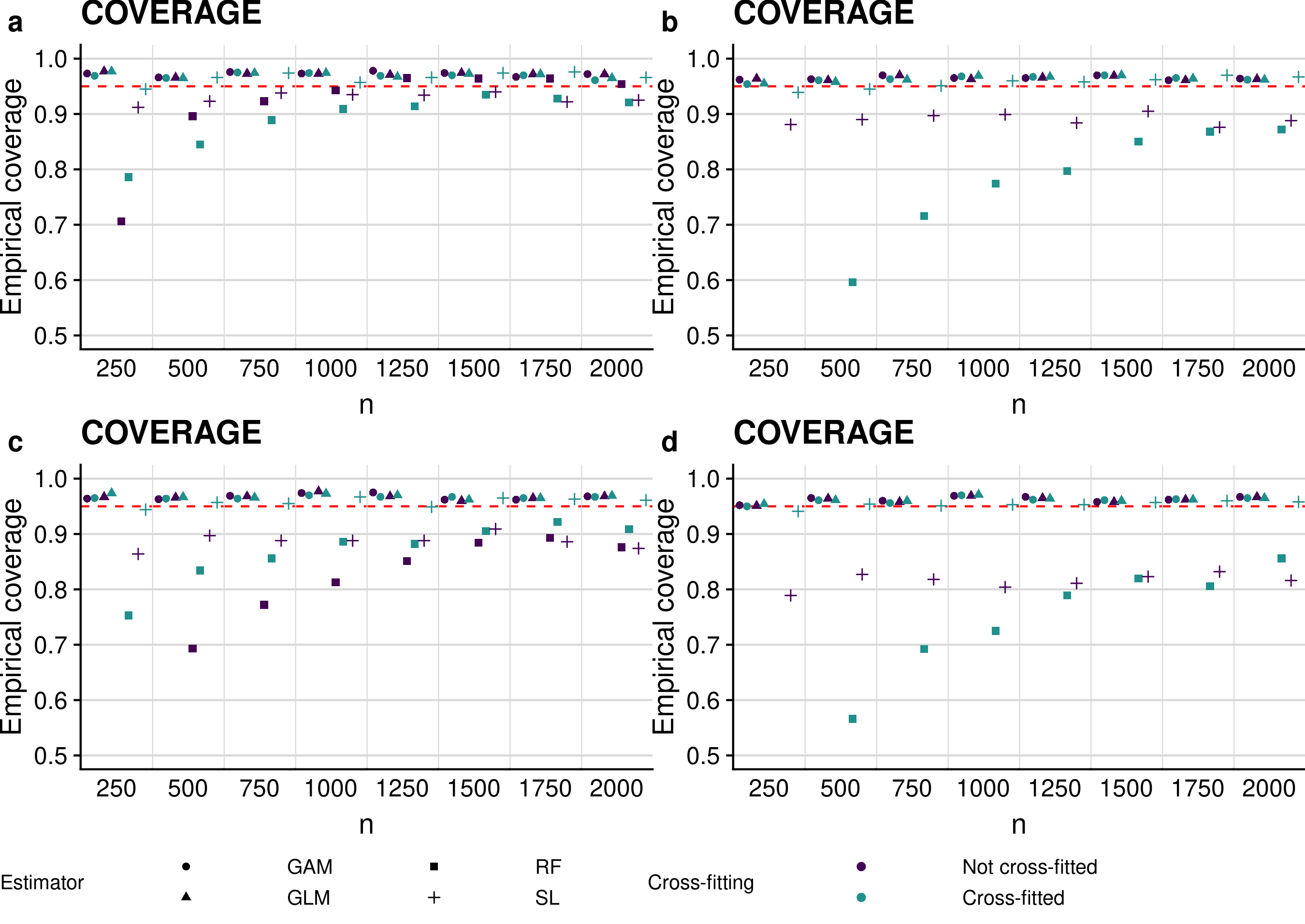}
\caption{Empirical coverage of confidence intervals based on the non-cross-fitted standard error estimator under Scenario 1 (all features have non-zero importance). The rows correspond to the feature of interest, while the columns correspond to accuracy and AUC, respectively. Circles, triangles, squares, and plus symbols denote estimators based on the use of generalized additive models (GAMs), probit regression (GLM), random forests (RF) or the Super Learner (SL), respectively. Blue and green symbols denote non-cross-fitted and cross-fitted VIM estimators, respectively. Coverage of intervals based on the non-cross-fitted RF-based estimator (panels b and d) never exceeds 0.5 and is as low as zero in some cases. This figure appears in color in the electronic version of this article.}
\label{fig:supp_alt_nocfse}
\end{figure}

\subsection{Properties of our proposal under the null hypothesis}\label{sec:cv-null}

In this section, we present additional results under Scenario 2. In this case, $p = 4$. We again generated 1000 random datasets of size $n \in \{100, 500, 1000, \dots, 4000\}$, and \revision{considered the importance of both $X_2$ (a non-null feature) and $X_3$ (a null feature)}. Here, we highlight results for both features based on the AUC and for $X_2$ based on accuracy, \revision{and we provide the coverage of nominal 95\% confidence intervals and proportion of tests rejected}. We assess performance in the same way as in the main manuscript.

\revision{We present the results based on a cross-fitted standard error estimator in Figures~\ref{fig:supp_null_accuracy_2}--\ref{fig:supp_null_auc_3}. In Figures~\ref{fig:supp_null_accuracy_2} and \ref{fig:supp_null_auc_2}, we observe high power across all sample sizes. We again observe residual bias for the non-cross-fitted VIM estimators based on flexible nuisance estimation (random forests and the Super Learner). In Figure~\ref{fig:supp_null_auc_3}, the cross-fitted VIM estimator based on random forests exhibits some residual bias but coverage and type I error are still near the nominal level. It is possible that this bias could be mitigated with cross-validation over a richer grid of tuning parameters. Similarly as in the main manuscript, since the bias for estimating the null feature appears to be small for the non-cross-fitted estimators, type I error is not inflated in these simulations. However, we expect in most cases that cross-fitting will yield a more adequate type I error control. Indeed, we see that this is the case by comparing the results for the cross-fitted estimator and  cross-fitted versus non-cross-fitted standard error estimators (Figure~\ref{fig:supp_null_nocfse}). Here, we see a vastly inflated type I error for the cross-fitted random forests-based estimator, reflecting that in this case the non-cross-fitted standard error appears to be too small.}

\begin{figure}
\centering
\includegraphics[width = 1\textwidth]{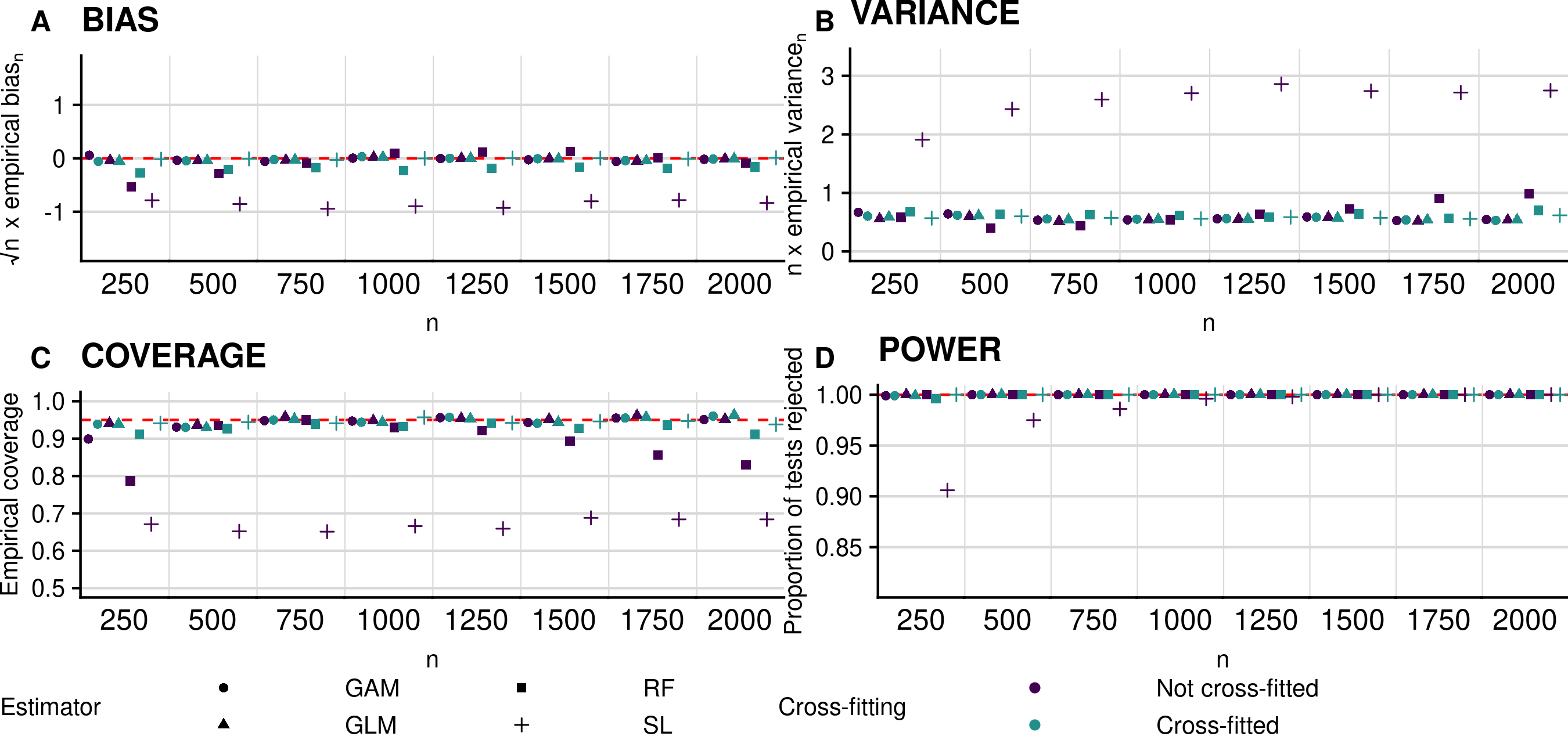}
\caption{Performance of plug-in estimators for estimating (non-zero) importance of $X_2$ in terms of accuracy under Scenario 2. Clockwise from top left: empirical bias of the proposed plug-in estimator scaled by $n^{1/2}$; empirical variance scaled by $n$; empirical coverage of nominal 95\% confidence intervals; and empirical power of the proposed hypothesis test. Circles, triangles, squares, and plus symbols denote estimators based on the use of generalized additive models (GAMs), probit regression (GLM), random forests (RF) or the Super Learner (SL), respectively. Blue and green symbols denote non-cross-fitted and cross-fitted estimators, respectively. This figure appears in color in the electronic version of this article.}
\label{fig:supp_null_accuracy_2}
\end{figure}

\begin{figure}
\centering
\includegraphics[width = 1\textwidth]{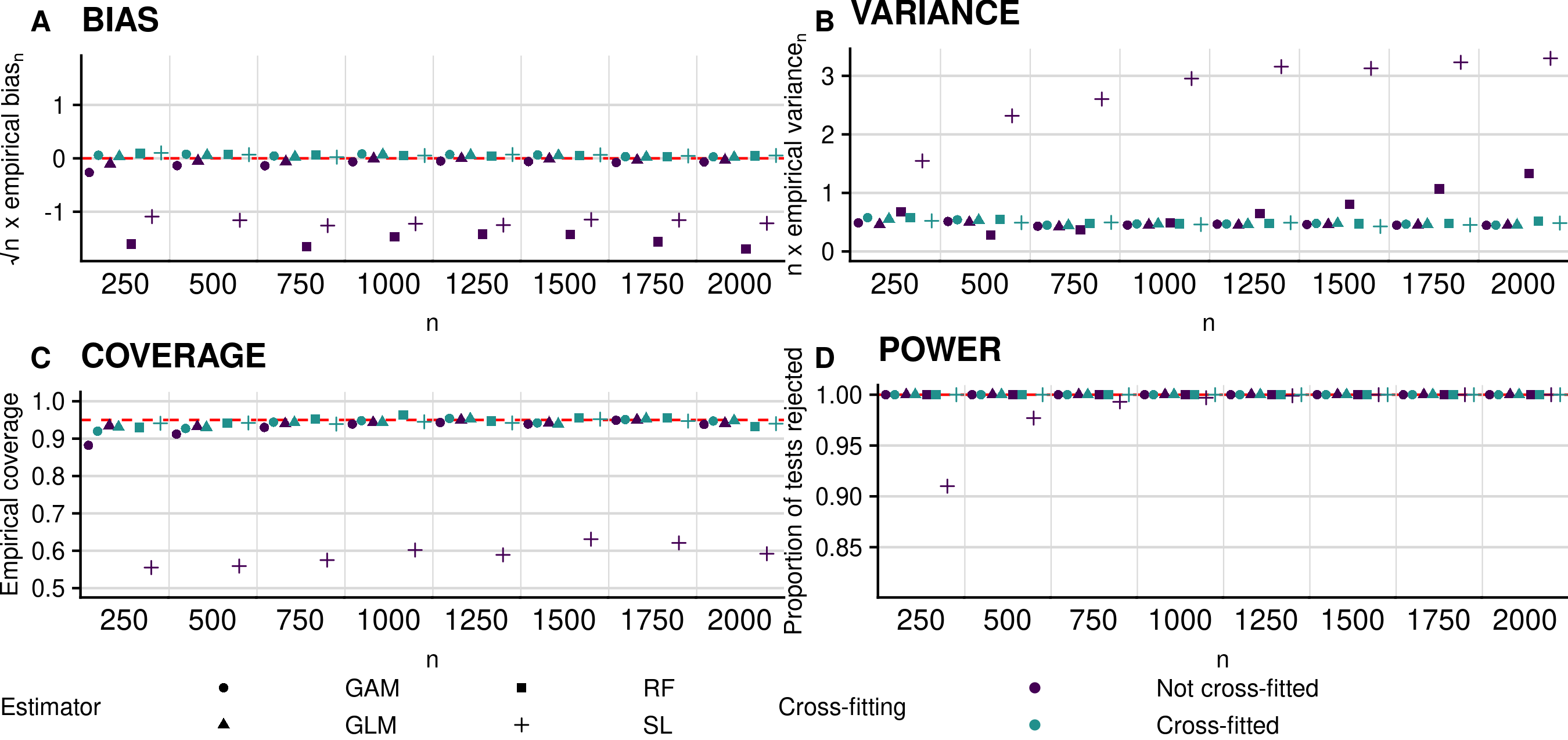}
\caption{Performance of plug-in estimators for estimating (non-zero) importance of $X_2$ in terms of AUC under Scenario 2. Clockwise from top left: empirical bias of the proposed plug-in estimator scaled by $n^{1/2}$; empirical variance scaled by $n$; empirical coverage of nominal 95\% confidence intervals; and empirical power of the proposed hypothesis test. Circles, triangles, squares, and plus symbols denote estimators based on the use of generalized additive models (GAMs), probit regression (GLM), random forests (RF) or the Super Learner (SL), respectively. Blue and green symbols denote non-cross-fitted and cross-fitted estimators, respectively. Coverage of intervals based on the non-cross-fitted RF-based estimator never exceeds 0.5 and is as low as zero in some cases. This figure appears in color in the electronic version of this article.}
\label{fig:supp_null_auc_2}
\end{figure}

\begin{figure}
\centering
\includegraphics[width = 1\textwidth]{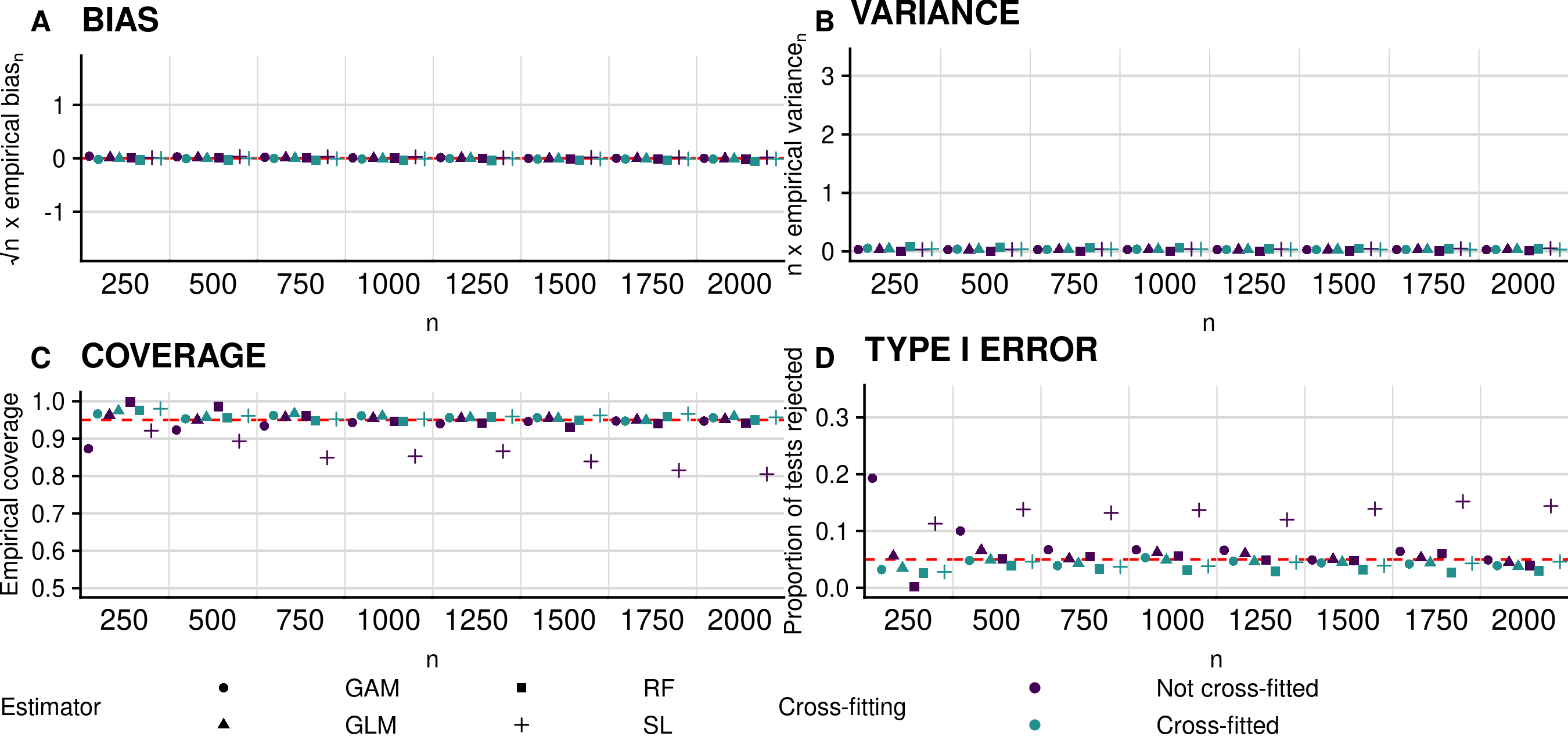}
\caption{Performance of plug-in estimators for estimating (zero) importance of $X_3$ in terms of AUC under Scenario 2. Clockwise from top left: empirical bias of the proposed plug-in estimator scaled by $n^{1/2}$; empirical variance scaled by $n$; empirical coverage of nominal 95\% confidence intervals; and empirical type I error of the proposed hypothesis test. Circles, triangles, squares, and plus symbols denote estimators based on the use of generalized additive models (GAMs), probit regression (GLM), random forests (RF) or the Super Learner (SL), respectively. Blue and green symbols denote non-cross-fitted and cross-fitted estimators, respectively. This figure appears in color in the electronic version of this article.}
\label{fig:supp_null_auc_3}
\end{figure}

\begin{figure}
\centering
\includegraphics[width = 1\textwidth]{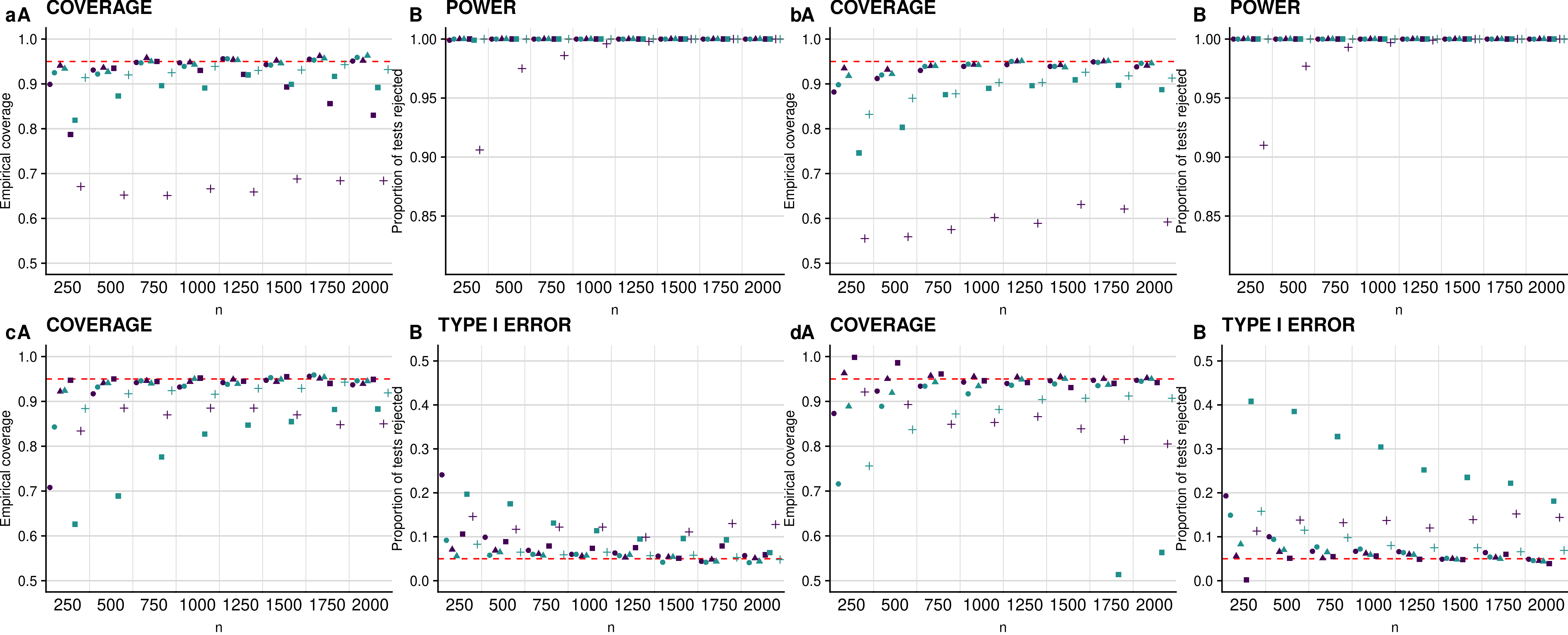}
\caption{Empirical coverage of confidence intervals (A) and proportion of tests rejected (B) based on the non-cross-fitted standard error estimator under Scenario 2. The rows correspond to $X_2$ and $X_3$, respectively, while the columns correspond to accuracy (a,c) and AUC (b,d), respectively. Circles, triangles, squares, and plus symbols denote estimators based on the use of generalized additive models (GAMs), probit regression (GLM), random forests (RF) or the Super Learner (SL), respectively. Blue and green symbols denote non-cross-fitted and cross-fitted estimators, respectively.  Coverage of intervals based on the non-cross-fitted RF-based estimator of importance of $X_2$ (panel bA) never exceeds 0.5 and is as low as zero in some cases. This figure appears in color in the electronic version of this article.}
\label{fig:supp_null_nocfse}
\end{figure}

\subsection{Using the bootstrap for interval estimation}\label{sec:boot}

\revision{In some cases,  particularly those with limited sample sizes, it may be of interest to use a bootstrap scheme for interval estimation rather than a Wald construction using an influence function-based estimator of the asymptotic variance. Because estimation of $f_0$ and $f_{0,s}$ only contributes to the second-order behavior of the plug-in VIM estimator, a valid nonparametric bootstrap here would consist of bootstrapping the empirical distribution $P_n$ but fixing the nuisance estimators $f_n$ and $f_{n,s}$ across all bootstrap runs. Not having  to re-fit estimators of the nuisance functions on each bootstrap sample makes this scheme particularly efficient to implement. Additionally, since we only use the bootstrap for  interval estimation, we do not need to bootstrap the cross-fitting procedure. Our proposed bootstrap procedure in a case with no sample-splitting (i.e., under the alternative hypothesis) is as follows:
\begin{enumerate}
    \item obtain estimator $\psi_{n,s}$ or $\psi_{n,s}^*$ of $\psi_{0,s}$;
    \item obtain estimators $f_{n}$ and $f_{n,s}$ of $f_0$ and $f_{0,s}$ based on the entire dataset;
    \item create $B$ bootstrap resamples of the original dataset;
    \item For $b = 1,2, \ldots, B$:
    \begin{enumerate}
        \item obtain $v_{n,b} := V(f_n, P_{n,b})$ and $v_{n,s,b} := V(f_{n,s}, P_{n,b})$ using the nuisance functions estimated on the entire dataset and the bootstrap empirical distribution $P_{n,b}$;
        \item set $\psi_{n,s,b} := v_{n,b} - v_{n,s,b}$;
    \end{enumerate}
    \item compute bootstrap variance estimator $\tau^2_{n,s,B} := \frac{1}{B}\sum_{b=1}^B \left(\psi_{n,s,b} - \frac{1}{B}\sum_{b=1}^B \psi_{n,s,b}\right)^2$ and resulting Wald-type confidence intervals (using $\psi_{n,s}^*$ or $\psi_{n,s}$), or form a percentile-based confidence interval with endpoints given by the 5th and 95th sample percentiles of $\{\psi_{n,s,1},\psi_{n,s,2},\ldots,\psi_{n,s,B}\}$.
\end{enumerate}
}

\revision{We consider again Scenario 1, where $p = 2$. For each scenario presented here, we generated 1000 random datasets of size $n \in \{100, 500, 1000, \dots, 4000\}$, and considered the importance of both $X_1$ and $X_2$. We assess performance in the same way as in the main manuscript, though we use the bootstrap-based intervals in place of those based on the influence function. We present the results of this experiment in Figures~\ref{fig:supp_boot_accuracy_1}--\ref{fig:supp_boot_auc_2}. The results for bias and variance are unchanged from the previous experiments. Encouragingly, both coverage and width for the bootstrap-based intervals are similar to the coverage and width of the IF-based intervals, though in the smaller sample size settings the bootstrap-based intervals are slightly narrower than the IF-based intervals.}

\begin{figure}
\centering
\includegraphics[width = 1\textwidth]{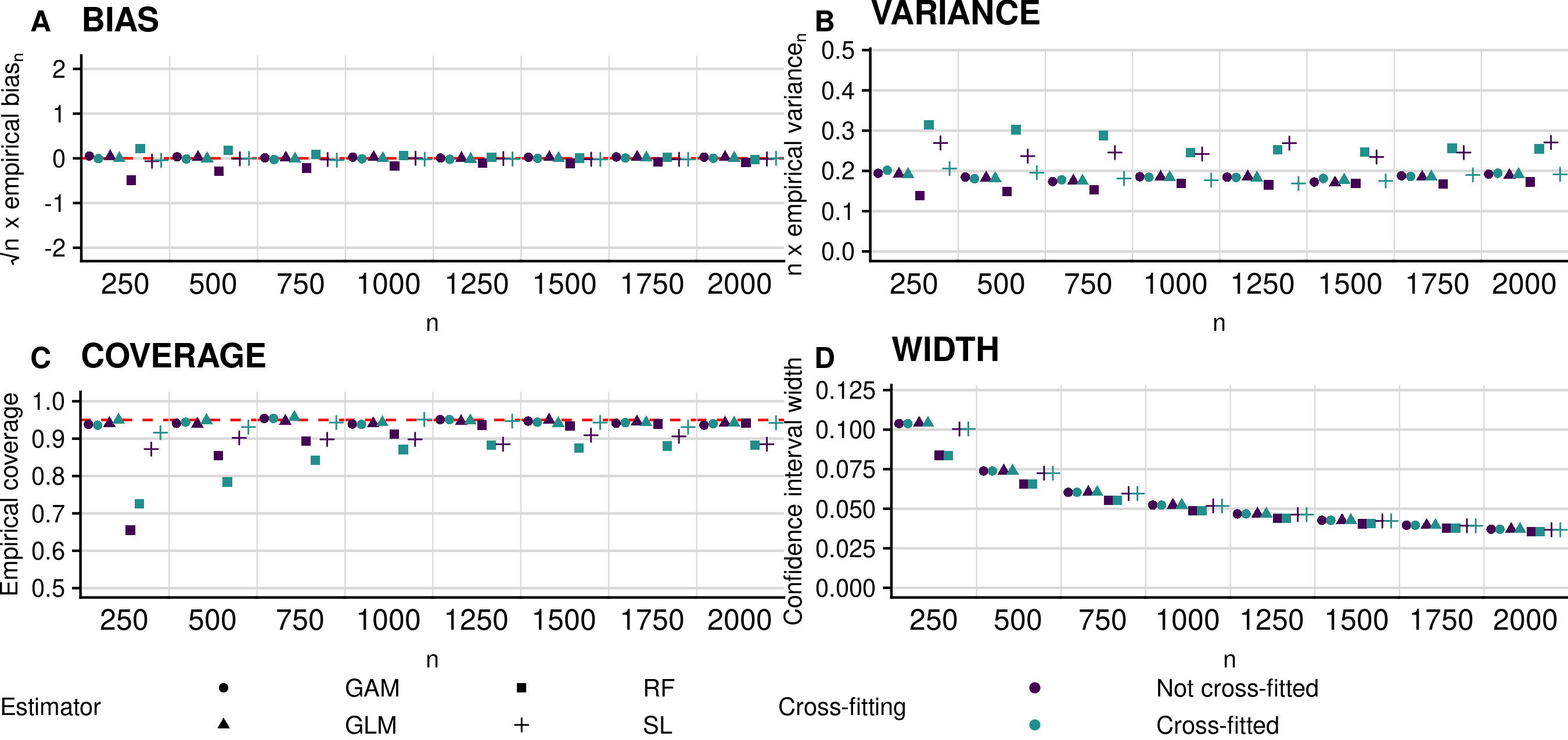}
\caption{Performance of plug-in estimators for estimating (non-zero) importance of $X_1$ in terms of accuracy under Scenario 1, using the bootstrap for interval estimation. Clockwise from top left: empirical bias of the proposed plug-in estimator scaled by $n^{1/2}$; empirical variance scaled by $n$; empirical coverage of nominal 95\% confidence intervals; and average width of these intervals. Circles, triangles, squares and plus symbols denote estimators based on the use of generalized additive models (GAMs), probit regression (GLM), random forests (RF), and the Super Learner (SL), respectively. Blue and green symbols denote non-cross-fitted and cross-fitted estimators, respectively. This figure appears in color in the electronic version of this article.}
\label{fig:supp_boot_accuracy_1}
\end{figure}

\begin{figure}
\centering
\includegraphics[width = 1\textwidth]{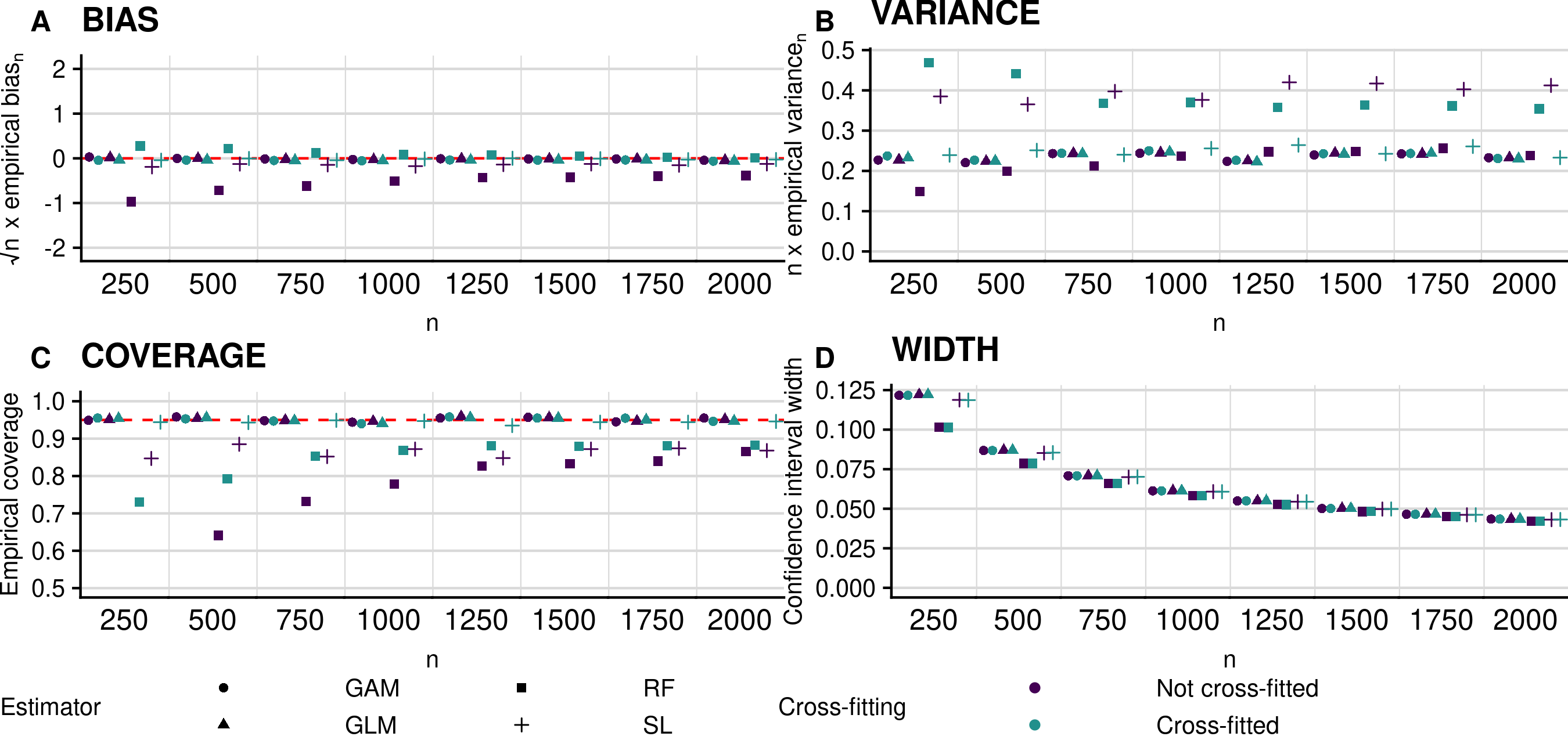}
\caption{Performance of plug-in estimators for estimating (non-zero) importance of $X_2$ in terms of accuracy under Scenario 1, using the bootstrap for interval estimation. Clockwise from top left: empirical bias of the proposed plug-in estimator scaled by $n^{1/2}$; empirical variance scaled by $n$; empirical coverage of nominal 95\% confidence intervals; and average width of these intervals. Circles, triangles, squares and plus symbols denote estimators based on the use of generalized additive models (GAMs), probit regression (GLM), random forests (RF), and the Super Learner (SL), respectively. Blue and green symbols denote non-cross-fitted and cross-fitted estimators, respectively. This figure appears in color in the electronic version of this article.}
\label{fig:supp_boot_accuracy_2}
\end{figure}

\begin{figure}
\centering
\includegraphics[width = 1\textwidth]{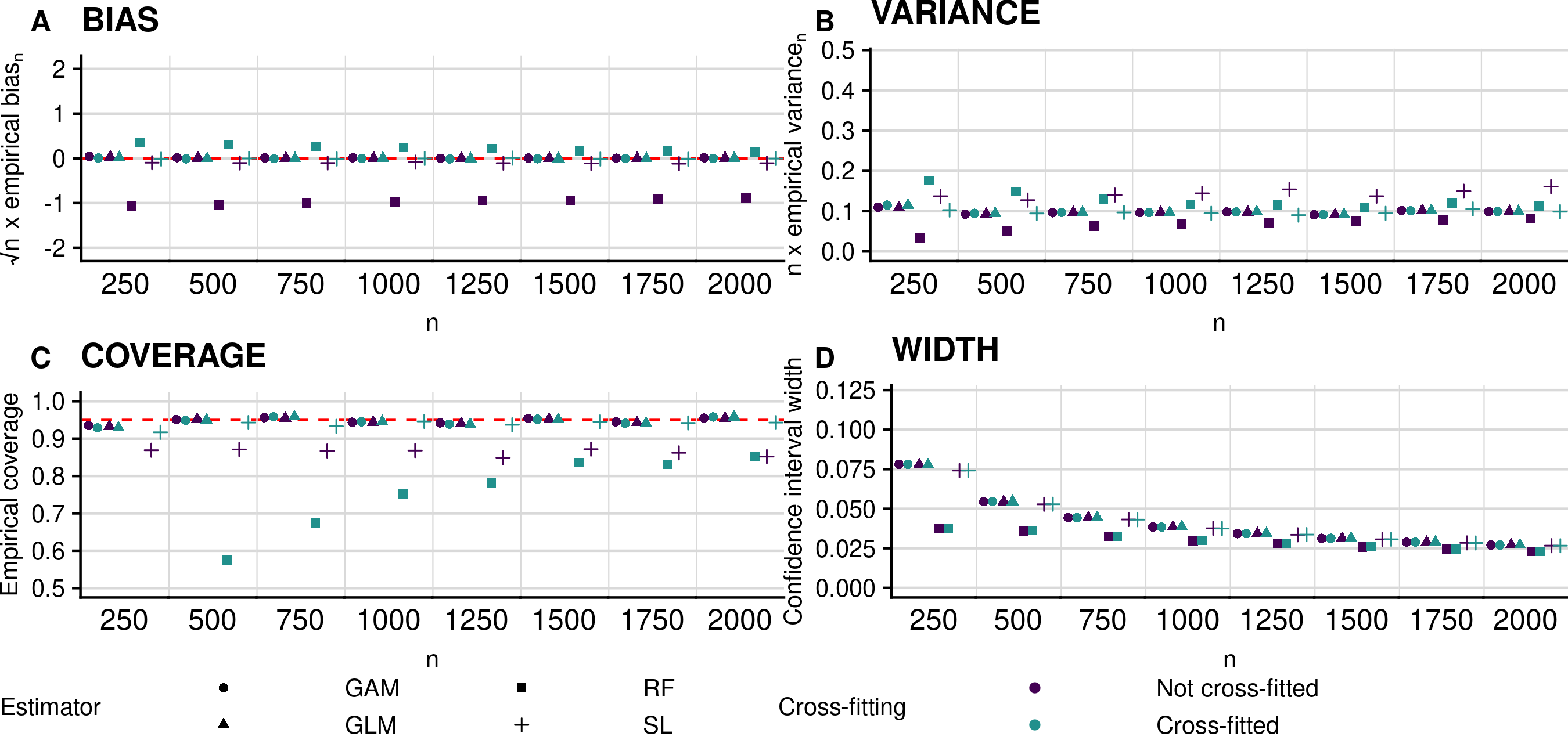}
\caption{Performance of plug-in estimators for estimating (non-zero) importance of $X_1$ in terms of AUC under Scenario 1, using the bootstrap for interval estimation. Clockwise from top left: empirical bias of the proposed plug-in estimator scaled by $n^{1/2}$; empirical variance scaled by $n$; empirical coverage of nominal 95\% confidence intervals; and average width of these intervals. Circles, triangles, squares and plus symbols denote estimators based on the use of generalized additive models (GAMs), probit regression (GLM), random forests (RF), and the Super Learner (SL), respectively. Blue and green symbols denote non-cross-fitted and cross-fitted estimators, respectively. Coverage of intervals based on the non-cross-fitted RF-based estimator never exceeds 0.5 and is as low as zero in some cases. This figure appears in color in the electronic version of this article.}
\label{fig:supp_boot_auc_1}
\end{figure}

\begin{figure}
\centering
\includegraphics[width = 1\textwidth]{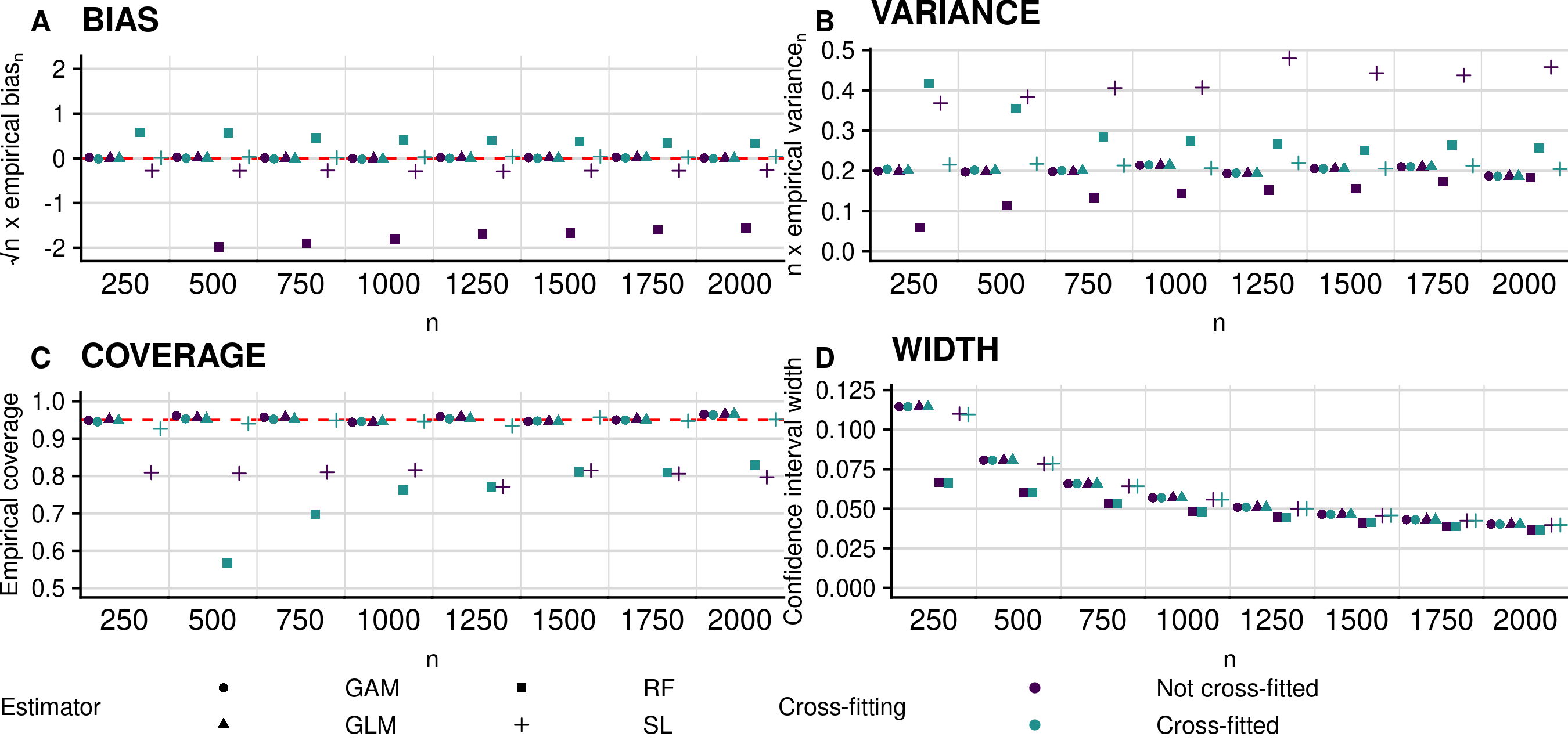}
\caption{Performance of plug-in estimators for estimating (non-zero) importance of $X_2$ in terms of AUC under Scenario 1, using the bootstrap for interval estimation. Clockwise from top left: empirical bias of the proposed plug-in estimator scaled by $n^{1/2}$; empirical variance scaled by $n$; empirical coverage of nominal 95\% confidence intervals; and average width of these intervals. Circles, triangles, squares and plus symbols denote estimators based on the use of generalized additive models (GAMs), probit regression (GLM), random forests (RF), and the Super Learner (SL), respectively. Blue and green symbols denote non-cross-fitted and cross-fitted estimators, respectively. In this experiment, the coverage of non-cross-fitted RF was never above 0.5, and was as low as zero. This figure appears in color in the electronic version of this article.}
\label{fig:supp_boot_auc_2}
\end{figure}

\subsection{Higher dimensions and correlated features}

\revision{We now consider two scenarios under increasing dimension, both with and without correlated features. Here, $p \in \{50, 100, 200\}$ and $\Sigma$ is either a $p \times p$ identity matrix (Scenario 3) or a $p \times p$ diagonal matrix with 1 on the diagonal and all off-diagonal elements equal to zero except $\Sigma_{1,3} = \Sigma_{3,1} = 0.7$ and $\Sigma_{2,4} = \Sigma_{4,2} = 0.2$ (Scenario 4). Thus, in Scenario 4, $X_3$ and $X_4$ are not directly important for predicting the outcome, but might be found to be important in isolation due to their correlation with the important features $X_1$ and $X_2$. In these experiments, we considered $n \in \{500, 3000\}$ for each $p$, and assessed the importance of each individual feature as well as the feature groups $(X_1, X_3)$ and $(X_2, X_4)$, again using both accuracy and AUC. We use cross-fitting to estimate the VIM value in all cases, and we use the Super Learner with candidate library consisting of boosted trees, random forests, and the lasso to estimate $f_0$ and $f_{0,s}$. We then compute the empirical bias scaled by $n^{1/2}$, the empirical variance scaled by $n$, the empirical coverage of nominal 95\% confidence intervals, and the proportion of tests rejected.}

\revision{We display the results under Scenario 3 in Figures~\ref{fig:supp_highdim_accuracy_uncorr} and \ref{fig:supp_highdim_auc_uncorr}. Here, we find that at the smaller sample size ($n = 500$), there is some excess bias for the features with non-null importance, and that this bias increases with increasing $p$; this is accompanied by a decrease in coverage. However, with a larger sample size ($n = 3000$), we recover similar performance to that observed in Section~\ref{sec:sims} of the main manuscript and the preceeding sections of this supplement. Type I error is controlled at the nominal level in all cases.}

\begin{figure}
\centering
\includegraphics[width = 1\textwidth]{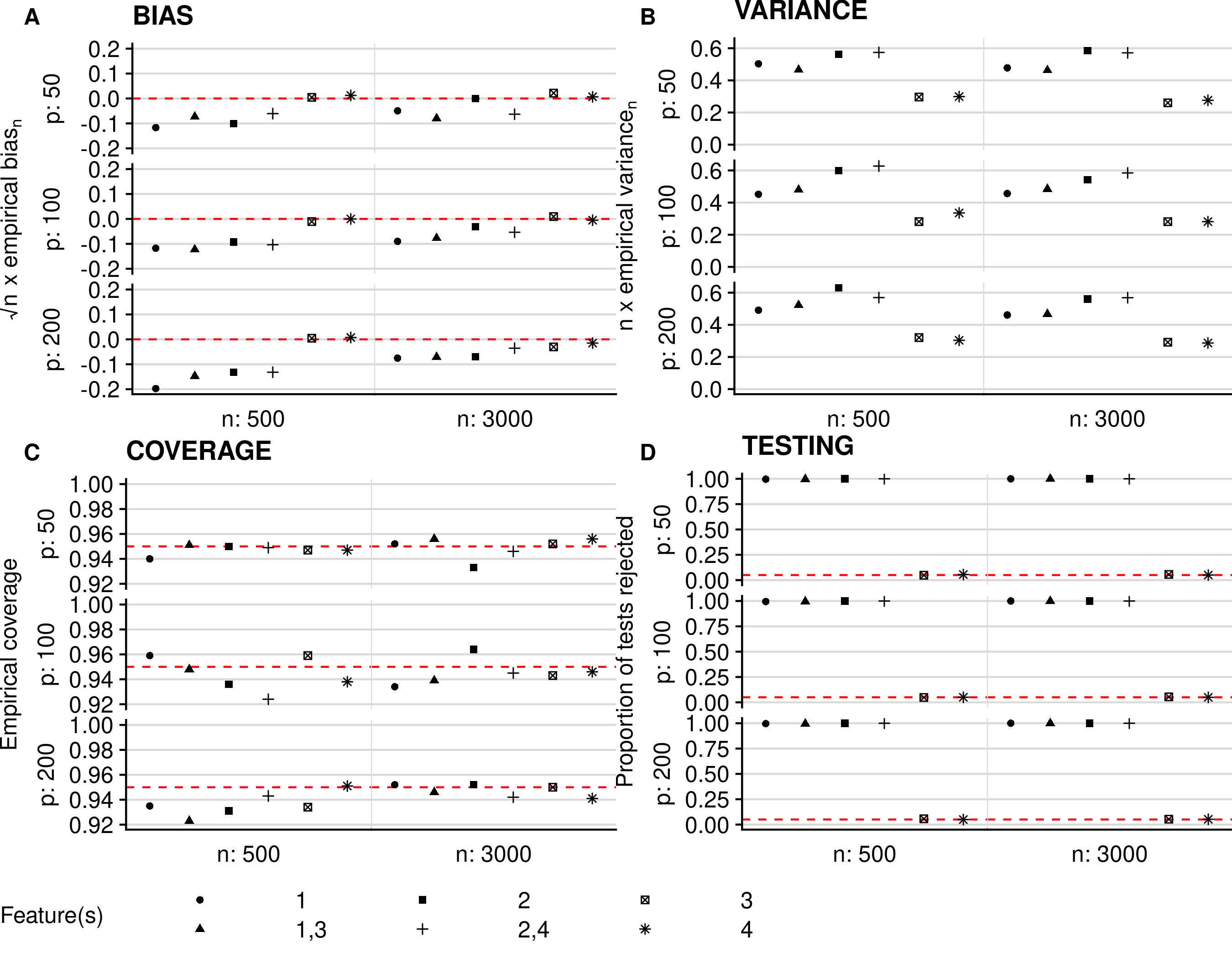}
\caption{Performance of plug-in estimators for estimating importance in terms of accuracy under Scenario 3 (all features are independent). Clockwise from top left: empirical bias for the proposed plug-in estimator scaled by $n^{1/2}$; empirical variance scaled by $n$; empirical coverage of nominal 95\% confidence intervals for the true importance; and empirical type I error of the proposed hypothesis test. The different symbols denote the feature(s) of interest.}
\label{fig:supp_highdim_accuracy_uncorr}
\end{figure}

\begin{figure}
\centering
\includegraphics[width = 1\textwidth]{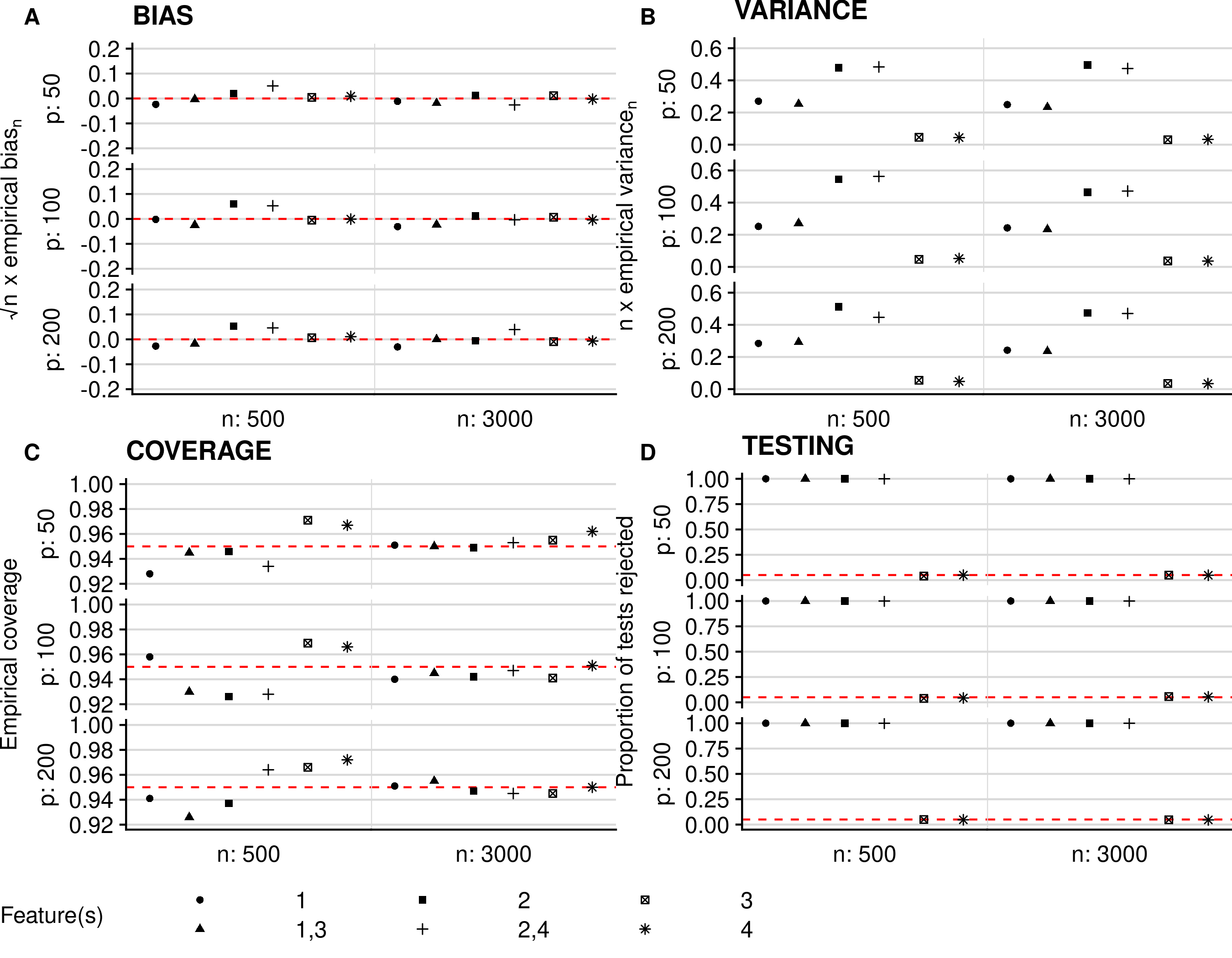}
\caption{Performance of plug-in estimators for estimating importance in terms of AUC under Scenario 3 (all features are independent). Clockwise from top left: empirical bias for the proposed plug-in estimator scaled by $n^{1/2}$; empirical variance scaled by $n$; empirical coverage of nominal 95\% confidence intervals for the true importance; and empirical type I error of the proposed hypothesis test. The different symbols denote the feature(s) of interest.}
\label{fig:supp_highdim_auc_uncorr}
\end{figure}

\revision{We display the results under Scenario 4 in Figures~\ref{fig:supp_highdim_accuracy_corr} and \ref{fig:supp_highdim_auc_corr}. We find similar results overall to those from Scenario 3. In smaller samples, it appears to be advantageous to consider groups of correlated features rather than the features alone; this is particularly striking in Figure~\ref{fig:supp_highdim_auc_corr}. As the sample size grows, the difference in performance diminishes.}

\begin{figure}
\centering
\includegraphics[width = 1\textwidth]{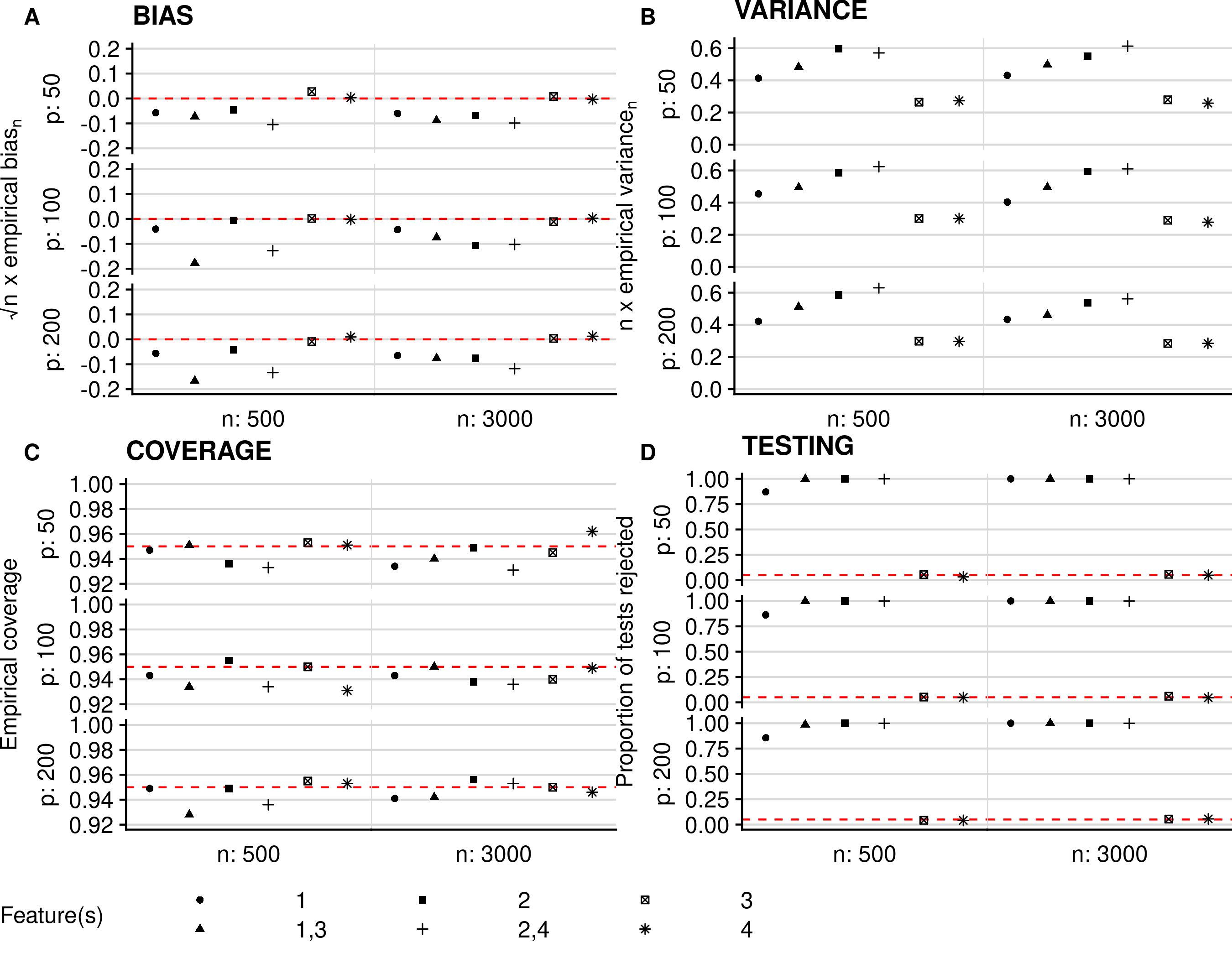}
\caption{Performance of plug-in estimators for estimating importance in terms of accuracy under Scenario 4 (some features are correlated). Clockwise from top left: empirical bias for the proposed plug-in estimator scaled by $n^{1/2}$; empirical variance scaled by $n$; empirical coverage of nominal 95\% confidence intervals for the true importance; and empirical type I error of the proposed hypothesis test. The different symbols denote the feature(s) of interest.}
\label{fig:supp_highdim_accuracy_corr}
\end{figure}

\begin{figure}
\centering
\includegraphics[width = 1\textwidth]{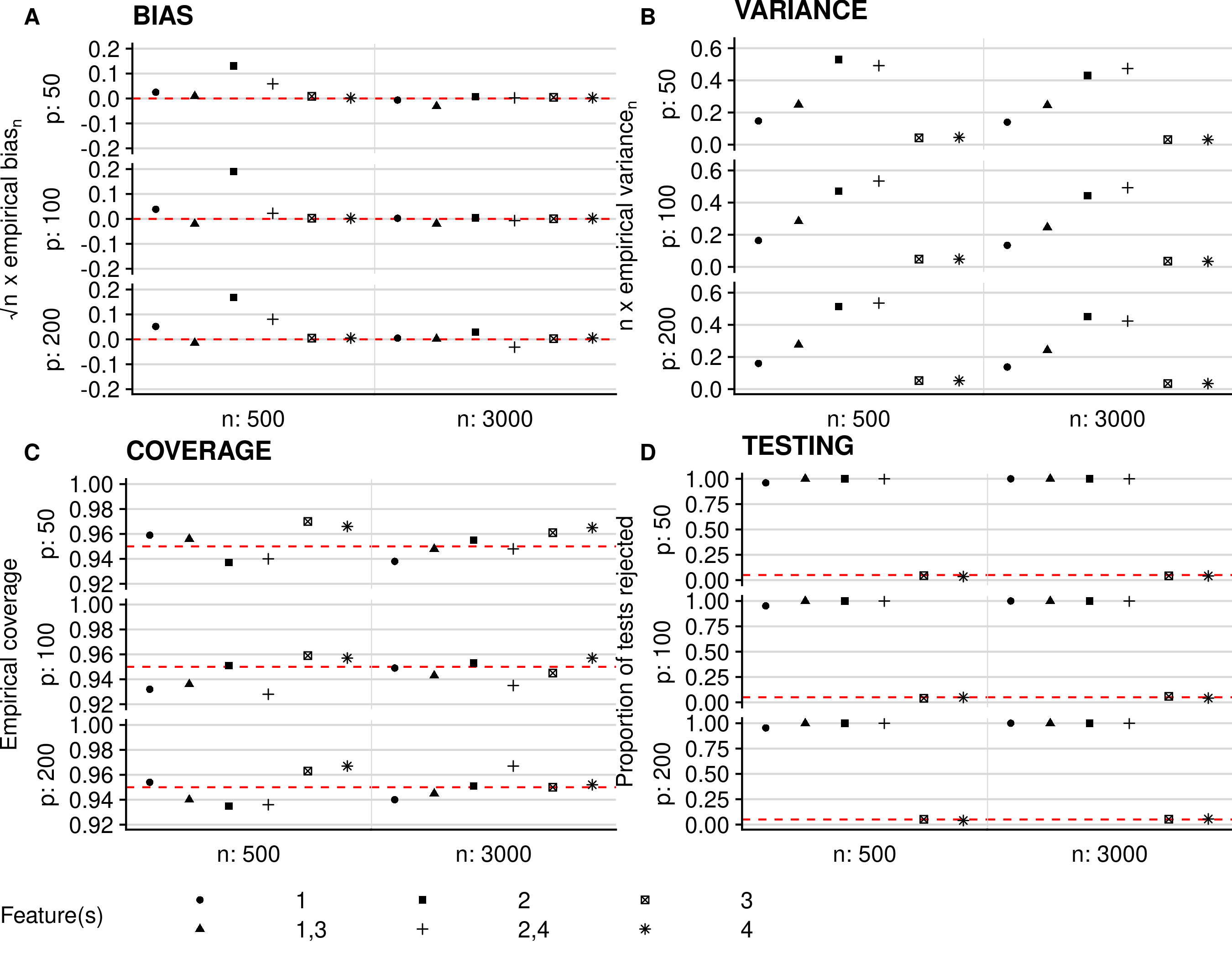}
\caption{Performance of plug-in estimators for estimating importance in terms of AUC under Scenario 4 (some features are correlated). Clockwise from top left: empirical bias for the proposed plug-in estimator scaled by $n^{1/2}$; empirical variance scaled by $n$; empirical coverage of nominal 95\% confidence intervals for the true importance; and empirical type I error of the proposed hypothesis test. The different symbols denote the feature(s) of interest.}
\label{fig:supp_highdim_auc_corr}
\end{figure}

\revision{Overall, the statistical performance of our procedure appear to be impacted more strongly by noise covariates in small samples than in large samples, regardless of the level of correlation among covariates. It is possible that this performance could be improved in small samples by including more aggressive sparsity-inducing algorithms in our ensemble. Indeed, the performance of our estimator of each VIM value depends on  the rate at which the nuisance functions can be estimated, and this rate certainly slows down as the number of covariates grows, unless we can leverage stronger structure. We note that, while perhaps minimally impacting the statistical performance of our procedure, correlated features nevertheless render the interpretation of individual-variable importance more challenging: the population-level importance value itself changes in the presence of correlation. This difficulty can be partially mitigated by assessing group variable importance instead; however, this requires groups to either be known a priori (as in Section~\ref{sec:data} of the main manuscript) or estimated, and in this latter case, further work must be done to ensure that the desired inferential properties (e.g., correct coverage) are preserved.}

\section{Additional details for the study of an antibody against HIV-1}\label{sec:additional_data_analysis_results}

\subsection{Harmonized analysis with \citet{magaret2019}}

\revision{In Figure~\ref{fig:cens_amp_results}, we display the results of an analysis harmonized to use the same outcome as in \cite{magaret2019}. This sensitivity outcome is the indicator of whether or not the IC$_{50}$ value was right-censored. Viruses with right-censored IC$_{50}$ values} are thought to be resistant to VRC01, while viruses with non-censored IC$_{50}$ values may instead be more sensitive to VRC01. In this case, we consider the \textit{conditional} importance of each group of features relative to the remaining features. Overall, these results are largely in line with both \cite{magaret2019} and with the results presented in the main manuscript. However, we see here that only the VRC01 binding footprint has p-value less than 0.0038 (denoted by stars in Figure~\ref{fig:cens_amp_results}; this value results from a Bonferroni correction from testing 13 groups and an initial level of 0.05), and only for the AUC measure. The exact p-value is given by $6.1\times 10^{-4}$.

\begin{figure}
\centering
\includegraphics[width=1\textwidth]{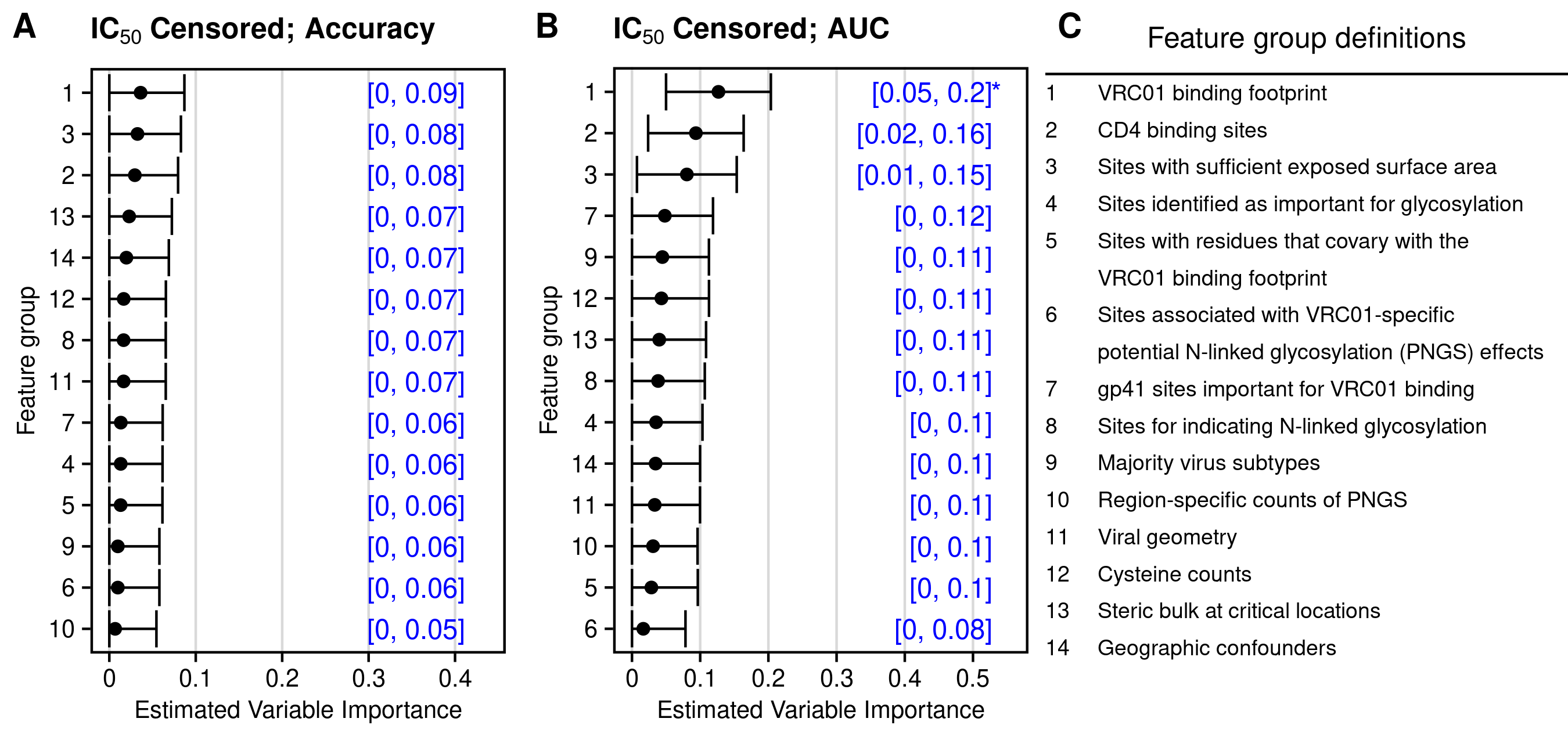}
\caption{Variable importance measured by accuracy (panel A) and AUC (panel B) for the groups defined in panel C. Stars denote importance deemed statistically significantly different from zero at the 0.0038 (0.05 / 13) level.}
\label{fig:cens_amp_results}
\end{figure}

\subsection{Library of candidate learning algorithms}\label{sec:learner_lib}

In this section, we describe the library of candidate learning algorithms  used in our analysis replicating the results of \citet{magaret2019}. We used a wide array of flexible machine learning-based algorithms in the hope that this large library  would yield a cross-validated algorithm with good predictive performance. The particular machine learning techniques included were: the lasso with logit link function (implemented in the \texttt{glmnet} R package), random forests (implemented in the \texttt{ranger} R package), and gradient boosted decision trees (implemented in the \texttt{xgboost} R package), each with a variety of choices for the tuning parameters. In Table~\ref{tab:sl_lib}, we provide a description of each candidate learning algorithm in our library. Our final estimator is the convex combination of these algorithms chosen to minimize the ten-fold cross-validated negative log likelihood. In all cases, we adjusted for geographic region as a potential confounding variable.

\input{supplement_sl_lib_table}

\subsection{Super Learner performance}

We now describe the empirical performance of the Super Learner in this application for both the outcome considered in the main manuscript (IC$_{50} < 1$) and the IC$_{50}$ censored outcome described above. In Table~\ref{tab:sl_weights}, we show the coefficients of each candidate learner in the final Super Learner ensemble for each outcome. The rows of this table are each of the ten cross-validation folds broken down by outcome, while the columns are the individual learners. Here, we see that for the IC$_{50}$ censored outcome, the most commonly chosen algorithms in the final ensemble were boosted trees with maximum depth of 2 or 4, random forests with a large number of features chosen at each split, and the elastic net with various values of $\alpha$. For the IC$_{50} < 1$ outcome, the most commonly chosen algorithms were again boosted trees with maximum depth of 2, 4, or 6, random forests with a medium and large number of features chosen at each split; the elastic net was often not chosen by the Super Learner.

\input{supplement_sl_weight_table}

In Figure~\ref{fig:sl_perf}, we display the cross-validated AUC and 95\% confidence intervals (obtained on the logit scale and then inverted; thus, the intervals may not be symmetric about the point estimate of AUC) for both outcomes and each of the candidate learning algorithms in the Super Learner, along with the Super Learner ensemble algorithm and the classical cross-validated selector (the ``discrete Super Learner''). We used the R package \texttt{cvAUC} to compute these point and interval estimates. Similarly to \citet{magaret2019}, we see that, of all the individual algorithms, random forests have the best performance in this application for both outcomes, followed by the lasso and boosted trees (for the IC$_{50}$ censored outcome) and the reverse for the IC$_{50} < 1$ outcome. Additionally, we estimate the cross-validated AUC of the overall Super Learner to be 0.90 for the IC$_{50}$ censored outcome, with a 95\% confidence interval of (0.87, 0.94). For the IC$_{50} < 1$ outcome, we estimate the cross-validated AUC of the overall Super Learner to be 0.83 (0.80, 0.86). \citet{magaret2019} performed an analysis for IC$_{50}$ censored separately on two independent splits of these data, and obtained cross-validated AUCs of 0.86 (0.81, 0.92) and 0.87 (0.81, 0.93) on these two subsets.

\begin{figure}
\centering
\includegraphics[width = 0.6\textwidth]{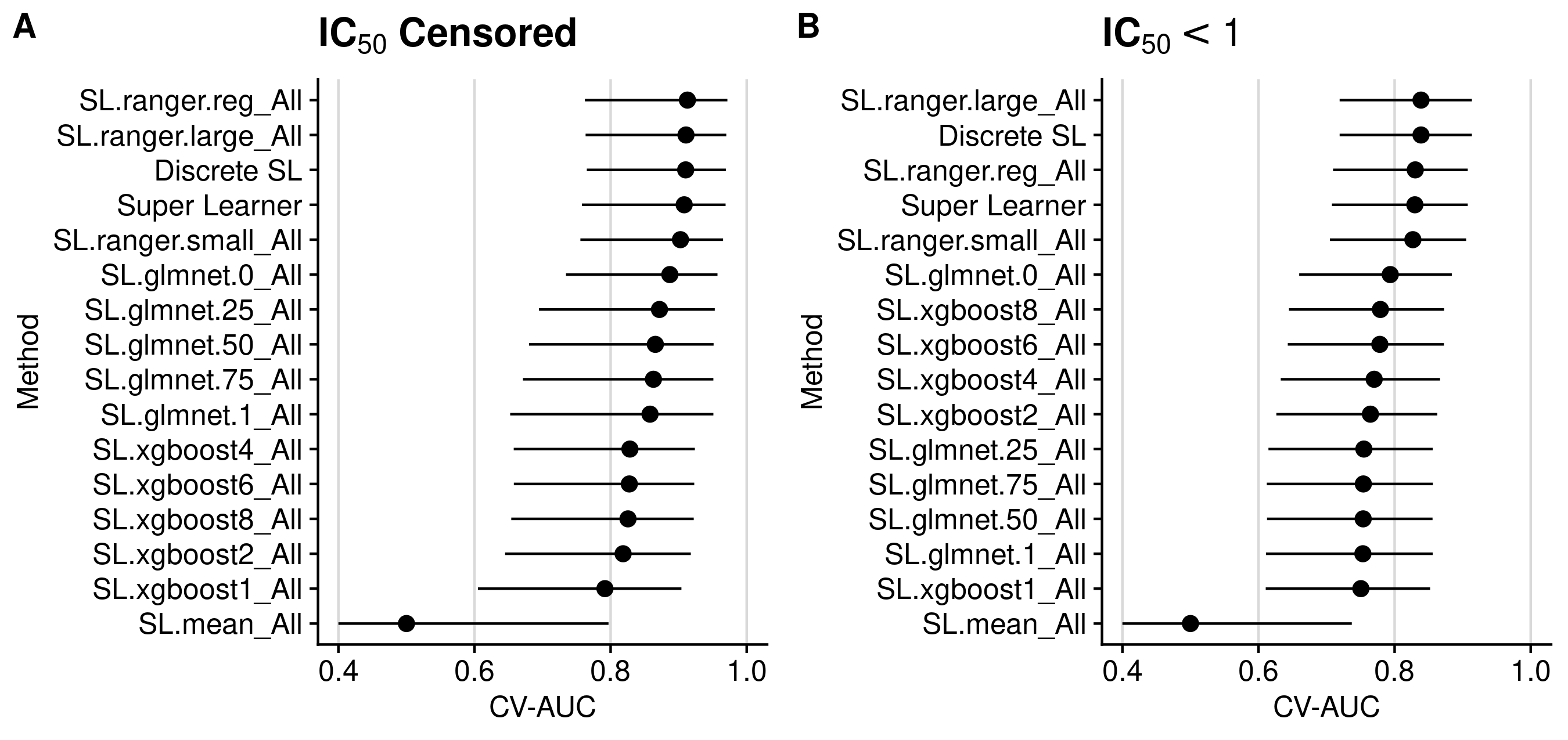}
\caption{Point estimates of cross-validated AUC with 95\% confidence intervals for each candidate learning algorithm in the Super Learner for each outcome.}
\label{fig:sl_perf}
\end{figure}

In Figure~\ref{fig:sl_roc}, we display cross-validated ROC curves for the Super Learner, discrete Super Learner, and the top-performing individual algorithm. These ROC curves are similar to those presented in \citet{magaret2019} --- in both analyses, we see a large cross-validated true positive rate for each chosen cross-validated false positive rate. These results suggest that for both outcomes, our predictor is well-calibrated for discriminating between the outcome classes.

\begin{figure}
\centering
\includegraphics[width = 0.8\textwidth]{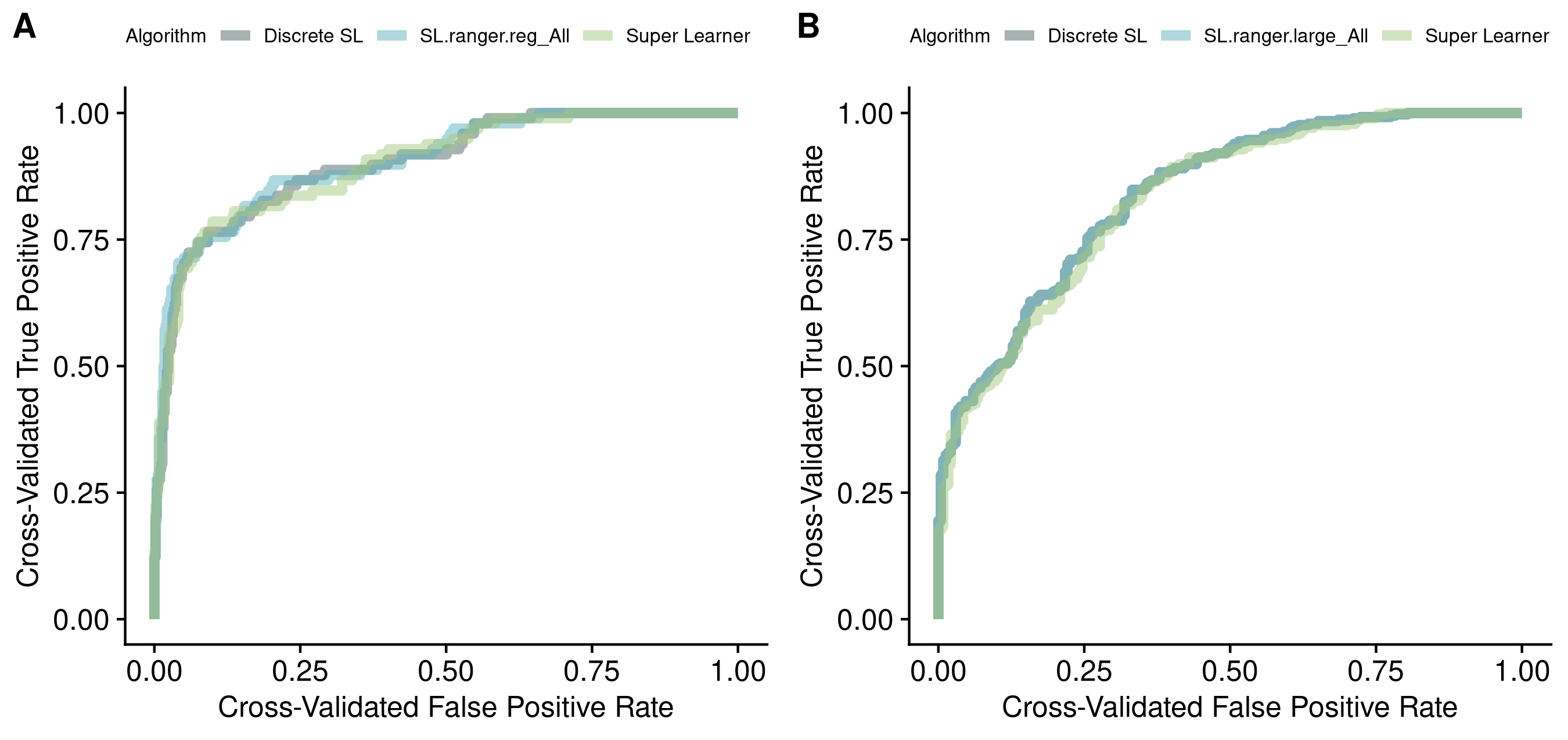}
\caption{Cross-validated ROC curves for each outcome for the Super Learner (light green), discrete Super Learner (gray), and top-performing individual algorithm (random forests). IC$_{50}$ censored is displayed in panel A, while IC$_{50} < 1$ is displayed in panel B.}
\label{fig:sl_roc}
\end{figure}

\end{document}

%% file: supplement_sl_lib_table.tex
\begin{table}

\caption{\label{tab:sl_lib}Library of canididate learners for the Super Learner with descriptions.}
\centering
\resizebox{\linewidth}{!}{
\begin{tabular}[t]{l|l}
\hline
Function name & Description\\
\hline
SL.mean & intercept only regression\\
\hline
SL.xgboost1 & boosted regression trees with maximum depth of 1\\
\hline
SL.xgboost2 & boosted regression trees with maximum depth of 2\\
\hline
SL.xgboost4 & boosted regression trees with maximum depth of 4\\
\hline
SL.xgboost6 & boosted regression trees with maximum depth of 6\\
\hline
SL.xgboost8 & boosted regression trees with maximum depth of 8\\
\hline
SL.ranger.small & random forest with mtry equal to one-half times square root of number of predictors\\
\hline
SL.ranger.reg & random forest with mtry equal to square root of number of predictors\\
\hline
SL.ranger.large & random forest with mtry equal to two times square root of number of predictors\\
\hline
SL.glmnet.0 & GLMNET with lambda selected by 5-fold CV and alpha equal to 0\\
\hline
SL.glmnet.25 & GLMNET with lambda selected by 5-fold CV and alpha equal to 0.25\\
\hline
SL.glmnet.50 & GLMNET with lambda selected by 5-fold CV and alpha equal to 0.5\\
\hline
SL.glmnet.75 & GLMNET with lambda selected by 5-fold CV and alpha equal to 0.75\\
\hline
SL.glmnet.1 & GLMNET with lambda selected by CV and alpha equal to 1\\
\hline
\end{tabular}}
\end{table}

%% file: supplement_sl_weight_table.tex
\begin{table}

\caption{\label{tab:sl_weights}Table of Super Learner weights for each outcome, candidate learner and cross-validation fold. We have removed `SL.' from the name of each learner.}
\centering
\resizebox{\linewidth}{!}{
\begin{tabular}[t]{r|r|r|r|r|r|r|r|r|r|r|r|r|r|r}
\hline
mean & xgboost1 & xgboost2 & xgboost4 & xgboost6 & xgboost8 & ranger.small & ranger.reg & ranger.large & glmnet.0 & glmnet.25 & glmnet.50 & glmnet.75 & glmnet.1 & fold\\
\hline
\multicolumn{15}{l}{\textbf{IC$_{50}$ censored}}\\
\hline
\hspace{1em}0 & 0 & 0.05 & 0.00 & 0.00 & 0.00 & 0 & 0.00 & 0.74 & 0 & 0.00 & 0.21 & 0.00 & 0.00 & 1\\
\hline
\hspace{1em}0 & 0 & 0.08 & 0.00 & 0.00 & 0.00 & 0 & 0.00 & 0.62 & 0 & 0.00 & 0.30 & 0.00 & 0.00 & 2\\
\hline
\hspace{1em}0 & 0 & 0.01 & 0.00 & 0.00 & 0.00 & 0 & 0.00 & 0.50 & 0 & 0.48 & 0.00 & 0.00 & 0.00 & 3\\
\hline
\hspace{1em}0 & 0 & 0.00 & 0.08 & 0.00 & 0.00 & 0 & 0.00 & 0.51 & 0 & 0.40 & 0.00 & 0.00 & 0.00 & 4\\
\hline
\hspace{1em}0 & 0 & 0.00 & 0.04 & 0.00 & 0.00 & 0 & 0.00 & 0.58 & 0 & 0.00 & 0.38 & 0.00 & 0.00 & 5\\
\hline
\hspace{1em}0 & 0 & 0.11 & 0.00 & 0.00 & 0.00 & 0 & 0.00 & 0.62 & 0 & 0.27 & 0.00 & 0.00 & 0.00 & 6\\
\hline
\hspace{1em}0 & 0 & 0.11 & 0.00 & 0.00 & 0.00 & 0 & 0.00 & 0.51 & 0 & 0.07 & 0.00 & 0.31 & 0.00 & 7\\
\hline
\hspace{1em}0 & 0 & 0.05 & 0.00 & 0.00 & 0.00 & 0 & 0.00 & 0.74 & 0 & 0.00 & 0.14 & 0.00 & 0.08 & 8\\
\hline
\hspace{1em}0 & 0 & 0.01 & 0.01 & 0.00 & 0.00 & 0 & 0.00 & 0.62 & 0 & 0.23 & 0.00 & 0.12 & 0.00 & 9\\
\hline
\hspace{1em}0 & 0 & 0.07 & 0.00 & 0.00 & 0.00 & 0 & 0.00 & 0.36 & 0 & 0.27 & 0.00 & 0.00 & 0.31 & 10\\
\hline
\multicolumn{15}{l}{\textbf{IC$_{50} < 1$}}\\
\hline
\hspace{1em}0 & 0 & 0.00 & 0.13 & 0.00 & 0.00 & 0 & 0.00 & 0.87 & 0 & 0.00 & 0.00 & 0.00 & 0.00 & 1\\
\hline
\hspace{1em}0 & 0 & 0.00 & 0.18 & 0.00 & 0.00 & 0 & 0.00 & 0.82 & 0 & 0.00 & 0.00 & 0.00 & 0.00 & 2\\
\hline
\hspace{1em}0 & 0 & 0.00 & 0.00 & 0.16 & 0.00 & 0 & 0.00 & 0.84 & 0 & 0.00 & 0.00 & 0.00 & 0.00 & 3\\
\hline
\hspace{1em}0 & 0 & 0.00 & 0.06 & 0.06 & 0.00 & 0 & 0.00 & 0.89 & 0 & 0.00 & 0.00 & 0.00 & 0.00 & 4\\
\hline
\hspace{1em}0 & 0 & 0.02 & 0.12 & 0.05 & 0.00 & 0 & 0.00 & 0.82 & 0 & 0.00 & 0.00 & 0.00 & 0.00 & 5\\
\hline
\hspace{1em}0 & 0 & 0.11 & 0.00 & 0.00 & 0.03 & 0 & 0.00 & 0.86 & 0 & 0.00 & 0.00 & 0.00 & 0.00 & 6\\
\hline
\hspace{1em}0 & 0 & 0.05 & 0.00 & 0.00 & 0.00 & 0 & 0.11 & 0.84 & 0 & 0.00 & 0.00 & 0.00 & 0.00 & 7\\
\hline
\hspace{1em}0 & 0 & 0.00 & 0.00 & 0.07 & 0.03 & 0 & 0.00 & 0.90 & 0 & 0.00 & 0.00 & 0.00 & 0.00 & 8\\
\hline
\hspace{1em}0 & 0 & 0.12 & 0.00 & 0.00 & 0.00 & 0 & 0.39 & 0.41 & 0 & 0.00 & 0.00 & 0.07 & 0.00 & 9\\
\hline
\hspace{1em}0 & 0 & 0.00 & 0.00 & 0.07 & 0.00 & 0 & 0.00 & 0.93 & 0 & 0.00 & 0.00 & 0.00 & 0.00 & 10\\
\hline
\end{tabular}}
\end{table}